\documentclass[10pt,a4paper,reqno]{amsart}
\usepackage{amsmath,amssymb,amsthm}
\usepackage{xspace}
\usepackage[usenames]{color}

\usepackage[numbers]{natbib}
\usepackage{amsfonts, amsmath, wasysym}
\usepackage{amssymb, amsthm}
\usepackage{mathrsfs}
\usepackage{verbatim}
\newfont{\ffont}{cmr10}

\allowdisplaybreaks

\newcommand{\poly}{\mathop{poly}}

\newcommand{\cA}{\mathcal{A}}

\newcommand{\cF}{\mathcal{F}}

\newcommand{\cK}{\mathcal{K}}

\newcommand{\cM}{\mathcal{M}}
\newcommand{\cN}{\mathcal{N}}
\newcommand{\cC}{\mathcal{C}}

\newcommand{\ot}{\otimes}

\newcommand{\Tr}{\mathrm{Tr}}
\newcommand{\al}{\alpha}

\newcommand{\ga}{\gamma}
\newcommand{\del}{\delta}
\newcommand{\ka}{\kappa}
\newcommand{\Om}{\Omega}
\newcommand{\si}{\sigma}

\newcommand{\DelO}{d}

\newcommand{\bP}{\mathbb{P}}

\newcommand{\la}{\lambda}
\newcommand{\ti}{\times}

\newcommand{\btnorm}[1]{\Big\vert\kern-0.25ex\Big\vert\kern-0.25ex\Big\vert #1 
    \Big\vert\kern-0.25ex\Big\vert\kern-0.25ex\Big\vert}
\newcommand\tnorm[1]{\vert\kern-0.25ex\vert\kern-0.25ex\vert #1 
    \vert\kern-0.25ex\vert\kern-0.25ex\vert}

\DeclareMathOperator*{\E}{\mathbb{E}}
\let\Pr\relax
\DeclareMathOperator*{\Pr}{\mathbb{P}}

\newcommand{\ceil}[1]{\left\lceil #1 \right\rceil}

\newcommand{\eps}{\varepsilon}
\newcommand{\inprod}[1]{\left\langle #1 \right\rangle}
\newcommand{\R}{\mathbb{R}}
\newcommand{\eqdef}{\mathbin{\stackrel{\rm def}{=}}}
\newcommand{\lts}{\mathbin{\substack{<\\ *}}\xspace}
\newcommand{\gts}{\mathbin{\substack{>\\ *}}\xspace}
\newcommand{\vertiii}[1]{{\left\| #1 \right\|_{\delta}}}

\newtheorem{theorem}{Theorem}
\newtheorem{remark}[theorem]{Remark}
\newtheorem{claim}[theorem]{Claim}
\newtheorem{corollary}[theorem]{Corollary}
\newtheorem{lemma}[theorem]{Lemma}
\newtheorem{question}[theorem]{Question}
\numberwithin{equation}{section}
\newtheorem{example}[theorem]{Example}

\newcommand{\EquationName}[1]{\label{eq:#1}}
\newcommand{\ClaimName}[1]{\label{clm:#1}}
\newcommand{\LemmaName}[1]{\label{lem:#1}}
\newcommand{\CorollaryName}[1]{\label{cor:#1}}
\newcommand{\SectionName}[1]{\label{sec:#1}}
\newcommand{\TheoremName}[1]{\label{thm:#1}}
\newcommand{\RemarkName}[1]{\label{rem:#1}}
\newcommand{\FigureName}[1]{\label{fig:#1}}
\newcommand{\QuestionName}[1]{\label{que:#1}}

\newcommand{\Equation}[1]{Eq.\:\eqref{eq:#1}}
\newcommand{\Claim}[1]{Claim~\ref{clm:#1}}
\newcommand{\Lemma}[1]{Lemma~\ref{lem:#1}}
\newcommand{\Corollary}[1]{Corollary~\ref{cor:#1}}
\newcommand{\Section}[1]{Section~\ref{sec:#1}}
\newcommand{\Theorem}[1]{Theorem~\ref{thm:#1}}
\newcommand{\Remark}[1]{Remark~\ref{rem:#1}}
\newcommand{\Figure}[1]{Figure~\ref{fig:#1}}
\newcommand{\Question}[1]{Question~\ref{que:#1}}

\newcommand{\Eqsub}[1]{\eqref{eq:#1}}

\begin{document} 

\author
{Jean~Bourgain}
\address
{Institute for Advanced Study, Princeton, NJ 08540}
\email
{bourgain@math.ias.edu}
\thanks{J.B.\ partially supported by NSF grant DMS-1301619.}

\author{Sjoerd~Dirksen}
\address{RWTH Aachen University, 52062 Aachen, Germany}
\email{dirksen@mathc.rwth-aachen.de}
\thanks{S.D.\ partially supported by SFB grant 1060 of the Deutsche Forschungsgemeinschaft (DFG)}

\author{Jelani Nelson}
\address{Harvard University, Cambridge, MA 02138}
\email{minilek@seas.harvard.edu}
\thanks{J.N.\ supported by NSF grant IIS-1447471 and CAREER award CCF-1350670, ONR grant N00014-14-1-0632, and a Google Faculty Research Award. Part of this work done while supported by NSF grants CCF-0832797 and DMS-1128155.}

\title[Toward a unified theory of sparse dimensionality reduction]{Toward a unified theory of sparse dimensionality reduction in Euclidean space}

\maketitle

\vspace{-.4in}\begin{abstract}
Let $\Phi\in\R^{m\times n}$ be a sparse Johnson-Lindenstrauss transform \cite{KN14} with $s$ non-zeroes per column. For a subset $T$ of the unit sphere, $\eps\in(0,1/2)$ given, we study settings for $m,s$ required to ensure
$$
\mathop{\mathbb{E}}_\Phi \sup_{x\in T} \left|\|\Phi x\|_2^2 - 1 \right| < \eps ,
$$
i.e.\ so that $\Phi$ preserves the norm of every $x\in T$ simultaneously and multiplicatively up to $1+\eps$. We introduce a new complexity parameter, which depends on the geometry of $T$, and show that it suffices to choose $s$ and $m$ such that this parameter is small. Our result is a sparse analog of Gordon's theorem, which was concerned with a dense $\Phi$ having i.i.d.\ Gaussian entries. We qualitatively unify several results related to the Johnson-Lindenstrauss lemma, subspace embeddings, and Fourier-based restricted isometries. Our work also implies new results in using the sparse Johnson-Lindenstrauss transform in numerical linear algebra, classical and model-based compressed sensing, manifold learning, and constrained least squares problems such as the Lasso.
\end{abstract}

\tableofcontents

\section{Introduction}\SectionName{intro}
Dimensionality reduction is a ubiquitous tool across a wide array of disciplines: machine learning \cite{WDLSA09}, high-dimensional computational geometry \cite{Indyk01}, privacy \cite{BlockiBDS12}, compressed sensing \cite{CT05}, spectral graph theory \cite{SpielmanS11}, interior point methods for linear programming \cite{LeeS13a}, numerical linear algebra \cite{Sarlos06}, computational learning theory \cite{BalcanB05,BalcanBV06}, manifold learning \cite{BHW07,Clarkson08}, motif-finding in computational biology \cite{BuhlerT02}, astronomy \cite{ContrerasM12}, and several others. Across all these disciplines one is typically faced with data that is not only massive, but each data item itself is represented as a very high-dimensional vector. For example, when learning spam classifiers a data point is an email, and it is represented as a high-dimensional vector indexed by dictionary words \cite{WDLSA09}. In astronomy a data point could be a star, represented as a vector of light intensities measured over 
various 
points sampled in time \cite{KovacsZM02,VanderburgJ14}. Dimensionality reduction techniques in such applications provide the following benefits:
\begin{itemize}
\item Smaller storage consumption.
\item Speedup during data analysis.
\item Cheaper signal acquisition.
\item Cheaper transmission of data across computing clusters.
\end{itemize}

The technical guarantees required from a dimensionality reduction routine are application-specific, but typically such methods must reduce dimension while still preserving point geometry, e.g.\ inter-point distances and angles. That is, one has some point set $X\subset\R^n$ with $n$ very large, and we would like a dimensionality-reducing map $f:X\rightarrow\R^m$, $m\ll n$, such that
\begin{equation}
\forall x,y\in X,\ (1-\eps)\|x-y\| \le \|f(x)-f(y)\| \le (1+\eps)\|x-y\| \EquationName{jlf}
\end{equation}
for some norm $\|\cdot\|$. Note also that for unit vectors $x,y$, $\cos(\angle (x,y)) = (1/2)(\|x\|_2^2 + \|y\|_2^2 - \|x-y\|_2^2)$, and thus $f$ also preserves angles with additive error if it preserves Euclidean norms of points in $X\cup(X-X)$.

A powerful tool for achieving \Equation{jlf}, used in nearly all the applications cited above, is the {\em Johnson-Lindenstrauss (JL) lemma} \cite{JL84}.

\begin{theorem}[JL lemma]
For any subset $X$ of Euclidean space and $0<\eps<1/2$, there exists $f:X\rightarrow\ell_2^m$ with $m = O(\eps^{-2}\log |X|)$ providing \Equation{jlf} for $\|\cdot\| = \|\cdot\|_2$.
\end{theorem}

This bound on $m$ is nearly tight: for any $n\ge 1$ Alon exhibited a point set $X\subset\ell_2^n$, $|X|=n+1$, such that any such JL map $f$ must have $m = \Omega(\eps^{-2}(\log n)/\log(1/\eps))$ \cite{Alon03}. In fact, all known proofs of the JL lemma provide linear $f$, and the JL lemma is tight up to a constant factor in $m$ when $f$ must be linear \cite{LarsenN14}. Unfortunately, for actual applications such {\em worst-case} understanding is unsatisfying. Rather we could ask: if given a distortion parameter $\eps$ and point set $X$ as input (or a succinct description of it if $X$ is large or even infinite, as in some applications), what is the best target dimension $m = m(X,\eps)$ such that a JL map exists for $X$ with this particular $\eps$? That is, in practice we are more interested in moving beyond worst case analysis and being as efficient as possible {\em for our particular data} $X$.

Unfortunately the previous question seems fairly difficult. For the related question of computing the optimal distortion for embedding $X$ into a line (i.e.\ $m=1$), it is computationally hard to approximate the optimal distortion even up to a multiplicative factor polynomial in $|X|$ \cite{BadoiuCIS05}. In practice, however, typically $f$ cannot be chosen arbitrarily as a function of $X$ anyway. For example, when employing certain learning algorithms such as stochastic gradient descent on dimensionality-reduced data, it is at least required that $f$ is differentiable on $\R^n$ (where $X\subset \R^n$) \cite{WDLSA09}. For several applications it is also crucial that $f$ is linear, e.g.\ in numerical linear algebra \cite{Sarlos06} and compressed sensing \cite{CT05,Donoho04}. In one-pass streaming applications \cite{ClarksonW09} and data structural problems such as nearest neighbor search \cite{HarPeledIM12}, it is further required that $f$ is chosen randomly without knowing $X$. For any particular $X$, a 
random $f$ drawn from some distribution must satisfy the JL guarantee with good probability. In streaming applications this is because $X$ is not fully known up front, but is gradually observed in a stream. In data structure applications this is because $f$ must preserve distances to some future query points, which are not known at the time the data structure is constructed.

Due to the considerations discussed, in practice typically $f$ is chosen as a random linear map drawn from some distribution with a small number of parameters (in some cases simply the parameter $m$). For example, popular choices of $f$ include a random matrix with independent Gaussian \cite{HarPeledIM12} or Rademacher \cite{Achlioptas03} entries. While worst case bounds inform us how to set parameters to obtain the JL guarantee for worst case $X$, we typically can obtain better parameters by exploiting prior knowledge about $X$. Henceforth we only discuss linear $f$, so we write $f(x) = \Phi x$ for $\Phi\in\R^{m\times n}$. Furthermore by linearity, rather than preserving Euclidean distances in $X$ it is equivalent to discuss preserving norms of all vectors in $T = \{(x-y)/\|x-y\|_2 : x,y\in X\}\subset S^{n-1}$, the set of all normalized difference vectors amongst points in $X$. Thus up to changing $\eps$ by roughly a factor of $2$, \Equation{jlf} is equivalent to
\begin{equation}
\sup_{x\in T} \Big|\|\Phi x\|^2 - 1\Big| < \eps . \EquationName{jl-condition}
\end{equation}
Furthermore, since we consider $\Phi$ chosen at random, we more specifically want
\begin{equation}
\E_\Phi \sup_{x\in T} \Big|\|\Phi x\|^2 - 1\Big| < \eps .\EquationName{rjl-condition}
\end{equation}

Instance-wise understanding for achieving \Equation{rjl-condition} was first provided by Gordon \cite{Gordon88}, who proved that a random $\Phi$ with i.i.d.\ Gaussian entries satisfies \Equation{rjl-condition} as long as $m \gtrsim (g^2(T)+1)/\eps^2$, where we write $A\gtrsim B$ if $A \ge CB$ for a universal constant $C>0$. Denoting by $g$ a standard $n$-dimensional Gaussian vector, the parameter $g(T)$ is defined as the {\em Gaussian mean width}
$$g(T)\eqdef \E_g\sup_{x\in T}\langle g,x\rangle .$$
One can think of $g(T)$ as describing the $\ell_2$-geometric complexity of $T$. It is always true that $g^2(T) \lesssim \log |T|$, and thus Gordon's theorem implies the JL lemma. In fact for all $T$ we know from applications, such as for the restricted isometry property from compressed sensing \cite{CT05} or subspace embeddings from numerical linear algebra \cite{Sarlos06}, the best bound on $m$ is a corollary of Gordon's theorem. Later works extended Gordon's theorem to other distributions for $\Phi$, such as $\Phi$ having independent subgaussian entries (e.g.\ Rademachers) \cite{KM05,MPT07,Dirksen14}.

Although Gordon's theorem gives a good understanding for $m$ in most scenarios, it suffers from the fact that it analyzes a {\em dense} random $\Phi$, which means that performing the dimensionality reduction $x\mapsto \Phi x$ is dense matrix-vector multiplication, and is thus slow. For example, in some numerical linear algebra applications (such as least squares regression \cite{Sarlos06}), multiplying a dense unstructured $\Phi$ with the input turns out to be slower than obtaining the solution of the original, high-dimensional problem! In compressed sensing, certain iterative recovery algorithms such as CoSamp \cite{NT09} and Iterative Hard Thresholding \cite{BD08} involve repeated multiplications by $\Phi$ and $\Phi^*$, the conjugate transpose of $\Phi$, and thus $\Phi$ supporting fast matrix-vector multiply are desirable in such applications as well.

The first work to provide $\Phi$ with small $m$ supporting faster multiplication is the Fast Johnson-Lindenstrauss Transform (FJLT) of \cite{AC09} for finite $T$. The value of $m$ was still $O(\eps^{-2}\log|T|)$, with the time to multiply $\Phi x$ being $O(n\log n + m^3)$. \cite{AL09} later gave an improved construction with time $O(n\log n + m^{2+\gamma})$ for any small constant $\gamma>0$. Most recently several works gave nearly linear embedding time in $n$, independent of $m$, at the expense of increasing $m$ by a $(\log n)^c$ factor \cite{AL13,KW11,NPW14}. In all these works $\Phi$ is the product of some number of very sparse matrices and Fourier matrices, with the speed coming from the Fast Fourier Transform (FFT) \cite{CooleyT65}. It is also known that this FFT-based approach can be used to obtain fast RIP matrices for compressed sensing \cite{CT06,RV08,CGV13} and fast oblivious subspace embeddings for numerical linear algebra applications \cite{Sarlos06} (see also \cite{Tropp11,LDFU13} for refined 
analyses in the latter case).

Another line of work, initiated in \cite{Achlioptas03} and greatly advanced in \cite{DKS10}, sought fast embedding time by making $\Phi$ sparse. If $\Phi$ is drawn from a distribution over matrices having at most $s$ non-zeroes per column, then $\Phi x$ can be computed in time $s\cdot \|x\|_0$. After some initial improvements \cite{KN10,BOR10}, the best known achievable value of $s$ to date for the JL lemma while still maintaining $m \lesssim \eps^{-2}\log |T|$ is the sparse Johnson-Lindenstrauss Transform (SJLT) of \cite{KN14}, achieving $s \lesssim \eps^{-1}\log |T| \lesssim \eps m$. Furthermore, an example of a set $T$ exists which requires this bound on $s$ up to $O(\log(1/\eps))$ for any linear JL map \cite{NN13a}. Note however that, again, this is an understanding of the {\em worst-case} parameter settings over all $T$.

In summary, while Gordon's theorem gives us a good understanding of instance-wise bounds on $T$ for achieving good dimensionality reduction, it only does so for dense, slow $\Phi$.  Meanwhile, our understanding for efficient $\Phi$, such as the SJLT with small $s$, has not moved beyond the worst case. In some very specific examples of $T$ we do have good bounds for settings of $s, m$ that suffice, such as $T$ the unit norm vectors in a $d$-dimensional subspace \cite{CW13,MM13,NN13b}, or all elements of $T$ having small $\ell_\infty$ norm \cite{Matousek08,DKS10,KN10,BOR10}. However, our understanding for general $T$ is non-existent. This brings us to the main question addressed in this work, where $S^{n-1}$ denotes the $\ell_2$-unit sphere $\{x\in\R^n : \|x\|_2 = 1\}$.

\begin{question}\QuestionName{main}
\textup{Let $T\subseteq S^{n-1}$ and $\Phi$ be the SJLT. What relationship must $s,m$ satisfy, in terms of the geometry of $T$, to ensure \Eqsub{rjl-condition}?}
\end{question}

We also note that while the FFT-based and sparse $\Phi$ approaches seem orthogonal at first glance, the two are actually connected, as pointed out before \cite{AC09,Matousek08,NPW14}. The FJLT sets $\Phi = SP$ where $P$ is some random preconditioning matrix that makes $T$ ``nice'' with high probability, and $S$ is a random sparse matrix. We point out that although the SJLT is not the same as the matrix $S$ typically used in the FFT-based literature, one could replace $S$ with the SJLT and hope for similar algorithmic outcome if $s$ is small: nearly linear embedding time.

The answer to the analog of \Question{main} for a standard Gaussian matrix depends only on the $\ell_2$-metric structure of $T$. Indeed, since both $\ell_2$-distances and Gaussian matrices are invariant under orthogonal transformations, so is \Eqsub{rjl-condition} in this case. This is reflected in Gordon's theorem, where the embedding dimension $m$ is governed by the Gaussian width, which is invariant under orthogonal transformations. We stress that in sharp contrast, a resolution of \Question{main} cannot solely depend on the $\ell_2$-metric structure of $T$. Indeed, we require that $\Phi$ be sparse {\em in a particular basis} and is therefore not invariant under orthogonal transformations. Thus an answer to \Question{main} must be more nuanced (see our main theorem, \Theorem{main}).

\medskip

\paragraph{\textbf{Our Main Contribution:}} We provide a general theorem which answers \Question{main}. Specifically, for every $T\subseteq S^{n-1}$ analyzed in previous work that we apply our general theorem to here, we either (1) qualitatively recover or improve the previous best known result, or (2) prove the first non-trivial result for dimensionality reduction with sparse $\Phi$. We say ``qualitatively'' since applying our general theorem to these applications loses a factor of $\log^c(n/\eps)$ in $m$ and $\log^c(n/\eps)/\eps$ in $s$. 

In particular for (2), our work is the first to imply that non-trivially sparse dimensionality reducing linear maps can be applied for gain in model-based compressed sensing \cite{BCDH10}, manifold learning \cite{TSL00,DonohoG03}, and constrained least squares problems such as the popular Lasso \cite{Tibshirani96}.

\medskip

Specifically, we prove the following theorem.

\begin{theorem}[Main Theorem]\TheoremName{main}
Let $T\subset S^{n-1}$ and $\Phi$ be an SJLT with column sparsity $s$. Define the complexity parameter 
$$\kappa(T)\eqdef \kappa_{s,m}(T) = \max_{q\le \frac ms \log s} \Big\{\frac 1{\sqrt{qs}} \Big(\E_\eta \Big(\E_g \sup_{x\in T} |\sum_{j=1}^n \eta_j g_j x_j|\Big)^q\Big)^{1/q}\Big\},$$
where $(g_j)$ are i.i.d.\ standard Gaussian and $(\eta_j)$ i.i.d.\ Bernoulli with mean $qs/(m\log s)$. If
\begin{align*}
m &\gtrsim (\log m)^3(\log n)^5 \cdot\frac{(g^2(T) + 1)}{\eps^2}\\
s &\gtrsim (\log m)^6 (\log n)^4 \cdot\frac{1}{\eps^2} .
\end{align*}
Then \Eqsub{rjl-condition} holds as long as $s,m$ furthermore satisfy the condition
\begin{equation}
(\log m)^2(\log n)^{5/2}\kappa(T) < \eps \EquationName{kappa-suffices} .
\end{equation}
\end{theorem}

The complexity parameter $\kappa(T)$ may seem daunting at first, but \Section{applications} shows it can be controlled quite easily for all the $T$ we have come across in applications.

\subsection{Applications}

Here we describe various $T$ and their importance in certain applications. Then  we state the consequences that arise from our theorem. In order to highlight the qualitative understanding arising from our work, we introduce the notation $A\lts B$ if $A \le B \cdot \eps^{-2}(\log n)^c$. A summary of our bounds is in \Figure{fig}.

\medskip

\paragraph{\textbf{Finite $|T|$:}} This is the setting $|T|<\infty$, for which the SJLT satisfies \Equation{rjl-condition} with $s\lesssim \eps^{-1}\log|T|$, $m\lesssim \eps^{-2}\log|T|$ \cite{KN14}. If also $T\subset B_{\ell_\infty^n}(\alpha)$, i.e.\ $\|x\|_\infty \le \alpha$ for all $x\in T$, \cite{Matousek08} showed it is possible to achieve $m \lesssim \eps^{-2}\log|T|$ with a $\Phi$ that has an {\em expected} $O(\eps^{-2}(\alpha\log|T|)^2)$ non-zeroes per column.

Our theorem implies $s,m\lts \log|T|$ suffices in general, and $s\lts 1 + (\alpha \log |T|)^2, m \lts \log|T|$ in the latter case, qualitatively matching the above.

\medskip

\paragraph{\textbf{Linear subspace:}} Here $T = \{x \in E : \|x\|_2 = 1\}$ for some $d$-dimensional linear subspace $E\subset \R^n$, for which achieving \Equation{rjl-condition} with $m\lesssim d/\eps^2$ is possible \cite{AHK06,CW13}. A distribution satisfying \Equation{rjl-condition} for any $d$-dimensional subspace $E$ is known as an {\em oblivious subspace embedding (OSE)}. The paper \cite{Sarlos06} pioneered the use of OSE's for speeding up approximate algorithms for numerical linear algebra problems such as low-rank approximation and least-squares regression. More applications have since been found to approximating leverage scores \cite{DMMW12}, $k$-means clustering \cite{BZMD11,CohenEMMP14}, canonical correlation analysis \cite{ABTZ13}, support vector machines \cite{PBMD13}, $\ell_p$-regression \cite{CDM+13,WZ13}, ridge regression \cite{LDFU13}, streaming approximation of eigenvalues \cite{AN13}, and speeding up interior point methods for linear programming \cite{LeeS13a}. In many of these applications there is 
some input $A\in\R^{n\times d}$, $n\gg d$, and the subspace $E$ is for example the column space of $A$. Often an exact solution requires computing the singular value decomposition (SVD) of $A$, but using OSE's the running time is reduced to that for computing $\Phi A$, plus computing the SVD of the smaller matrix $\Phi A$. The work \cite{CW13} showed $s=1$ with small $m$ is sufficient, yielding algorithms for least squares regression and low-rank approximation whose running times are linear in the number of non-zero entries in $A$ for sufficiently lopsided rectangular matrices.

Our theorem implies that $s \lts 1$ and $m \lts d$ suffices to satisfy \Equation{rjl-condition}, which is qualitatively correct. Furthermore, a subset of our techniques reveals that if the maximum incoherence $\max_{1\le i\le n} \|P_E e_i\|_2$ is at most $\poly(\eps/\log n)$, then $s=1$ suffices (Theorem~\ref{thm:JLsubspace}). This was not known in previous work. A random $d$-dimensional subspace has incoherence $\sqrt{d/n}$ w.h.p.\ for $d\gtrsim\log n$ by the JL lemma, and thus is very incoherent if $n\gg d$.

\medskip

\paragraph{\textbf{Closed convex cones:}} For $A\in\R^{n\times d}$, $b\in\R^n$, and $\cC\subseteq \R^d$ a closed convex set, consider the constrained least squares problem of minimizing $\|Ax-b\|_2^2$ subject to $x\in\cC$. A popular choice is the Lasso \cite{Tibshirani96}, in which the constraint set $\cC = \{ x\in\R^d : \|x\|_1 \le R\}$ encourages sparsity of $x$. Let $x_*$ be an optimal solution, and let $T_{\cC}(x)$ be the tangent cone of $\cC$ at some point $x\in\R^d$ (see Appendix~\ref{sec:convex-tools} for a definition). Suppose we wish to accelerate approximately solving the constrained least squares problem by instead computing a minimizer $\hat{x}$ of $\|\Phi Ax - \Phi b\|_2^2$ subject to $x\in \cC$. The work \cite{PiW14} showed that to guarantee $\|A\hat{x} - b\|_2^2 \le (1+\eps)\|Ax_* - b\|_2^2$, it suffices that $\Phi$ satisfy two conditions, one of which is \Equation{jl-condition} for $T = A T_{\cC}(x_*)\cap S^{n-1}$. The paper \cite{PiW14} then analyzed dense random matrices for 
sketching constrained least squares problems. For example, for the Lasso if we are promised that the optimal solution $x_*$ is $k$-sparse, \cite{PiW14} shows that it suffices to set
$$
m\gtrsim \frac 1{\eps^2}\max_{j=1,\ldots, d} \frac{\|A_j\|_2^2}{\si_{\min,k}^2}k\log d,
$$
where $A_j$ is the $j$th column of $A$ and $\si_{\min,k}$ the smallest $\ell_1$-restricted eigenvalue of $A$:
$$
\si_{\min,k} = \inf_{\|y\|_2=1, \ \|y\|_{1}\leq 2\sqrt{k}}\|Ay\|_2 .
$$

Our work also applies to such $T$ (and we further show the SJLT with small $s, m$ satisfies the second condition required for approximate constrained least squares; see Theorem~\ref{thm:SJLmain}). For example for the Lasso, we show that again it suffices that 
$$
m\gts \max_{j=1,\ldots,d} \frac{\|A_j\|_2^2}{\si_{\min,k}^2}k\log d
$$
but where we only require from $s$ that
$$
s\gts \max_{\substack{1\le i\le n\\1\le j\le d}} \frac{A_{ij}^2}{\si_{\min,k}^2}k
$$
That is, the sparsity of $\Phi$ need only depend on the largest entry in $A$ as opposed to the largest column norm in $A$. The former can be much smaller.

\medskip

\paragraph{\textbf{Unions of subspaces:}} Define $T = \cup_{\theta\in \Theta} E_\theta\cap S^{n-1}$, where $\Theta$ is some index set and each $E_\theta\subset\R^n$ is a $d$-dimensional linear subspace. A case of particular interest is when $\theta\in\Theta$ ranges over all $k$-subsets of $\{1,\ldots,n\}$, and $E_\theta$ is the subspace spanned by $\{e_j\}_{j\in \theta}$ (so $d=k$). Then $T$ is simply the set of all $k$-sparse unit vectors of unit Euclidean norm: $S_{n,k}\eqdef \{x \in \R^n : \|x\|_2 = 1, \|x\|_0 \le k\}$ for $\|\cdot\|_0$ denoting support size. $\Phi$ satisfying \Eqsub{jl-condition} is then referred to as having the {\em restricted isometry property (RIP) of order $k$} with restricted isometry constant $\eps_k = \eps$ \cite{CT05}. Such $\Phi$ are known to exist with $m\lesssim \eps_k^{-2}k\log(n/k)$, and furthermore it is known that $\eps_{2k} < \sqrt{2}-1$ implies that any (approximately) $k$-sparse $x\in\R^n$ can be (approximately) recovered from $\Phi x$ in polynomial time by solving 
a certain linear program \cite{CT05,Candes08}. Unfortunately it is known for $\eps = \Theta(1)$ that {\em any} RIP matrix $\Phi$ with such small $m$ must have $s\gtrsim m$ \cite{NN13a}. Related is the case of vectors sparse in some other basis, i.e.\ $T = \{Dx \in \R^n : \|Dx\|_2 = 1,\|x\|_0 \le k\}$ for some so-called ``dictionary'' $D$ (i.e.\ the subspaces are acted on by $D$), or when $T$ only allows for some subset of all $\binom{n}{k}$ sparsity patterns in {\em model-based} compressed sensing \cite{BCDH10} (so that $|\Theta| < \binom{n}{k}$).

Our main theorem also implies RIP matrices with $s, m \lts k\log(n/k)$. More generally when a dictionary $D$ is involved such that the subspaces $\mathrm{span}(\{D e_j\}_{j\in \theta})$ are all $\alpha$-incoherent (as defined above), the sparsity can be further improved to $s\lts 1 + (\alpha k\log(n/k))^2$. That is, to satisfy the RIP with dictionaries yielding incoherent subspaces, we can keep $m$ qualitatively the same while making $s$ much smaller. For the general problem of unions of $d$-dimensional subspaces, our theorem implies one can either set $m \lts d + \log|\Theta|, s \lts \log|\Theta|$ or $m \lts d + \log|\Theta|, s \lts 1 + (\alpha\log|\Theta|)^2$. Previous work required $m$ to depend on the {\em product} of $d$ and $(\log|\Theta|)^c$ instead of the {\em sum} \cite{NN13b}, including a nice recent improvement by Cohen \cite{Cohen14}, and is thus unsuitable for the application to the set of $k$-sparse vectors (RIP matrices with $\lts k^2$ rows are already attainable via simpler methods using incoherence; e.g.\ see \cite[Proposition 1]{BDFKK11}). 
Iterative recovery algorithms such as CoSamp can also be used in model-based sparse recovery \cite{BCDH10}, which again involves multiplications by $\Phi, \Phi^*$, and thus sparse $\Phi$ is relevant for faster recovery. Our theorem thus shows, for the first time, that the benefit of model-based sparse recovery is not just a smaller $m$, but rather that the measurement matrix $\Phi$ can be made much sparser if the model is simple (i.e.\ $|\Theta|$ is small).  For example, in the {\em block-sparse model} one wants to (approximately) recover a signal $x\in\R^n$ based on $m$ linear measurements, where $x$ is (approximately) $k$-block-sparse. That is, the $n$ coordinates are partitioned into $n/b$ blocks of size $b$ each, and each block is either ``on'' (at least one coordinate in that block non-zero), or ``off'' (all non-zero). A $k$-block-sparse signal has at most $k/b$ blocks on (thus $\|x\|_0 \le k$). Thus $s\lts \log|\Theta| = \log (\binom{n/b}{k/b}) \lesssim (k/b)\log(n/k)$. Then as long as $b = \omega(\log(n/k))$, 
our results imply non-trivial column-sparsity $s \ll m$. Ours is the first result yielding non-trivial sparsity in a model-RIP $\Phi$ for any model with a number of measurements qualitatively matching the optimal bound (which is on the order of $m\lesssim k + (k/b)\log(n/k)$ \cite{ADR14}). We remark that for model-based RIP$_1$, where one wants to approximately preserve $\ell_1$-norms of $k$-block-sparse vectors, which is useful for $\ell_1/\ell_1$ recovery, \cite{IndykR13} have shown a much better sparsity bound of $O(\ceil{\log_b(n/k)})$ non-zeroes per column in their measurement matrix. However, they have also shown that any model-based RIP$_1$ matrix for block-sparse signals must satisfy the higher lower bound of $m\gtrsim k + (k/\log b)\log(n/k)$ (which is tight).

Previous work also considered $T = H S_{n,k}$, where $H$ is any bounded orthonormal system, i.e.\ $H\in\R^{n\times n}$ is orthogonal and $\max_{i,j} |H_{ij}| = O(1/\sqrt{n})$ (e.g.\ the Fourier matrix). Work of \cite{CT06,RV08,CGV13} shows $\Phi$ can then be a sampling matrix (one non-zero per row) with $m\lesssim \eps^{-2} k\log(n)(\log k)^3$. Since randomly flipping the signs of every column in an RIP matrix yields JL \cite{KW11}, this also gives a good implementation of an FJLT. Our theorem recovers a similar statement, but using the SJLT instead of a sampling matrix, where we show $m \lts k$ and $s \lts 1$ suffice for orthogonal $H$ satisfying the weaker requirement $\max_{i,j} |H_{ij}| = O(1/\sqrt{k})$.

\medskip

\paragraph{\textbf{Smooth manifolds:}} 

Suppose we are given several images of a human face, but with varying lighting and angle of rotation. Or perhaps the input is many sample handwritten images of letters. In these examples, although input data is high-dimensional ($n$ is the number of pixels), we imagine all possible inputs come from a set of low intrinsic dimension. That is, they lie on a $d$-dimensional manifold $\mathcal{M}\subset \R^n$ where $d\ll n$. The goal is then, given a large number of manifold examples, to learn the parameters of $\mathcal{M}$ to allow for nonlinear dimensionality reduction (reducing just to the few parameters of interest). This idea, and the first successful algorithm (ISOMAP) to learn a manifold from sampled points is due to \cite{TSL00}. Going back to the example of human faces, \cite{TSL00} shows that different images of a human face can be well-represented by a 3-dimensional manifold, where the parameters are brightness, left-right angle of rotation, and up-down angle of rotation. Since \cite{TSL00}, several 
more algorithms have been developed to handle more general classes of manifolds than ISOMAP \cite{RS00,DonohoG03}.

Baraniuk and Wakin \cite{BW09} proposed using dimensionality-reducing maps to first map $\mathcal{M}$ to $\Phi \mathcal{M}$ and then learn the parameters of interest in the reduced space (for improved speed). Sharper analyses were later given in \cite{Clarkson08,EfW13,Dirksen14}. Of interest are both that (1) any $C^1$ curve in $\mathcal{M}$ should have length approximately preserved in $\Phi\mathcal{M}$, and (2) $\Phi$ should be a {\em manifold embedding}, in the sense that all $C^1$ curves $\gamma'\in\Phi\mathcal{M}$ should satisfy that the preimage of $\gamma'$ (in $\mathcal{M}$) is a $C^1$ curve in $\mathcal{M}$. Then by (1) and (2), {\em geodesic} distances are preserved by $\Phi$.

To be concrete, let $\mathcal{M}\subset\R^n$ be a $d$-dimensional manifold obtained as the image $\mathcal{M} = F(B_{\ell_2^d})$, for smooth $F:B_{\ell_2^d}\rightarrow \R^n$ ($B_X$ is the unit ball of $X$). We assume $\|F(x) - F(y)\|_2 \simeq \|x - y\|_2$ (where $A\simeq B$ denotes that both $A\lesssim B$ and $A\gtrsim B$), and that the map sending $x\in\mathcal{M}$ to the tangent plane at $x$, $E_x$, is Lipschitz from $d_{\mathcal{M}}$ to $\rho_{\mathrm{Fin}}$. Here $\rho_{\mathcal{M}}$ is geodesic distance on $\mathcal{M}$, and $\rho_{\mathrm{Fin}}(E_x, E_y) = \|P_{E_x} - P_{E_y}\|_{\ell_2^n\rightarrow\ell_2^n}$ is the {\em Finsler distance}, where $P_E$ is the orthogonal projection onto $E$.

We want $\Phi$ satisfying $(1-\eps)|\gamma| \le |\Phi(\gamma)| \le (1+\eps)|\gamma|$ for all $C^1$ curves $\gamma\subset\mathcal{M}$. Here $|\cdot|$ is curve length. To obtain this, it suffices that $\Phi$ satisfy \Equation{jl-condition} for $T = \bigcup_{x\in \mathcal{M}} E_x \cap S^{n-1}$ \cite{Dirksen14}, an infinite union of subspaces. Using this observation, \cite{Dirksen14} showed this property is satisfied for $s=m \lesssim d/\eps^2$ with a {\em dense} matrix of subgaussian entries. For $F$ as given above, preservation of geodesic distances is also satisfied for this $m$.

Our main theorem implies that to preserve curve lengths one can set $m \lts d$ for $s \lts 1 + (\alpha d)^2$, where $\alpha$ is the largest incoherence of any tangent space $E_x$ for $x\in\mathcal{M}$. That is, non-trivial sparsity with $m\lts d$ is possible for any $\alpha \ll 1/\sqrt{d}$. Furthermore, we show that this is {\em optimal} by constructing a manifold with maximum incoherence of a tangent space approximately $1/\sqrt{d}$ such that {\em any} linear map $\Phi$ preserving curve lengths with $m\gts d$ must have $s\gts d$ (see \Remark{bad-manifold}). We also show that $\Phi$ is a manifold embedding with large probability if the weaker condition $m\gts d, s \gts 1$ is satisfied. Combining these observations, we see that for the SJLT to preserve geodesic distances it suffices to set $m \lts d$ and $s \lts 1 + (\alpha d)^2$.

\vspace{.1in}

\begin{figure}
\begin{center}
{\scriptsize 
\begin{tabular}{|c|c|c|c|c|c|}
\hline
\textbf{set $T$ to preserve} & \textbf{our $m$} & \textbf{our $s$} & \textbf{previous $m$} & \textbf{previous $s$} & \textbf{ref}\\
\hline
$|T|<\infty$ & $\log |T|$ & $\log |T|$ & $\log |T|$ & $\log |T|$ & \cite{JL84}\\
\hline
$|T|<\infty,\forall x\in T \|x\|_\infty \le \alpha$ & $\log |T|$ & $\ceil{\alpha\log|T|}^2$ & $\log |T|$ & $\ceil{\alpha\log|T|}^2$ & \cite{Matousek08}\\
\hline
$E,\dim(E)\le d$ & $d$ & $1$ & $d$ & $1$ & \cite{NN13b} \\
\hline
$S_{n,k}$ & $k\log(n/k)$ & $k\log(n/k)$ & $k\log(n/k)$ & $k\log(n/k)$ & \cite{CT05}\\
\hline
$HS_{n,k}$ & $k\log(n/k)$ & $1$ & $k\log(n/k)$ & $1$ & \cite{RV08}\\
\hline
tangent cone for Lasso & $\max_j\frac{\|A_j\|_2^2}{\si_{min,k}^2}k$ & $\max_{i,j}\frac{A_{ij}^2}{\si_{min,k}^2}k$ & same as here & $s=m$ & \cite{PiW14}\\
\hline
$|\Theta|<\infty$ & $d + \log|\Theta|$ & $\log|\Theta|$ & $d \cdot (\log|\Theta|)^6$& $(\log|\Theta|)^3$ & \cite{NN13b}\\
$\forall E\in\Theta,\dim(E)\le d$ &&&&&\\
\hline
$|\Theta|<\infty$ & $d + \log|\Theta|$ & $\ceil{\alpha \log|\Theta|}^2$ & --- & --- & --- \\
$\forall E\in\Theta,\dim(E)\le d$ &&&&&\\
$\max_{\substack{1\le j\le n\\ E\in\Theta}} \|P_E e_j\|_2 \le \alpha$&&&&&\\
\hline
$|\Theta|$ infinite & see appendix & see appendix & similar to& $m$ & \cite{Dirksen14}\\
$\forall E\in\Theta,\dim(E)\le d$ & & (non-trivial) & this work & &\\
\hline
$\mathcal{M}$ a smooth manifold & $d$& $1+(\alpha d)^2$&$d$ & $d$ & \cite{Dirksen14} \\
\hline
\end{tabular}
}
\caption{The $m,s$ that suffice when using the SJLT with various $T$ as a consequence of our main theorem, compared with the best known bounds from previous work. All bounds shown hide $\mathrm{poly}(\eps^{-1}\log n)$ factors. One row is blank in previous work due to no non-trivial results being previously known. For the case of the Lasso, we assume $k$ is the sparsity of the true optimal solution.}\FigureName{fig}
\end{center}
\end{figure}

As seen above, not only does our answer to \Question{main} qualitatively explain all known results, but it gives new results not known before with implications in numerical linear algebra, compressed sensing (model-based, and with incoherent dictionaries), constrained least squares, and manifold learning. We also believe it is possible for future work to sharpen our analyses to give asymptotically correct parameters for all the applications; see the discussion in \Section{discussion}. 

We now end the introduction with an outline for the remainder. \Section{prelim} defines the notation for the rest of the paper. \Section{overview} provides an overview of the proof of our main theorem, \Theorem{main}. \Section{linear} is a warmup that applies a subset of the proof ideas for \Theorem{main} to the special case where $T = E\cap S^{n-1}$, $E$ a linear subspace of dimension $d$. In fact the proof of our main theorem reduces the case of general $T$ to several linear subspaces of varying dimensions, and thus this case serves as a useful warmup. \Section{typetwo} applies a different subset of our proof ideas to the special case where the norm $\tnorm{\cdot}_T$ defined by $\tnorm{y}_T = \sup_{x\in T}|\inprod{x,y}|$ has a small type-$2$ constant, which is relevant for analyzing the FFT-based approaches to RIP matrices. \Section{cls} shows how similar ideas can be applied to constrained least squares problems, such as Lasso. \Section{general} states and proves our most general theorem for arbitrary $T$.
 \Section{applications} shows how to apply \Theorem{main} to obtain good bounds for various $T$, albeit losing a $\log^c(n/\eps)$ factor in $m$ and a $\log^c(n/\eps)/\eps$ factor in $s$ as mentioned above. Finally, \Section{discussion} discusses avenues for future work. 

In the appendix, in Appendix~\ref{sec:probTools} for the benefit of the reader we review many probabilistic tools that are used throughout this work . Appendix~\ref{sec:convex-tools} reviews some introductory material related to convex analysis, which is helpful for understanding \Section{cls} on constrained least squares problems. In Appendix~\ref{sec:fjlt-cls} we also give a direct analysis for using the FJLT for sketching constrained least squares programs, providing quantitative benefits over some analyses in \cite{PiW14}.

\section{Preliminaries}\SectionName{prelim}
 
We fix some notation that we will be used throughout this paper. For a positive integer $t$, we set $[t] = \{1,\ldots,t\}$. For $a,b \in \R$, $a\lesssim b$ denotes $a \le C b$ for some universal constant $C > 0$, and $a\simeq b$ signifies that both $a\lesssim b$ and $b\lesssim a$ hold. For any $x\in \R^n$ and $1\leq p\leq\infty$, we let $\|x\|_p$ denote its $\ell_p$-norm. To any set $S\subset\R^n$ we associate a semi-norm $\tnorm{\cdot}_S$ defined by $\tnorm{z}_S = \sup_{x\in S}|\inprod{z,x}|$. Note that $\tnorm{z}_S = \tnorm{z}_{\mathrm{conv}(S)}$, where $\mathrm{conv}(S)$ is the closed convex hull of $S$, i.e., the closure of the set of all convex combinations of elements in $S$. We use $(e_i)_{1\leq i\leq n}$ and $(e_{ij})_{1\leq i\leq m,1\leq j\leq n}$ to denote the standard basis in $\R^n$ and $\R^{m\ti n}$, respectively.\par
If $\eta=(\eta_i)_{i\geq 1}$ is a sequence of random variables, we let $(\Omega_{\eta},\bP_{\eta})$ denote the probability space on which it is defined. We use $\E_{\eta}$ and $L_{\eta}^p$ to denote the associated expected value and $L^p$-space, respectively. If $\zeta$ is another sequence of random variables, then $\|\cdot\|_{L^p_{\eta},L^q_{\zeta}}$ means that we first take the $L^p_{\eta}$-norm and afterwards the $L^q_{\zeta}$-norm. We reserve the symbol $g$ to denote a sequence $g=(g_i)_{i\geq 1}$ of i.i.d.\ standard Gaussian random variables; unless stated otherwise, the covariance matrix is the identity.\par
If $A$ is an $m\ti n$ matrix, then we use $\|A\|$ or $\|A\|_{\ell_2^n\to\ell_2^m}$ to denote its operator norm. Moreover we let $\mathrm{Tr}$ be the trace operator and use $\|A\|_F=(\mathrm{Tr}(A^*A))^{1/2}$ to denote the Frobenius norm. 

In the remainder, we reserve the letter $\rho$ to denote (semi-)metrics. If $\rho$ corresponds to a (semi-)norm $\|\cdot\|_X$, then we let $\rho_X(x,y)=\|x-y\|_X$ denote the associated (semi-)metric. Also, we use $d_{\rho}(S)=\sup_{x,y \in S}\rho(x,y)$ to denote the diameter of a set $S$ with respect to $\rho$ and write $d_X$ instead of $d_{\rho_X}$ for brevity. So, for example, $\rho_{\ell_2^n}$ is the Euclidean metric and $d_{\ell_2^n}(S)$ the $\ell_2$-diameter of $S$. From here on, $T$ is always a fixed subset of $S^{n-1} = \{x\in\R^n : \|x\| = 1\}$, and $\eps\in (0,1/2)$ the parameter appearing in \Eqsub{rjl-condition}.

We make use of chaining results in the remainder, so we define some relevant quantities. For a bounded set $S\subset \R^n$, $g(S) = \E_g \sup_{x\in S} \inprod{g, x}$ is the {\em Gaussian mean width} of $S$, where $g\in\R^n$ is a Gaussian vector with identity covariance matrix. For a (semi-)metric $\rho$ on $\R^n$, Talagrand's $\ga_2$-functional is defined by
\begin{equation}
\ga_2(S,\rho) = \inf_{\{S_r\}_{r=0}^\infty} \sup_{x\in S} \sum_{r=0}^\infty 2^{r/2} \cdot \rho(x, S_r)
\end{equation}
where $\rho(x, S_r)$ is the distance from $x$ to $S_r$, and the infimum is taken over all collections $\{S_r\}_{r=0}^\infty$, $S_0 \subset S_1 \subset \ldots \subseteq S$, with $|S_0| = 1, |S_r| \le 2^{2^r}$. If $\rho$ corresponds to a (semi-)norm $\|\cdot\|_X$, then we usually write $\gamma_2(S, \|\cdot\|_X)$ instead of $\ga_2(S,\rho_X)$. It is known that for any bounded $S\subset \R^n$, $g(S)$ and $\gamma_2(S, \|\cdot\|_2)$ differ multiplicatively by at most a universal constant \cite{Fernique75,Talagrand05}. Whenever $\gamma_2$ appears without a specified norm, we imply use of $\ell_2$ or $\ell_2\rightarrow\ell_2$ operator norm. We frequently use the entropy integral estimate (see \cite{Talagrand05})
\begin{equation}
\gamma_2(S, \rho) \lesssim \int_0^{\infty} \log^{1/2} \mathcal{N}(S,\rho, u) du .\EquationName{Dudley}
\end{equation}
Here $\mathcal{N}(S,\rho, u)$ denotes the minimum number of $\rho$-balls of radius $u$ centered at points in $S$ required to cover $S$. If $\rho$ corresponds to a (semi-)norm $\|\cdot\|_X$, then we write $\mathcal{N}(S,\|\cdot\|_X, u)$ instead of $\mathcal{N}(S,\rho_X, u)$.   

Let us now introduce the sparse Johnson-Lindenstrauss transform in detail. Let $\si_{ij}:\Om_{\si}\rightarrow \{-1,1\}$ be independent Rademacher random variables, i.e., $\bP(\si_{ij}=1)=\bP(\si_{ij}=-1)=1/2$. We consider random variables $\del_{ij}:\Om_{\del}\rightarrow\{0,1\}$ with the following properties:
\begin{itemize}
 \item For fixed $j$ the $\del_{ij}$ are negatively correlated, i.e. 
\begin{equation}
\forall 1\leq i_1<i_2<\ldots<i_k\leq m,\  \E\Big(\prod_{t=1}^k\del_{i_t,j}\Big)\leq \prod_{t=1}^k\E \del_{i_t,j} = \Big(\frac sm\Big)^k;\EquationName{neg-cor}
\end{equation}
 \item For any fixed $j$ there are exactly $s$ nonzero $\del_{ij}$, i.e., $\sum_{i=1}^m \del_{ij}=s$;
  \item The vectors $(\delta_{ij})_{i=1}^m$ are independent across different $1\le j\le n$.
\end{itemize}
We emphasize that the $\si_{ij}$ and $\del_{ij}$ are independent, as they are defined on different probability spaces. The \emph{sparse Johnson-Lindenstrauss transform $\Phi$ with column sparsity $s$}, or \emph{SJLT} for short, is defined by 
\begin{equation}
\Phi_{ij}=\frac{1}{\sqrt{s}}\si_{ij}\del_{ij}.\EquationName{1.1}
\end{equation}
The work \cite{KN14} gives two implementations of such a $\Phi$ satisfying the above conditions. In one example, the columns are independent, and in each column we choose exactly $s$ locations uniformly at random, without replacement, to specify the $\delta_{ij}$. The other example is essentially the CountSketch of \cite{CCF04}. In this implementation, the rows of $\Phi$ are partitioned arbitrarily into $s$ groups of size $m/s$ each. Then each column of $\Phi$ is chosen independently, where in each column and for each block we pick a random row in that block to be non-zero for that column (the rest are zero).

We say that $\Phi$ is an \emph{$\eps$-restricted isometry} on $T$ if \Equation{jl-condition} holds. We define
$$\eps_{\Phi,T} = \sup_{x\in T} |\|\Phi x\|_2^2 - 1|$$
as the \emph{restricted isometry constant} of $\Phi$ on $T$. In the following we will be interested in estimating $\eps_{\Phi,T}$. For this purpose, we use the following $L^p$-bound for the supremum of a second order Rademacher chaos from \cite{KMR14} (see \cite[Theorem 6.5]{Dir13} for the refinement stated here).
\begin{theorem}
Let $\cA\subset \R^{m\ti n}$ and let $\si_1,\ldots,\si_n$ be independent Rademacher random variables. For any $1\leq p<\infty$
\begin{align}
\label{eqn:supChaosImproved}
\Big(\E_{\si}\sup_{A\in\cA}\Big|\|A\si\|_2^2 - \E\|A\si\|_2^2\Big|^p\Big)^{1/p} & \lesssim \ga_{2}^2(\cA,\|\cdot\|_{\ell_2\to\ell_2}) + \DelO_F(\cA)\ga_{2}(\cA,\|\cdot\|_{\ell_2\to\ell_2}) \nonumber \\
& \qquad + \sqrt{p}\DelO_{F}(\cA)\DelO_{\ell_2\to\ell_2}(\cA) + p\DelO_{\ell_2\to\ell_2}^2(\cA).
\end{align}
\end{theorem}
For $u \in \R^n$ and $v\in \R^m$, let $u\ot v:\R^m\rightarrow \R^n$ be the operator $(u\ot v)w=\langle v,w\rangle u$. Then, for any $x\in \R^n$ we can write $\Phi x=A_{\del,x}\si$, where
\begin{equation}
\label{eqn:AdelxDef}
A_{\del,x} := \frac{1}{\sqrt{s}}\sum_{i=1}^m\sum_{j=1}^n \del_{ij}x_j e_i\ot e_{ij} 
=  \frac 1{\sqrt{s}}\begin{bmatrix} 
- x^{(\delta_{1,\cdot})} - & 0 & \cdots & 0\\
0 & - x^{(\delta_{2,\cdot})} - & \cdots & 0\\
\vdots &\vdots &  &\vdots\\
0&0&\cdots& - x^{(\delta_{m,\cdot})} -
\end{bmatrix} .
\end{equation}
for $x^{(\delta_{i,\cdot})}_j = \delta_{ij}x_j$. Note that $\E\|\Phi x\|_2^2 = \|x\|_2^2$ for all $x\in \R^n$ and therefore
$$\sup_{x\in T}|\|\Phi x\|_2^2 - \|x\|_2^2| = \sup_{x\in T} |\|A_{\del,x}\si\|_2^2 - \E\|A_{\del,x}\si\|_2^2|.$$
Associated with $\del=(\del_{ij})$ we define a random norm on $\R^n$ by
\begin{equation}
\|x\|_{\del} = \frac{1}{\sqrt{s}}\max_{1\leq i\leq m}\Big(\sum_{j=1}^n \del_{ij}x_j^2\Big)^{1/2}. \EquationName{defDelNorm}
\end{equation}
With this definition, we have for any $x,y\in T$,
\begin{equation}
\label{eqn:AdelxNorm}
\|A_{\del,x} - A_{\del,y}\| = \|x-y\|_{\del}, \qquad \|A_{\del,x} - A_{\del,y}\|_F = \|x-y\|_2.
\end{equation}
Therefore, (\ref{eqn:supChaosImproved}) implies that
\begin{align}
& \Big(\E_{\si}\sup_{x \in T} \Big|\|\Phi x\|_2^2 - \|x\|_2^2\Big|^p\Big)^{1/p} \lesssim \ga_2^{2}(T,\|\cdot\|_{\del}) + \ga_2(T,\|\cdot\|_{\del}) + \sqrt{p}d_{\del}(T) + pd_{\del}^2(T).\EquationName{RIPExpect}
\end{align}
Taking $L_p(\Om_{\del})$-norms on both sides yields
\begin{align}
\nonumber \Big(\E_{\del,\si}\sup_{x \in T} \Big|\|\Phi x\|_2^2 - \|x\|_2^2\Big|^p\Big)^{1/p} & \lesssim (\E_{\del}\ga_2^{2p}(T,\|\cdot\|_{\del}))^{1/p} + (\E_{\del}\ga_2^{p}(T,\|\cdot\|_{\del}))^{1/p}  \\
& \qquad \qquad + \sqrt{p}(\E d_{\del}^p(T))^{1/p} + p(\E d_{\del}^{2p}(T))^{1/p}. \EquationName{RIPLpEst}
\end{align}
Thus, to find a bound for the expected value of the restricted isometry constant of $\Phi$ on $T$, it suffices to estimate $\E_{\del}\ga_2^2(T,\|\cdot\|_{\del})$ and $\E_{\del}d_{\del}^2(T)$. If we can find good bounds on $(\E_{\del}\ga_2^{p}(T,\|\cdot\|_{\del}))^{1/p}$ and $(\E d_{\del}^p(T))^{1/p}$ for all $p\geq 1$, then we in addition obtain a high probability bound for the restricted isometry constant.\par
Unless explicitly stated otherwise, from here on $\Phi$ always denotes the SJLT with $s$ non-zeroes per column.

\section{Overview of proof of main theorem}\SectionName{overview}
Here we give an overview of the proof of \Theorem{main}. To illustrate the ideas going into our proof, it is simplest to first consider the case where $T$ is the set of all unit vectors in a $d$-dimensional linear subspace $E\subset\R^n$.  By \Equation{RIPLpEst} for $p=1$ we have to bound for example $\E_\del \gamma_2(T,\|\cdot\|_\del)$. Standard estimates give, up to $\log d$,
\begin{equation}
\gamma_2(T, \|\cdot\|_\del) \le \gamma_2(B_E, \|\cdot\|_\del) \ll \sup_{t>0} t [\log\mathcal{N}(B_E, \|\cdot\|_\del, t)]^{1/2} . \EquationName{overview}
\end{equation}
for $B_E$ the unit ball of $E$. Let $U\in\R^{n\times d}$ have columns forming an orthonormal basis for $E$. Then dual Sudakov minoration \cite[Proposition 4.2]{BLM89}, \cite{PTJ86} states that
\begin{equation}
\sup_{t>0} t [\log\mathcal{N}(B_E, \|\cdot\|_\del, t)]^{1/2} \le \E_g \|Ug\|_\del \EquationName{sudakov}
\end{equation}
for a Gaussian vector $g$. From this point onward one can arrive at a result using the non-commutative Khintchine inequality \cite{LP86,LPP91} and other standard Gaussian concentration arguments (see appendix).

Unfortunately, \Equation{sudakov} is very specific to unit balls of linear subspaces and has no analog for a general set $T$. At this point we use a statement about the duality of entropy numbers \cite[Proposition 4]{BPST89}. This is the principle that for two symmetric convex bodies $K$ and $D$, $\mathcal{N}(K, D)$ and $\mathcal{N}(D^\circ, a K^\circ)$ are roughly comparable for some constant $a$ ($\mathcal{N}(K,D)$ is the number of translates of $D$ needed to cover $K$). Although it has been an open conjecture for over 40 years as to whether this holds in the general case \cite[p.\ 38]{P72}, the work \cite[Proposition 4]{BPST89} shows that these quantities are comparable up to logarithmic factors as well as a factor depending on the type-$2$ constant of the norm defined by $D$ (i.e.\ the norm whose unit vectors are those on the boundary of $D$). In our case, this lets us relate $\log \mathcal{N}(\tilde{T}, \|\cdot\|_\del, t)$ with $\log\mathcal{N}(\mathrm{conv}(\cup_{i=1}^m B_{J_i}), \tnorm{\cdot}_T, \sqrt{s}t/8)$, losing small factors. Here $\tilde{T}$ is the convex hull of $T\cup -T$, and $B_{J_i}$ is the unit ball of $\mathrm{span}\{e_j : \delta_{ij} = 1\}$. We next use Maurey's lemma, which is a tool for bounding covering numbers of the set of convex combinations of vectors in various spaces. This lets us relate $\log\mathcal{N}(\mathrm{conv}(\cup_{i=1}^m B_{J_i}), \tnorm{\cdot}_T, \epsilon)$ to quantities of the form $\log\mathcal{N}(\frac 1k\sum_{i\in A}B_{J_i}, \tnorm{\cdot}_T, \epsilon)$, where $A\subset [m]$ has size $k\lesssim 1/\epsilon^2$ (see \Lemma{maurey}). For a fixed $A$, we bucket $j\in [n]$ according to $\sum_{i\in A} \delta_{ij}$ and define $U_\alpha = \{j\in [n] : \sum_{i\in A} \delta_{ij} \simeq 2^\alpha\}$. Abusing notation, we also let $U_\alpha$ denote the coordinate subspace spanned by $j\in U_\alpha$. This leads to (see \Equation{5.4})
\begin{equation}
\log\mathcal{N}\Big(\frac 1k\sum_{i\in A}B_{J_i}, \tnorm{\cdot}_T, \epsilon\Big) \lesssim \sum_\alpha \log \mathcal{N}\Big(B_{U_\alpha}, \tnorm{\cdot}_T, \sqrt{\frac{k}{2^\alpha}} \frac{\epsilon}{\log m}\Big) \EquationName{finally-sudakov}
\end{equation}

Finally we are in a position to apply dual Sudakov minoration to the right hand side of \Equation{finally-sudakov}, after which point we apply various concentration arguments to yield our main theorem.

\section{The case of a linear subspace}\SectionName{linear}
Let $E\subset\R^n$ be a $d$-dimensional linear subspace, $T = E\cap S^{n-1}$, $B_E$ the unit $\ell_2$-ball of $E$. We use $P_E$ to denote the orthogonal projection onto $E$. The values $\|P_E e_j\|_2$, $j=1,\ldots,n$, are typically referred to as the {\em leverage scores} of $E$ in the numerical linear algebra literature. We denote the maximum leverage score by
$$\mu(E) = \max_{1\leq j\leq n}\|P_E e_j\|_2,$$
which is sometimes referred to as the {\em incoherence} $\mu(E)$ of $E$.

In our proof we make use of the well-known Gaussian concentration of Lipschitz functions \cite[Corollary 2.3]{Pi86}: if $g$ is an $n$-dimensional standard Gaussian vector and $\phi:\R^n\rightarrow \R$ is $L$-Lipschitz from $\ell_2^n$ to $\R$, then for all $p\ge 1$
\begin{equation}
\|\phi(g) - \E\phi(g)\|_{L^p_g} \lesssim L \sqrt{p}. \EquationName{GaussLip}
\end{equation}
\begin{theorem}
\label{thm:JLsubspace}
For any $p\geq \log m$ and any $0<\epsilon<1$,
\begin{equation}
(\E\ga_2^{2p}(T,\|\cdot\|_{\del}))^{1/p} \lesssim \epsilon^2 + \frac{(d+\log m)\log^2(d/\epsilon)}{m} + \frac{p\log^2(d/\epsilon)\log m}{s}\mu(E)^2  \EquationName{ga2Subs}
\end{equation}
and
\begin{equation}
(\E d^{2p}_{\del}(T))^{1/p} \lesssim \frac{d}{m} + \frac{p}{s}\mu(E)^2. \EquationName{diamSubs}
\end{equation}
As a consequence, if $\eta\leq 1/m$ and
\begin{align}
\nonumber m &\gtrsim ((d+\log m)\min\{\log^2(d/\eps),\log^2(m)\} + d\log(\eta^{-1}))/\eps^2\\
s &\gtrsim (\log(m)\log(\eta^{-1})\min\{\log^2(d/\eps),\log^2(m)\}+\log^2(\eta^{-1}))\mu(E)^2/\eps^2\EquationName{two-ten}
\end{align}
then \Equation{jl-condition} holds with probability at least $1-\eta$.
\end{theorem}
\begin{proof}
By dual Sudakov minoration (Lemma~\ref{lem:dualSudakov} in Appendix~\ref{sec:probTools}) 
\begin{equation*}
\log \mathcal{N}(B_E, \|\cdot\|_{\del}, t) \lesssim \frac{\E_g \|U g\|_{\del}}t, \qquad \mathrm{for \ all} \ t>0, \EquationName{2.1}
\end{equation*}
with $U\in\R^{n\times d}$ having columns $\{f_k\}_{k=1}^d$ forming an orthonormal basis for $E$ and $g$ a random Gaussian vector. Let $U^{(i)}$ be $U$ but where each row $j$ is multiplied by $\delta_{ij}$. Then by \Equation{defDelNorm} and taking $\ell = \log m$,
\begin{align}
\nonumber \E_g \|U g\|_{\del} &= \frac 1{\sqrt{s}} \E_g\max_{1\le i \le m} \Big[\sum_{j=1}^n \delta_{ij} \Big|\sum_{k=1}^d g_k \inprod{f_k, e_j} \Big|^2\Big]^{1/2}\\
\nonumber {}&= \frac 1{\sqrt{s}} \E_g \max_{1\le i \le m} \|U^{(i)} g\|_2\\
\nonumber {}&\le \frac 1{\sqrt{s}}\Big(\max_{1\le i\le m} \E_g \|U^{(i)} g\|_2 + \E_g \max_{1\le i\le m} \Big|\|U^{(i)} g\|_2 - \E_g \|U^{(i)} g\|_2\Big|\Big)\\
\nonumber {}&\le \frac 1{\sqrt{s}}\Big(\max_{1\le i\le m} \|U^{(i)}\|_F + \Big(\sum_{i=1}^m \E_g \Big|\|U^{(i)} g\|_2 - \E_g \|U^{(i)} g\|_2\Big|^\ell\Big)^{1/\ell}\Big)\\
{}&\lesssim \frac 1{\sqrt{s}}\Big(\max_{1\le i\le m} \|U^{(i)}\|_F + \sqrt{\ell} \cdot \max_{1\le i\le m} \|U^{(i)}\|_{\ell_2^d\rightarrow\ell_2^n}\Big) \EquationName{2.2}
\end{align}
where \Equation{2.2} used Gaussian concentration of Lipschitz functions (cf.\ \Equation{GaussLip}), noting that $g\mapsto \|U^{(i)} g\|_2$ is $\|U^{(i)}\|$-Lipschitz.

Since $\|\cdot\|_{\del} \le (1/\sqrt{s})\|\cdot\|_2$, we can use for small $t$ the bound (cf.\ \Equation{volComp})
\begin{align}
\nonumber \log \mathcal{N}(T, \|\cdot\|_{\del}, t) &\le \log \mathcal{N}(B_E, \|\cdot\|_2, t\sqrt{s})\\
&< d \cdot \log\Big(1 + \frac{2}{t\sqrt{s}}\Big) \EquationName{2.3}
\end{align}

Using \Equation{2.3} for small $t$ and noting $\|U^{(i)}\|_F = (\sum_{j=1}^n \delta_{ij}\|P_E e_j\|_2^2)^{1/2}$ for $P_E$ the orthogonal projection onto $E$,  \Equation{Dudley} yields for $t^* = (\epsilon/d)/\log(d/\epsilon)$
\begin{align}
& \gamma_2(T,\|\cdot\|_{\del}) \nonumber \\
& \qquad \lesssim \int_0^{t^*} \sqrt{d}\cdot \Big[\log\Big(2 + \frac{1}{t\sqrt{s}}\Big)\Big]^{1/2} dt + \int_{t^*}^{1/\sqrt{s}} \frac{\E_g \|Ug\|_\del}t dt \nonumber \\
& \qquad \lesssim \sqrt{d}t_*\Big[\log\Big(\frac{1}{t_*\sqrt{s}}\Big)\Big]^{1/2} + \E_g \|Ug\|_\del \log\Big(\frac{1}{t_*\sqrt{s}}\Big) \nonumber \\
{}& \qquad \lesssim \epsilon +  \frac{\log(d/\epsilon)}{\sqrt{s}}\cdot \Big[\max_{1\le i\le m} \Big[\sum_{j=1}^n \delta_{ij} \|P_E e_j\|_2^2 \Big]^{1/2} + \sqrt{\log m}\max_{1\le i\le m} \|U^{(i)}\|_{\ell_2^d\rightarrow\ell_2^n}\Big]
\EquationName{splittingArg}
\end{align}
As a consequence,
\begin{align}
(\E_{\del}\gamma_2^{2p}(T,\|\cdot\|_\delta))^{1/p} & \lesssim \epsilon^2 + \frac{\log^2(d/\epsilon)}{s} \Big[\Big(\E_{\del}\max_{1\le i\le m} \Big[\sum_{j=1}^n \delta_{ij} \|P_E e_j\|_2^2 \Big]^{p}\Big)^{1/p} \nonumber \\
& \qquad \qquad \qquad \qquad + \log(m)\Big(\E_{\del}\max_{1\le i\le m} \|U^{(i)}\|_{\ell_2^d\rightarrow\ell_2^n}^{2p}\Big)^{1/p}\Big]\EquationName{ga2SubsIntermed}
\end{align}
We estimate the first non-trivial term on the right hand side. Since $p\geq \log m$, 
\begin{align*}
\Big(\E_{\del}\max_{1\le i\le m} \Big[\sum_{j=1}^n \delta_{ij} \|P_E e_j\|_2^2 \Big]^{p}\Big)^{1/p} & \le \Big(\sum_{i=1}^m\E_\delta\Big|\sum_{j=1}^n \delta_{ij} \|P_E e_j\|_2^2\Big|^p\Big)^{1/p} \\
& \leq e\cdot\max_{1\leq i\leq m} \Big(\E_{\del}\Big[\sum_{j=1}^n \delta_{ij} \|P_E e_j\|_2^2 \Big]^{p}\Big)^{1/p}
\end{align*}
Note that for fixed $i$, the $(\delta_{ij})_{1\le j\le n}$ are i.i.d.\ variables of mean $s/m$ and therefore
\begin{equation*}
\E\sum_{j=1}^n \del_{ij}\|P_E e_j\|_2^2 = \frac{s}{m} \sum_{j=1}^n \langle P_E e_j,e_j\rangle = \frac{s}{m}\Tr(P_E) = \frac{sd}{m}. 
\end{equation*}
Let $(r_j)_{1\le j\le n}$ be a vector of independent Rademachers. By symmetrization (\Equation{symmetrization}) and Khintchine's inequality (\Equation{KI}), 
\begin{align*}
\Big\|\sum_{j=1}^n \delta_{ij} \|P_E e_j\|_2^2\Big\|_{L^p_{\del}} &\le \frac{ds}m + \Big\|\sum_{j=1}^n (\delta_{ij}-\E\delta_{ij}) \|P_E e_j\|_2^2\Big\|_{L^p_{\del}} \\
{}&\le \frac{ds}m + 2\Big\|\sum_{j=1}^n r_j\delta_{ij}\|P_E e_j\|_2^2\Big\|_{L^p_{\del,r}}\\
{}&\lesssim \frac{ds}m + \sqrt{p}\Big\|\Big(\sum_{j=1}^n\delta_{ij}\|P_E e_j\|_2^4\Big)^{1/2}\Big\|_{L^p_{\del}}\\
{}&\le \frac{ds}m + \sqrt{p}\cdot \max_{1\le j\le n}\|P_E e_j\|_2\cdot \Big\|\sum_{j=1}^n\delta_{ij}\|P_E e_j\|_2^2\Big\|_{L^p_{\del}}^{1/2}
\end{align*}
Solving this quadratic inequality yields
\begin{equation}
\Big\|\sum_{j=1}^n \delta_{ij} \|P_E e_j\|_2^2\Big\|_{L^p_{\del}} \lesssim \frac{ds}m + p\cdot \max_{1\le j\le n}\|P_E e_j\|_2^2
\end{equation}
We conclude that
\begin{equation}
\Big(\E_{\del}\max_{1\le i\le m} \Big[\sum_{j=1}^n \delta_{ij} \|P_E e_j\|_2^2 \Big]^{p}\Big)^{1/p} \lesssim \frac{ds}{m} + p\max_j \|P_E e_j\|_2^2\EquationName{sqrt-trick}
\end{equation}
We now estimate the last term on the right hand side of \Equation{ga2SubsIntermed}. As $p\geq \log m$,
\begin{align*}
\Big(\E_\delta \max_{1\le i\le m} \|U^{(i)}\|_{\ell_2^d\rightarrow\ell_2^n}^{2p}\Big)^{1/p} 
{}&\lesssim \max_{1\leq i\leq m}\Big(\E \|U^{(i)}\|^{2p}_{\ell_2^d\rightarrow\ell_2^n}\Big)^{1/p}\\
{}& = \max_{1\leq i\leq m}\Big(\E \|(U^{(i)})^*U^{(i)}\|^{p}_{\ell_2^d\rightarrow\ell_2^d}\Big)^{1/p}.
\end{align*}
If $u_j$ denotes the $j$th row of $U$, then 
$$
(U^{(i)})^*U^{(i)} = \sum_{j=1}^n \delta_{ij} u_j u_j^*.
$$
By symmetrization and the non-commutative Khintchine inequality (cf.\ \Equation{NCKI}),
\allowdisplaybreaks
\begin{align*}
\Big\|\sum_{j=1}^n \delta_{ij} u_j u_j^*\Big\|_{L^p_{\del}} &\le \frac sm + \Big\|\sum_{j=1}^n \delta_{ij} u_j u_j^* - \frac sm I\Big\|_{L^p_{\del}}\\
&\le \frac sm + 2\Big\|\sum_{j=1}^n r_j\delta_{ij} u_j u_j^*\Big\|_{L^p_{\del,r}}\\
&\lesssim \frac sm + \sqrt{p} \Big\|\sum_{j=1}^n \delta_{ij} \|u_j\|_2^2 u_j u_j^*\Big\|_{L^{p/2}_{\del}}^{1/2}\\
&\le \frac sm + \sqrt{p} \max_{1\le j\le n} \|P_E e_j\|_2\cdot \Big\|\sum_{j=1}^n \delta_{ij} u_j u_j^*\Big\|_{L^{p}_{\del}}^{1/2}
\end{align*}
Solving this quadratic inequality, we find
\begin{equation*}
\Big\|\sum_{j=1}^n \delta_{ij} u_j u_j^*\Big\|_{L^p_{\del}} \lesssim \frac s m + p \cdot \max_{1\le j\le n} \|P_E e_j\|_2^2
\end{equation*}
and as a consequence,
\begin{equation*}
\Big(\E_\delta \max_{1\le i\le m} \|U^{(i)}\|_{\ell_2^d\rightarrow\ell_2^n}^{2p}\Big)^{1/p} \lesssim \frac s m + p \cdot \max_{1\le j\le n} \|P_E e_j\|_2^2
\end{equation*}
Applying this estimate together with \Equation{sqrt-trick} in \Equation{ga2SubsIntermed} yields the asserted estimate in \Equation{ga2Subs}. For the second assertion, note by Cauchy-Schwarz that
\begin{align*}
d_{\del}(T) & = \frac{1}{\sqrt{s}} \max_{1\le i\le m} \sup_{x\in T} \Big[\sum_{j=1}^n \delta_{ij} x_j^2\Big]^{1/2} \\
& = \frac{1}{\sqrt{s}}  \max_{1\le i\le m} \sup_{x\in T} \Big[\sum_{j=1}^n \delta_{ij} \Big(\sum_{k=1}^d \inprod{x, f_k}\inprod{f_k, e_j}\Big)^2\Big]^{1/2}\\
& \le \frac{1}{\sqrt{s}} \max_{1\le i\le m} \Big[\sum_{j=1}^n \delta_{ij} \|P_E e_j\|_2^2\Big]^{1/2}.
\end{align*}
Therefore,
$$(\E_{\del}d^{2p}_{\del}(T))^{1/p} \leq \frac{1}{s}\Big(\E_{\del}\max_{1\le i\le m} \Big[\sum_{j=1}^n \delta_{ij} \|P_E e_j\|_2^2 \Big]^{p}\Big)^{1/p}$$
and \Equation{diamSubs} immediately follows from \Equation{sqrt-trick}.\par
To prove the final assertion, we combine \Eqsub{RIPLpEst}, \Equation{ga2Subs} and \Equation{diamSubs} to obtain
\begin{align*}
& \Big(\E_{\del,\si}\sup_{x \in T} \Big|\|\Phi x\|^2 - \|x\|^2\Big|^p\Big)^{1/p} \\
& \qquad \lesssim \epsilon^2 + \frac{(d+\log m)\log^2(d/\epsilon)}{m} + \frac{p\log^2(d/\epsilon)\log m}{s}\max_{1\leq j\leq n} \|P_E e_j\|_2^2 \\
& \qquad \qquad \ + \Big[\epsilon^2 + \frac{(d+\log m)\log^2(d/\epsilon)}{m} + \frac{p\log^2(d/\epsilon)\log m}{s}\max_{1\leq j\leq n} \|P_E e_j\|_2^2\Big]^{1/2} \\
& \qquad \qquad \ + \frac{pd}{m} + \frac{p^2}{s}\max_{1\leq j\leq n} \|P_E e_j\|_2^2 + \Big[\frac{pd}{m} + \frac{p^2}{s}\max_{1\leq j\leq n} \|P_E e_j\|_2^2\Big]^{1/2}. 
\end{align*}
Now apply Lemma~\ref{lem:MomentsToTails} of Appendix~\ref{sec:probTools} to obtain, for any $w\geq \log m$,
\begin{align*}
\bP_{\del,\si}\Big(\eps_{\Phi,T} & \gtrsim \epsilon^2 + \frac{(d+\log m)\log^2(d/\epsilon)}{m} + \frac{w\log^2(d/\epsilon)\log m}{s}\max_{1\leq j\leq n} \|P_E e_j\|_2^2 \\
& \ \ \ + \Big[\epsilon^2 + \frac{(d+\log m)\log^2(d/\epsilon)}{m} + \frac{w\log^2(d/\epsilon)\log m}{s}\max_{1\leq j\leq n} \|P_E e_j\|_2^2\Big]^{1/2} \\
& \ \ \ + \frac{wd}{m} + \frac{w^2}{s}\max_{1\leq j\leq n} \|P_E e_j\|_2^2 + \Big[\frac{wd}{m} + \frac{w^2}{s}\max_{1\leq j\leq n} \|P_E e_j\|_2^2\Big]^{1/2}\Big) \leq e^{-w}. 
\end{align*}
Now set $w=\log(\eta^{-1})$, choose $\epsilon=\eps/C$ (with $C$ a large enough absolute constant) and $\epsilon=d/m$, and use the assumptions on $m$ and $s$ to conclude $\bP(\eps_{\Phi,T}\geq \eps)\leq \eta$.
\end{proof}
Theorem~\ref{thm:JLsubspace} recovers a similar result in \cite{NN13b} but via a different method, less logarithmic factors in the setting of $m$, and the revelation that $s$ can be taken smaller if all the leverage scores $\|P_E e_j\|_2$ are small (note if $\|P_E e_j\|_2 \ll (\log d\cdot \log m)^{-1}$ for all $j$, we may take $s=1$, though this is not true in general). Our dependence on $1/\eps$ in $s$ is quadratic instead of the linear dependence in \cite{NN13b}, but as stated earlier, in most applications of OSE's $\eps$ is taken a constant.

\begin{remark}
\label{rem:coherence}
\textup{
As stated, one consequence of the above is we can set $s=1$, and $m$ to nearly linear in $d$ as long as $\mu$ is at most inverse polylogarithmic in $d$. It is worth noting that if $d \gtrsim \log n$, then a random $d$-dimensional subspace $E$ has $\mu(E) \lesssim  \sqrt{d/n}$ by the JL lemma, and thus most subspaces do have even much lower coherence. Note that the latter bound is sharp. Indeed, let $E$ be any $d$-dimensional subspace and let $U\in\R^{n\times d}$ have orthonormal columns such that $E$ is the column space of $U$. Then
$$\max_{v\in E, \|v\|_2=1} \max_{1\le j\le n}|\inprod{v,e_j}|$$
is the maximum $\ell_2$ norm of a row $u_j$ of $U$. Since $U$ has $n$ rows and $\|U\|_F^2 = d$, there exists a $j$ such that $\|u_j\|_2\ge \sqrt{d/n}$ and in particular $\mu(E)\geq\sqrt{d/n}$.
}

\textup{
It is also worth noting that \cite[Theorem 3.3]{AMT10} observed that if $H\in\R^{n\times n}$ is a bounded orthonormal system (i.e.\ $H$ is orthogonal and $\max_{i,j} |H_{ij}| = O(1/\sqrt{n})$, and $D$ is a diagonal matrix with independent Rademacher entries, then $HDE$ has incoherence $\mu\lesssim \sqrt{d(\log n)/n}$ with probability $1 - 1/\mathrm{poly}(n)$. They then showed that incoherence $\mu$ implies that a sampling matrix $S\in\R^{m\times n}$ suffices to achieve \Equation{rjl-condition} with $m\lesssim n\mu^2\log m / \eps^2$. Putting these two facts together implies $SHD$ gives \Equation{rjl-condition} with $m\lesssim d(\log n)(\log m)/\eps^2$. \Equation{two-ten} combined with \cite[Theorem 3.3]{AMT10} implies a statement of a similar flavor but with an arguably stronger implication for applications: $\Phi HD$ with $s=1$ satisfies \Equation{rjl-condition} for $T = E\cap S^{n-1}$ without increasing $m$ as long as $n \gtrsim d(\log(d/\eps))^c/\eps^2$.
}

\end{remark}

\section{The type-$2$ case}\SectionName{typetwo}
Recall that the \emph{type-$2$ constant} $T_2(\|\cdot\|)$ of a semi-norm $\|\cdot \|$ is defined as the best constant (if it exists) such that the inequality
\begin{equation}
\E_g\Big\|\sum_\alpha g_\alpha x_\alpha\Big\| \le T_2\Big(\sum_\alpha \|x_\alpha\|^2\Big)^{1/2} \EquationName{3.1}
\end{equation}
holds, for all finite systems of vectors $\{x_{\alpha}\}$ and i.i.d.\ standard Gaussian $(g_\alpha)$.

Given $T\subset S^{n-1}$, define the semi-norm
$$
\tnorm{x}_T = \sup_{y\in T} |\inprod{x,y}|
$$
It will be convenient to use the following notation. For $i=1,\ldots,m$ we let $J_i$ be the set of `active' indices, i.e., $J_i = \{j=1,\ldots,n : \delta_{ij} = 1\}$. We can then write
\begin{equation}
\vertiii{x} = \frac 1{\sqrt{s}} \max_{1\le i\le m} \|P_{J_i} x\|_2, \EquationName{3.2}
\end{equation}
where we abuse notation and let $P_{J_i}$ denote orthogonal projection onto the coordinate subspace specified by $J_i$.
\begin{lemma}
\label{lem:BoN410}
Let $\|\cdot\|$ be any semi-norm satisfying $\tnorm{x}_T\leq \|x\|$ for all $x\in \R^n$. Then,
\begin{align*}
\ga_2(T,\|\cdot\|_{\del}) & \lesssim \frac{1}{\sqrt{s}} T_2(\|\cdot\|)(\log n)^3(\log m)^{1/2} \\
& \qquad \qquad \qquad \times \Big(\max_{1\leq i\leq m} \E_g\Big\|\sum_{j=1}^n g_j\del_{ij} e_j\Big\| + (\log m)^{1/2}d_{\|\cdot\|}\Big(\bigcup_{i=1}^m B_{J_i}\Big)\Big),
\end{align*}
where $T_2(\|\cdot\|)$ is the type-$2$ constant and for $1\leq i\leq m$,
$$B_{J_i} = \Big\{\sum_{j\in J_i} x_je_j : \sum_{j\in J_i} x_j^2 \le 1\Big\}.$$
\end{lemma}
Note also that $d_{\del}(T)$ can be bounded as
\begin{align}
\EquationName{delTermBd}
d_{\del}(T) & = \frac{1}{\sqrt{s}} \max_{1\leq i\leq m}\sup_{x\in T} \Big(\sum_{j=1}^n \del_{ij}\langle x,e_j\rangle^2\Big)^{1/2} \nonumber \\
& = \frac{1}{\sqrt{s}} \max_{1\leq i\leq m}\sup_{x\in T} \Big(\E_g\Big|\sum_{j=1}^n \del_{ij}g_j\langle x,e_j\rangle\Big|^2\Big)^{1/2} \nonumber \\
& \lesssim \frac{1}{\sqrt{s}} \max_{1\leq i\leq m}\sup_{x\in T} \E_g\Big|\Big\langle x,\sum_{j=1}^n \del_{ij}g_je_j\Big\rangle\Big|
 \leq \frac{1}{\sqrt{s}} \max_{1\leq i\leq m}\E_g\Big\|\sum_{j=1}^n \del_{ij}g_je_j\Big\|.
\end{align}
\begin{remark}
\textup{
Replacing $\tnorm{\cdot}_T$ by a stronger norm $\|\cdot\| \ge \tnorm{\cdot}_T$ in this lemma may reduce the type-$2$ constant. The application to $k$-sparse vectors will illustrate this.
}
\end{remark}
\begin{proof}[Proof of Lemma~\ref{lem:BoN410}]
By \Equation{Dudley}, to bound $\ga_2(T,\|\cdot\|_{\del})$ it suffices to evaluate the covering numbers $\mathcal{N}(T, \|\cdot\|_{\del}, t)$ for $t<d_{\del}^*(T)/\sqrt{s}$, where we set 
$$d_{\del}^*(T) = \sqrt{s}d_{\del}(T) = \sup_{x\in T}\max_{1\le i\le m} \|P_{J_i} x\|_2.$$
For large values of $t$ we use duality for covering numbers. It will be convenient to work with $\tilde{T} = \mathrm{conv}(T\cup (-T))$, which is closed, bounded, convex and symmetric. Clearly,
$$\mathcal{N}(T, \|\cdot\|_{\del}, t)\leq \mathcal{N}(\tilde{T}, \|\cdot\|_{\del}, t)$$
and in the notation of Appendix~\ref{sec:probTools},
$$\tnorm{x}_T = \sup_{y\in T\cup\{-T\}}\inprod{x,y} = \sup_{y\in \tilde{T}}\inprod{x,y} = \|x\|_{\tilde{T}^{\circ}}.$$
Here $\tilde{T}^{\circ}$ denotes the polar (see Appendix~\ref{sec:convex-tools} for a definition). Lemma~\ref{lem:BPST}, applied with $U=\tilde{T}$, $V=\cap_{i=1}^m \{x\in \R^n \ : \ s^{-1/2}\|P_{J_i}x\|_2\leq 1\}$, $\eps=t$ and $\theta=d_{\del}^*(T)/\sqrt{s}$, implies 
$$
\log \mathcal{N}(\tilde{T}, \|\cdot\|_{\del}, t) \le \Big(\log\frac {d_{\del}^*(T)}{t\sqrt{s}}\Big) A \log \mathcal{N}\Big(\frac 1{\sqrt{s}}\mathrm{conv}\Big(\bigcup_{i=1}^m B_{J_i}\Big), \tnorm{\cdot}_T, \frac t8\Big)
$$
where $A \lesssim T_2(\|\cdot\|_{\del})^2 \lesssim \log m$ (by \Eqsub{3.2}) and we used that
$$B_{J_i} = \Big\{\sum_{j\in J_i} x_je_j : \sum_{j\in J_i} x_j^2 \le 1\Big\} = \{x\in \R^n \ : \ \|P_{J_i}x\|_2\leq 1\}^{\circ}.$$
Hence,
\begin{equation}
\log\mathcal{N}(\tilde{T}, \|\cdot\|_{\del}, t) \lesssim \Big(\log\frac {d_{\del}^*(T)}{t\sqrt{s}}\Big)(\log m)\log\mathcal{N}\Big(\mathrm{conv}\Big(\bigcup_{i=1}^m B_{J_i}\Big), \|\cdot\|, \frac 18\sqrt{s} t\Big) \EquationName{3.3}
\end{equation}
To estimate the right hand side we use the following consequence of Maurey's lemma (Lemma~\ref{lem:Maurey} of Appendix~\ref{sec:probTools}).
\begin{claim}\ClaimName{claim}
Fix $0<R<\infty$. Let $\|\cdot\|$ be a semi-norm on $\R^n$ and $\Omega\subset \{x\in \R^n : \|x\| \le R\}$. Then for $0 < \epsilon < R$,
\begin{equation}
\epsilon [\log \mathcal{N}(\mathrm{conv}(\Omega), \|\cdot\|, \epsilon)]^{1/2} \lesssim \Big(\log \frac {R}{\epsilon}\Big)^{3/2} T_2(\|\cdot\|) \max_{\epsilon\leq \epsilon'\leq R} \epsilon'[\log\mathcal{N}(\Omega, \|\cdot\|, \epsilon')]^{1/2} \EquationName{3.4}
\end{equation}
\end{claim}
\begin{proof}
It suffices to prove the result for $R=1$, the general case then follows by considering $R^{-1}\Omega$. For each positive integer $k$, let $\Omega_k\subset \Omega$ be a finite set satisfying
$$
|\Omega_k| \le \mathcal{N}(\Omega, \|\cdot\|, 2^{-k})
$$
and
$$
\inf_{y\in \Omega_k} \|x - y\| < 2^{-k}\hbox{ for all } x\in\Omega
$$
Given $x\in\Omega$, denote $x_k\in\Omega_k$ a sequence s.t.\ $\|x - x_k\| < 2^{-k}$. Then, setting $y_0 = x_0$, $y_k = x_k - x_{k-1}$, we obtain
\begin{equation}
\Big\|x - \sum_{\epsilon/2<2^{-k}\leq 1} y_k\Big\| < \epsilon \EquationName{3.5}
\end{equation}

\begin{equation}
y_k \in \Omega_k - \Omega_{k-1}\hbox{ and } \|y_k\| < 2^{-k} + 2^{-k + 1} = 3\cdot 2^{-k} \EquationName{3.6}
\end{equation}

It follows from \Equation{3.5} that
$$
\mathrm{conv}(\Omega) \subset \sum_{\epsilon/2<2^{-k}\leq 1} \mathrm{conv}(\tilde{\Omega}_k) + B_{\|\cdot\|}(0, \epsilon)
$$
with
\begin{equation}
\tilde{\Omega}_k = \Big\{y\in \Omega_k - \Omega_{k-1} : \|y\| < 3\cdot 2^{-k} \Big\} \EquationName{3.7}
\end{equation}
By Maurey's lemma, given $z\in\mathrm{conv}(\tilde{\Omega}_k)$ and a positive integer $\ell_k$, there are points $y_1,\ldots,y_\ell\in\tilde{\Omega}_k$ such that
\begin{equation}
\Big\|z - \frac 1{\ell_k}(y_1 + \ldots + y_{\ell_k})\Big\| \lesssim T_2(\|\cdot\|)\frac 1{\sqrt{\ell_k}} \cdot 3\cdot 2^{-k} \EquationName{3.8}
\end{equation}
Taking $\ell_k\simeq T_2(\|\cdot\|)^2 4^{-k} \epsilon^{-2} \log^2(1/\epsilon)$, we obtain a point $z_k\in \Omega_k' = \frac 1{\ell_k}(\tilde{\Omega}_k + \ldots + \tilde{\Omega}_k)$ such that
$$
\|z - z_k\| < \frac{\epsilon}{2\log(1/\epsilon)}
$$
Moreover
$$
\log|\Omega_k'| \le \ell_k \log|\tilde{\Omega}_k| \lesssim T_2(\|\cdot\|)^2 4^{-k} \epsilon^{-2}\log^2(1/\epsilon) \log \mathcal{N}(\Omega, \|\cdot\|, 2^{-k})
$$
Since $|\{k \ : \ \epsilon/2<2^{-k}\leq 1\}|\leq \log(2 /\epsilon)$, we obtain from the preceding,
$$
\log\mathcal{N}(\mathrm{conv}(\Omega), \|\cdot\|, \epsilon) \lesssim \log^2(1/\epsilon) T_2(\|\cdot\|)^2 \epsilon^{-2} \sum_{\epsilon/2<2^{-k}\leq 1} 4^{-k} \log\mathcal{N}(\Omega, \|\cdot\|, 2^{-k})
$$
and \Equation{3.4} follows. This proves the claim.
\end{proof}
Observe that
\begin{align*}
d_{\del}^*(T) & = \sup_{x\in T} \max_{1\leq i\leq m} \|P_{J_i}x\|_2 \\
& = \sup_{x\in T}\sup_{y\in B_{\ell_2^n}} \max_{1\leq i\leq m} \langle x,P_{J_i}y\rangle = d_{\tnorm{\cdot}_T}\Big(\bigcup_{i=1}^m B_{J_i}\Big) \leq d_{\|\cdot\|}\Big(\bigcup_{i=1}^m B_{J_i}\Big) 
\end{align*}
Write $d_{\|\cdot\|}:=d_{\|\cdot\|}(\cup_{i=1}^m B_{J_i})$ for brevity. Applying \Claim{claim} with $\Omega = \cup_{i=1}^m B_{J_i}$, $\epsilon=t\sqrt{s}$, and the dominating semi-norm $\|\cdot\|\geq\tnorm{\cdot}_T$ to the right hand side of \Equation{3.3}, we find
\begin{align}
\nonumber t\sqrt{s}\Big[\log\mathcal{N}(\tilde{T}, \|\cdot\|_{\del}, t)\Big]^{1/2} \le &\Big(\log \frac {d_{\|\cdot\|}}{t\sqrt{s}}\Big)^2(\log m)^{1/2} T_2(\|\cdot\|)\cdot\\
{}&\Big[\max_{t\sqrt{s}\leq t'\leq d_{\|\cdot\|}} t'\Big[\log \mathcal{N}\Big(\bigcup_{i=1}^m B_{J_i}, \|\cdot\|, t'\Big)\Big]^{1/2}\Big]. \EquationName{3.9}
\end{align}
Using the dual Sudakov inequality (Lemma~\ref{lem:dualSudakov} of Appendix~\ref{sec:probTools}) we arrive at
\begin{align}
\nonumber \Big[\log\mathcal{N}(\tilde{T}, \|\cdot\|_{\del}, t)\Big]^{1/2} \le &\frac {1}{t\sqrt{s}}\Big(\log \frac {d_{\|\cdot\|}}{t\sqrt{s}}\Big)^2(\log m)^{1/2} T_2(\|\cdot\|)\cdot\\
{}&\Big[\max_{1\le i\le m} \E_g\Big\|\sum_{j\in J_i} g_j e_j\Big\| + d_{\|\cdot\|}(\log m)^{1/2}\Big]. \EquationName{dualsud}
\end{align}
We now apply \Equation{dualsud} for $(sn)^{-1/2}d_{\|\cdot\|}\leq t<s^{-1/2}d_{\|\cdot\|}$ and the estimate (cf.\ \Equation{volComp})
$$\cN(T,\|\cdot\|_{\del},t)\leq \cN(d_{\del}(T)B_{\|\cdot\|_{\del}},\|\cdot\|_{\del},t)\leq \Big(1+\frac{2d_{\del}^*(T)}{t\sqrt{s}}\Big)^n \leq \Big(1+\frac{2d_{\|\cdot\|}}{t\sqrt{s}}\Big)^n$$
for $0<t<(sn)^{-1/2}d_{\|\cdot\|}$, respectively, in \Equation{Dudley}. A straightforward computation, similar to \Equation{splittingArg}, yields the result.
\end{proof}

\subsection{Application to $k$-sparse vectors}\SectionName{ksparse}
Let $\|x\|_0$ denote the number of non-zero entries of $x$ and set
$$ 
S_{n,k} = \{x\in\R^n : \|x\|_2 = 1, \|x\|_0 \le k \}.
$$
\begin{theorem}
Let $A\in\R^{n\times n}$ be orthogonal and denote
$$
T = A(S_{n,k})
$$Let $\eta\leq 1/m$. 
If 
\begin{equation}
\max |A_{ij}| < (k\log n)^{-1/2} \EquationName{4.8}
\end{equation}
and
\begin{align}
\nonumber m &\gtrsim k(\log n)^8(\log m) / \eps^2 \\
s &\gtrsim (\log n)^7 (\log m)(\log \eta^{-1}) /\eps^2, \EquationName{4.7}
\end{align}
then with probability at least $1-\eta$ we have
$$(1-\eps)\|x\|_2^2 \leq \|\Phi x\|_2^2 \leq (1+\eps)\|x\|_2^2 \qquad \mathrm{for \ all} \ x\in T.$$
\end{theorem}
\begin{proof}
We apply \Equation{RIPLpEst} and estimate $(\E_{\del}\ga_2^{2p}(T,\|\cdot\|_{\del}))^{1/p}$ and $\E_{\del} d_{\del}^{2p}(T)$ for $p\geq \log m$. We trivially estimate $\E_{\del} d_{\del}^{2p}(T)\leq 1/s$. To bound the $\ga_2$-functional, note that $T\subset K$, where
\begin{equation}
K = B_{\ell_2^n} \cap \sqrt{k} A(B_{\ell_1^n}). \EquationName{4.1}
\end{equation}
The polar body of $K$ is given by
$$
K^\circ = B_{\ell_2^n} + \frac 1{\sqrt{k}} A(B_{\ell_{\infty}^n})
$$
and one can readily calculate that $T_2(\tnorm{\cdot}_K) \lesssim \sqrt{\log n}$ and $d_{\tnorm{\cdot}_K}(\cup_{i=1}^m B_{J_i})\leq 1$. We apply Lemma~\ref{lem:BoN410} with $\|\cdot\|=\tnorm{\cdot}_K$. Note that for every $1\leq i\leq m$,
$$\E_g\btnorm{\sum_j g_j\del_{ij}e_k}_K \leq \sqrt{k}\E_g\max_{1\leq \ell\leq n}\Big|\sum_j g_j\del_{ij} A_{j\ell}\Big| \leq \sqrt{k\log n} \max_{1\leq\ell\leq n}\Big(\sum_j \del_{ij}A_{j\ell}^2\Big)^{1/2}$$
and therefore the lemma implies that
\begin{align*}
& \ga_2(T,\|\cdot\|_{\del}) \\
& \qquad \leq \frac{1}{\sqrt{s}}(\log n)^{7/2}\log m + \sqrt{\frac{k}{s}}(\log n)^4(\log m)^{1/2} \max_{1\leq i\leq m,1\leq \ell\leq n} \Big(\sum_j \del_{ij}A_{j\ell}^2\Big)^{1/2}.
\end{align*}
As a consequence, as $p\geq \log m$,
\begin{align}
\EquationName{ga2pSparse}
& (\E_{\del}\ga_2^{2p}(T,\|\cdot\|_{\del}))^{1/p} \nonumber\\
& \ \ \ \leq \frac{1}{s}(\log n)^{7}(\log m)^2 + \frac{k}{s}(\log n)^8(\log m) \max_{1\leq i\leq m}\Big(\E_{\del}\max_{1\leq \ell\leq n} \Big(\sum_j \del_{ij}A_{j\ell}^2\Big)^{p}\Big)^{1/p}.
\end{align}
Since $\E\del_{ij}=s/m$ and $A$ is orthogonal, we obtain using symmetrization and Khintchine's inequality,
\begin{align*}
& \Big(\E_{\del}\max_{1\leq \ell\leq n} \Big(\sum_j \del_{ij}A_{j\ell}^2\Big)^{p}\Big)^{1/p} \\
& \qquad \leq \frac{s}{m} + \Big(\E_r\E_{\del}\max_{1\leq \ell\leq n} \Big(\sum_j r_j\del_{ij}A_{j\ell}^2\Big)^{p}\Big)^{1/p} \\
& \qquad \leq \frac{s}{m} + \sqrt{p}\Big(\E_{\del}\max_{1\leq \ell\leq n} \Big(\sum_j \del_{ij}A_{j\ell}^4\Big)^{p/2}\Big)^{1/p} \\
& \qquad \leq \frac{s}{m} + \sqrt{p}\max_{j,\ell}|A_{j\ell}|\Big(\E_{\del}\max_{1\leq \ell\leq n} \Big(\sum_j \del_{ij}A_{j\ell}^2\Big)^{p}\Big)^{1/(2p)}.
\end{align*}
By solving this quadratic inequality, we find
$$\Big(\E_{\del}\max_{1\leq \ell\leq n} \Big(\sum_j \del_{ij}A_{j\ell}^2\Big)^{p}\Big)^{1/p} \lesssim \frac{s}{m} + p\max_{j,\ell}A_{j\ell}^2 \leq \frac{s}{m} + \frac{p}{k\log n}.$$
Combine this bound with \Equation{ga2pSparse} and \Equation{RIPLpEst} to arrive at
\begin{align*}
& \Big(\E_{\del,\si}\sup_{x \in T} \Big|\|\Phi x\|_2^2 - \|x\|_2^2\Big|^p\Big)^{1/p} \\
& \ \ \lesssim  \frac{1}{s}(\log n)^{7}(\log m)^2 + \frac{k}{m}(\log n)^8(\log m) + \frac{p}{s}(\log n)^7(\log m) \\
& \ \ \ \ \ \ + \Big(\frac{1}{s}(\log n)^{7}(\log m)^2 + \frac{k}{m}(\log n)^8(\log m) + \frac{p}{s}(\log n)^7(\log m)\Big)^{1/2} \\
& \ \ \ \ \ \ \ \ \qquad \qquad + \sqrt{\frac{p}{s}} + \frac{p}{s}.
\end{align*}
Since $p\geq \log m$ was arbitrary, the result now follows from Lemma~\ref{lem:MomentsToTails}.
\end{proof}

\section{Sketching constrained least squares programs}\SectionName{cls}

Let us now apply the previous results to constrained least squares minimization. We refer to Appendix~\ref{sec:convex-tools} for any unexplained terminology. Consider $A\in \R^{n\ti d}$ and a sketching matrix $\Phi\in \R^{m\ti n}$. Define $f(x) = \|Ax - b\|_2^2$ and $g(x) = \|\Phi Ax - \Phi b\|_2^2$. Let $\cC\subset \R^d$ be any closed convex set. Let $x_*$ be a minimizer of the constrained least squares program
\begin{equation}
\label{eqn:CLS}
\min f(x) \qquad \mathrm{subject \ to} \qquad x\in \cC
\end{equation}
and let $\hat{x}$ be a minimizer of the associated sketched program
\begin{equation}
\label{eqn:SCLS}
\min g(x) \qquad \mathrm{subject \ to} \qquad x\in \cC.
\end{equation}
Before giving our analysis, we define two quantities introduced in \cite{PiW14}. Given $\cK\subset\R^d$ and $u\in S^{n-1}$ we set
\begin{align*}
Z_1(A,\Phi,\cK) & = \inf_{v\in A\cK\cap S^{n-1}} \|\Phi v\|_2^2 \\
Z_2(A,\Phi,\cK,u) & = \sup_{v \in A\cK\cap S^{n-1}} |\langle \Phi u, \Phi v\rangle - \langle u,v\rangle|.
\end{align*}
We denote the tangent cone of $\cC$ at a point $x$ by $T_{\cC}(x)$ (see Appendix~\ref{sec:convex-tools} for a definition). The first statement in the following lemma is \cite[Lemma 1]{PiW14}. For the convenience of the reader we provide a proof, which is in essence the same as in \cite{PiW14}, with the second case (when $x_*$ is a global minimizer) being a slight modification of the proof there. In what follows we assume $Ax_*\neq b$, since otherwise $f(\hat{x}) = f(x_*) = 0$.
\begin{lemma}
\label{lem:compare}
Define $u=(Ax_*-b)/\|Ax_*-b\|_2$ and set $Z_1=Z_1(A,\Phi,T_{\cC}(x_*))$ and $Z_2=Z_2(A,\Phi,T_{\cC}(x_*),u)$. Then,
$$f(\hat{x})\leq \Big(1+\frac{Z_2}{Z_1}\Big)^2f(x_*).$$
If $x_*$ is a global minimizer of $f$, then
$$f(\hat{x})\leq \Big(1+\frac{Z_2^2}{Z_1^2}\Big)f(x_*).$$
\end{lemma}
\begin{proof}
Set $e=\hat{x}-x_* \in T_{\cC}(x_*)$ and $w=b-Ax_*$. By optimality of $x_*$, for any $x\in \cC$,
$$\langle\nabla f(x_*),x-x_*\rangle = \langle Ax_*-b,A(x-x_*)\rangle\geq 0.$$
In particular, taking $x=\hat{x}$ gives $\langle w,Ae\rangle\leq 0$. As a consequence,
\begin{equation}
\label{eqn:compareZ2}
\Big\langle \frac{w}{\|w\|_2},\frac{\Phi^*\Phi Ae}{\|Ae\|_2}\Big\rangle \leq \Big\langle \frac{w}{\|w\|_2},\frac{\Phi^*\Phi Ae}{\|Ae\|_2}\Big\rangle - 
\Big\langle \frac{w}{\|w\|_2},\frac{Ae}{\|Ae\|_2}\Big\rangle \leq Z_2.
\end{equation}
By optimality of $\hat{x}$, for any $x\in \cC$,
$$\langle\nabla g(\hat{x}),x-\hat{x}\rangle = \langle \Phi A\hat{x}-\Phi b,\Phi A(x-\hat{x})\rangle\geq 0.$$
Taking $x=x_*$ yields
$$0\geq \langle \Phi A\hat{x} - \Phi b,\Phi Ae\rangle = \langle -\Phi w,\Phi Ae\rangle + \|\Phi Ae\|_2^2.$$
Therefore, by (\ref{eqn:compareZ2}),
$$Z_1\|Ae\|_2^2 \leq \|\Phi Ae\|_2^2 \leq \langle \Phi w,\Phi Ae\rangle \leq Z_2\|w\|_2\|Ae\|_2.$$
In other words,
$$\|Ae\|_2\leq \frac{Z_2}{Z_1}\|w\|_2$$
and by Cauchy-Schwarz,
$$\|A\hat{x}-b\|_2^2 = \|w\|_2^2 + \|Ae\|_2^2 - 2\langle w,Ae\rangle \leq \|w\|_2^2\Big(1+\frac{Z_2}{Z_1}\Big)^2.$$
If $x_*$ is a global minimizer of $f$, then $\nabla f(x_*)=0$ and therefore $\langle w,Ae\rangle=0$ (a similar observation to \cite[Theorem 12]{Sarlos06}). This yields the second assertion.
\end{proof}
Clearly, if $\Phi$ satisfies \Eqsub{jl-condition} for $T=A T_{\cC}(x_*)\cap S^{n-1}$ then $Z_1\geq 1-\eps$. We do not immediately obtain an upper bound for $Z_2$, however, as $u$ is in general not in $A T_{\cC}(x_*)\cap S^{n-1}$. We mend this using the observation in Lemma~\ref{lem:Z2boundGen}. For its proof we recall a chaining result from \cite{Dir13}. Let $(T,\rho)$ be a semi-metric space. Recall that a real-valued process $(X_t)_{t\in T}$ is called subgaussian if for all $s,t\in T$,
\begin{equation*}
\bP(|X_t - X_s|\geq w\rho(t,s)) \leq 2\exp(-w^2) \qquad (w\geq 0).
\end{equation*}
The following result is a special case of \cite[Theorem 3.2]{Dir13}.
\begin{lemma}
\label{lem:supSubg}
If $(X_t)_{t\in T}$ is subgaussian, then for any $t_0\in T$ and $1\leq p<\infty$,
\begin{equation*}
\Big(\E\sup_{t\in T}|X_t - X_{t_0}|^p\Big)^{1/p} \lesssim \ga_2(T,\rho) + \sqrt{p}d_{\rho}(T).
\end{equation*}
\end{lemma}
In particular, if $T$ contains $n$ elements, then
\begin{equation*}
\Big(\E\sup_{t\in T}|X_t - X_{t_0}|^p\Big)^{1/p} \lesssim (\sqrt{p} + (\log n)^{1/2})d_{\rho}(T).
\end{equation*}
\begin{lemma}
\label{lem:Z2boundGen}
Fix $u\in B_{\ell_2^n}$, $T\subset \R^n$ and let $\Phi$ be the SJLT. Set
$$Z = \sup_{v\in T}|\langle \Phi u,\Phi v\rangle - \langle u,v\rangle|.$$
For any $p\geq 1$,
$$(\E_{\del,\si}Z^p)^{1/p} \lesssim \Big(\sqrt{\frac{p}{s}} + 1\Big)\Big((\E_{\del}\ga_2^p(T,\|\cdot\|_{\del}))^{1/p} + \sqrt{p}(\E_{\del}d_{\del}^p(T))^{1/p}\Big).$$
\end{lemma}
\begin{proof}
With the notation (\ref{eqn:AdelxDef}) we have $\Phi x=A_{\del,x}\si$ for any $x\in \R^n$. Since $\E_{\si} \Phi^*\Phi=I$ we find
\begin{align*}
Z = \sup_{v\in T}|\langle u,(\Phi^*\Phi-I)v\rangle| = \sup_{v\in T} |\si^* A_{\del,u}^*A_{\del,v}\si - \E_{\si}(\si^* A_{\del,u}^*A_{\del,v}\si)|.
\end{align*}
Let $\si'$ be an independent copy of $\si$. By decoupling (see (\ref{eqn:decRad})),
$$(\E_{\si}Z^p)^{1/p} \leq 4 (\E_{\si,\si'}\sup_{v\in T} |\si^* A_{\del,u}^*A_{\del,v}\si'|^p)^{1/p}.$$
For any $v_1,v_2\in T$, Hoeffding's inequality implies that for all $w\geq 0$
\begin{align*}
& \bP_{\si'}(|\si^* A_{\del,u}^*A_{\del,v_1}\si' - \si^* A_{\del,u}^*A_{\del,v_2}\si'| \\
& \qquad \qquad \qquad \qquad \geq \sqrt{w}\|\si^*A_{\del,u}^*\|_2 \ \|A_{\del,v_1}-A_{\del,v_2}\|_{\ell_2\to\ell_2}) \leq 2e^{-w}
\end{align*} 
and therefore $v\mapsto \si^* A_{\del,u}^*A_{\del,v}\si'$ is a subgaussian process. By Lemma~\ref{lem:supSubg} (and (\ref{eqn:AdelxNorm})),
\begin{align*}
& (\E_{\si'}\sup_{v\in T} |\si^* A_{\del,u}^*A_{\del,v}\si'|^p)^{1/p} \\
& \qquad \lesssim \|A_{\del,u}\si\|_2 \ga_2(T,\|\cdot\|_{\del}) + \sqrt{p}\|A_{\del,u}\si\|_2 d_{\del}(T).
\end{align*}
Taking the $L_p(\Om_{\si})$-norm on both sides and using Lemma~\ref{lem:Lpl2} of Appendix~\ref{sec:probTools} we find
\begin{align*}
(\E_{\si}Z^p)^{1/p} & \lesssim (\E_{\si}\|A_{\del,u}\si\|_2^p)^{1/p}\Big(\ga_2(T,\|\cdot\|_{\del}) + \sqrt{p}d_{\del}(T)\Big) \\
& \lesssim (\sqrt{p}\|A_{\del,u}\| + \|A_{\del,u}\|_F)\Big(\ga_2(T,\|\cdot\|_{\del}) + \sqrt{p}d_{\del}(T)\Big) \\
& \leq \Big(\sqrt{\frac{p}{s}} + 1\Big)\Big(\ga_2(T,\|\cdot\|_{\del}) + \sqrt{p}d_{\del}(T)\Big).
\end{align*}
Now take the $L_p(\Om_{\del})$-norm on both sides to obtain the result. 
\end{proof}

\subsection{Unconstrained case}

We first consider unconstrained least squares minimization, i.e., the constraint set is $\cC=\R^d$. 
\begin{theorem}
\label{thm:LSuc}
Let $\mathrm{Col}(A)$ be the column space of $A$ and let 
$$\mu(A)=\max_{1\leq j\leq n} \|P_{\mathrm{Col}(A)}e_j\|_2$$
be its largest leverage score. Set $\cC=\R^d$ and let $x_*$ and $\hat{x}$ be minimizers of (\ref{eqn:CLS}) and (\ref{eqn:SCLS}), respectively. Let $r(A)$ be the rank of $A$. Suppose that
\begin{align*}
m &\gtrsim \eps^{-1}((\log \eta^{-1})(r(A) + \log m)(\log m)^2 + (\log \eta^{-1})^2 r(A)) \\
s &\gtrsim \eps^{-1}\mu(A)^2((\log \eta^{-1})^2 + (\log \eta^{-1})(\log m)^3) \\
 & \qquad \qquad + \eps^{-1/2}\mu(A)((\log \eta^{-1})^{3/2} + (\log \eta^{-1})(\log m)^{3/2})
\end{align*}
and $\eta\leq \frac{1}{m}$. Then, with probability at least $1-\eta$,
\begin{equation}
\label{eqn:LSuc}
f(\hat{x})\leq (1-\eps)^{-2}f(x_*).
\end{equation}
\end{theorem}
\begin{proof}
Set $Z_1=Z_1(A,\Phi,T_{\cC}(x_*))$ and $Z_2=Z_2(A,\Phi,T_{\cC}(x_*),u)$. Observe that $T_{\cC}(x_*)=\R^d$. Applying Theorem~\ref{thm:JLsubspace} for the $r(A)$-dimensional subspace $E=\mathrm{Col}(A)$ immediately yields that $Z_1\geq 1-\sqrt{\eps}$ with probability at least $1-\eta$. Setting $T=A T_{\cC}(x_*)\cap S^{n-1}=\mathrm{Col}(A)\cap S^{n-1}$ in Lemma~\ref{lem:Z2boundGen} and subsequently using \Eqsub{ga2Subs} and \Eqsub{diamSubs} yields, for any $p\geq \log(m)$,
\begin{align*}
(\E_{\del,\si} Z_2^p)^{1/p} & \lesssim \Big(\sqrt{\frac{p}{s}} + 1\Big)\Big((\E_{\del}\ga_2^p(T,\|\cdot\|_{\del}))^{1/p} + \sqrt{p}(\E_{\del}d_{\del}^p(T))^{1/p}\Big) \\
& \lesssim \Big(\sqrt{\frac{p}{s}} + 1\Big)\Big[\frac{(\sqrt{r(A)}+\log^{1/2}(m))\log(m)}{\sqrt{m}} + \sqrt{\frac{pr(A)}{m}}  \\
& \qquad \qquad \qquad \qquad \qquad \qquad  + \frac{\sqrt{p}\log^{3/2}(m)}{\sqrt{s}}\mu(A) + \frac{p}{\sqrt{s}}\mu(A)\Big].
\end{align*}
Applying Lemma~\ref{lem:MomentsToTails} (and setting $w=\log(\eta^{-1})$ in (\ref{eqn:MomentsToTails})) shows that $Z_2\leq \sqrt{\eps}$ with probability at least $1-\eta$. The second statement in Lemma~\ref{lem:compare} now implies that with probability at least $1-2\eta$, 
$$f(\hat{x})\leq \Big(1+\frac{\eps}{(1-\sqrt{\eps})^2}\Big)f(x_*) \leq (1-\eps)^{-2}f(x_*).$$ 
To see the last inequality, note that it is equivalent to
$$-\eps+4\eps^{3/2}-3\eps^2-2\eps^{5/2}+2\eps^3\leq 0.$$
Clearly this is satisfied for $\eps$ small enough (since then $-\eps$ is the leading term). In fact, one can verify by calculus that this holds for $\eps\leq 1/10$, which we may assume without loss of generality.
\end{proof}
If $\Phi$ is a full random sign matrix (i.e., $s=m$) then it follows from our proof that (\ref{eqn:LSuc}) holds with probability at least $1-\eta$ if
$$m\gtrsim \eps^{-1}(r(A)+\log(\eta^{-1})).$$
This bound on $m$ is new, and was also recently shown in \cite{Woodruff14}. Previous works \cite{Sarlos06,KN14} allowed either $m \gtrsim \eps^{-2}(r(A) + \log(\eta^{-1}))$ or $m\gtrsim \eps^{-1} r(A)\log(\eta^{-1})$. Theorem~\ref{thm:LSuc} substantially improves $s$ while maintaining $m$ up to logarithmic factors.

\subsection{$\ell_{2,1}$-constrained case}
\label{sec:l21Constraint}

Throughout this section, we set $d=bD$. For $x=(x_{B_1},\ldots,x_{B_b})\in \R^{d}$ consisting of $b$ blocks, each of dimension $D$, we define its $\ell_{2,1}$-norm by
\begin{equation}
\label{eqn:l21NormDef}
\|x\|_{2,1}:= \|x\|_{\ell_1^b(\ell_2^D)} = \sum_{\ell=1}^b \|x_{B_{\ell}}\|_2.
\end{equation}
The corresponding dual norm is
$$\|x\|_{2,\infty} := \|x\|_{\ell_{\infty}^b(\ell_2^D)} = \max_{1\leq \ell\leq b} \|x_{B_{\ell}}\|_2.$$
In this section we study the effect of sketching on the $\ell_{2,1}$-constrained least squares program 
$$\min \|Ax-b\|_2^2 \qquad \mathrm{subject \ to} \qquad \|x\|_{2,1}\leq R,$$
which is Eq.~(\ref{eqn:CLS}) for $\cC=\{x\in \R^d \ : \ \|x\|_{2,1}\leq R\}$. In the statistics literature, this is called the \emph{group Lasso} (with non-overlapping groups of equal size). The $\ell_{2,1}$-constraint encourages a block sparse solution, i.e., a solution which has few blocks containing non-zero entries. We refer to e.g.\ \cite{BuG11} for more information. In the special case $D=1$ the program reduces to
$$\min \|Ax-b\|_2^2 \qquad \mathrm{subject \ to} \qquad \|x\|_1\leq R,$$
which is the well-known \emph{Lasso} \cite{Tibshirani96}.\par
To formulate our results we consider two norms on $\R^{n\ti d}$, given by  
\begin{equation}
\label{eqn:blockNormA}
\tnorm{A} := \max_{1\leq \ell\leq b}\Big(\sum_{j=1}^n\sum_{k\in B_{\ell}}|A_{jk}|^2\Big)^{1/2}
\end{equation}
and
\begin{equation*}
\|A\|_{\ell_{2,1}\rightarrow \ell_{\infty}} = \max_{1\leq j\leq n}\max_{1\leq \ell\leq b}\|(A_{jk})_{k\in B_{\ell}}\|_2. 
\end{equation*}
\begin{lemma}
\label{lem:gamma2pEst}
Let $A\in\R^{n\ti d}$ and consider any set $\cK\subset\R^d$. If $p\geq \log m$, then
\begin{align*}
& (\E_{\del}\ga_2^{2p}(A\cK\cap S^{n-1},\|\cdot\|_{\del}))^{1/p} \\
& \qquad \lesssim \al\Big[\sup_{x\in \cK \ : \ \|Ax\|_2=1}\|x\|_{2,1}^2\Big] \Big(\frac{1}{s}(\log b + p)\|A\|_{\ell_{2,1}\rightarrow \ell_{\infty}}^2 + \frac{1}{m} \tnorm{A}^2\Big).
\end{align*}
where $\al = (\log n)^6(\log m)^2(\log b)^2$. 
\end{lemma}
\begin{proof}
We deduce the assertion from Lemma~\ref{lem:BoN410}. By H\"{o}lder's inequality,
$$\sup_{y\in A\cK\cap S^{n-1}}\langle x,y\rangle = \sup_{v \in \cK \ : \ \|Av\|_2=1} \langle A^*x,v\rangle\leq \|A^*x\|_{2,\infty} \sup_{v\in \cK \ : \ \|Av\|_2=1}\|v\|_{2,1}.$$
We define $\|x\|$ to be the expression on the right hand side. Observe that for any sequence $(x_k)_{k\geq 1}$ in $\R^n$,
\begin{align*}
\E_g\Big\|\sum_{k\geq 1} g_k A^*x_k\Big\|_{2,\infty} & = \E_g\max_{1\leq \ell\leq b}\Big\|\sum_{k\geq 1} g_k(A^*x_k)_{B_{\ell}}\Big\|_2 \\
& \leq (\log b)^{1/2} \max_{1\leq \ell\leq b} \Big(\sum_{k\geq 1} \|(A^*x_k)_{B_{\ell}}\|_2^2\Big)^{1/2} \\
& \leq (\log b)^{1/2} \Big(\sum_{k\geq 1} \|A^*x_k\|_{2,\infty}^2\Big)^{1/2}
\end{align*}
and therefore $T_2(\|\cdot\|)\leq (\log b)^{1/2}$. Similarly,
\begin{align}
\EquationName{gamma2pFirst}
\E_g\Big\|\sum_{j=1}^n g_j\del_{ij} A^*e_j\Big\|_{2,\infty} & = \E_g\max_{1\leq \ell\leq b}\Big\|\Big(\sum_{j=1}^n g_j\del_{ij}A_{jk}\Big)_{k\in B_{\ell}}\Big\|_2 \nonumber \\
& \leq (\log b)^{1/2} \max_{1\leq \ell\leq b} \Big(\sum_{j=1}^n \del_{ij}\|(A_{jk})_{k\in B_{\ell}}\|_2^2\Big)^{1/2}.
\end{align}
Finally,
\begin{equation*}
d_{\|\cdot\|}\Big(\bigcup_{i=1}^m B_{J_i}\Big) = \sup_{v\in \cK \ : \ \|Av\|_2=1}\|v\|_{2,1} \max_{1\leq i\leq m} \sup_{x\in B_{J_i}} \|A^*x\|_{2,\infty} 
\end{equation*}
and for any $1\leq i\leq m$ and $x\in B_{J_i}$,
\begin{align*}
\|A^*x\|_{2,\infty} & = \max_{1\leq \ell\leq b}\Big\|\sum_{j=1}^n x_j\del_{ij}(A_{jk})_{k\in B_{\ell}}\Big\|_2 \\
& \leq  \max_{1\leq \ell\leq b}\sum_{j=1}^n |x_j| \ \del_{ij}\|(A_{jk})_{k\in B_{\ell}}\|_2 \leq  \max_{1\leq \ell\leq b}\Big(\sum_{j=1}^n \del_{ij}\|(A_{jk})_{k\in B_{\ell}}\|_2^2\Big)^{1/2}.
\end{align*}
We now apply Lemma~\ref{lem:BoN410} for $T=A\cK\cap S^{n-1}$ and subsequently take $L_p(\Om_{\del})$-norms to conclude that for any $p\geq \log m$,
\begin{align}
\label{eqn:ga2pIntermediate}
& (\E\ga_2^{2p}(A\cK\cap S^{n-1},\|\cdot\|_{\del}))^{1/p} \nonumber\\
& \qquad \lesssim \frac{\al}{s}\Big[\sup_{x\in \cK \ : \ \|Ax\|_2=1}\|x\|_{2,1}^2\Big] \max_{1\leq i\leq m} \Big(\E_{\del}\max_{1\leq \ell\leq b}\Big(\sum_{j=1}^n \del_{ij}\|(A_{jk})_{k\in B_{\ell}}\|_2^2\Big)^p\Big)^{1/p}.
\end{align}
Since $\E\del_{ij} = \frac{s}{m}$ and $(\del_{ij})_{j=1}^n$ is independent for any fixed $i$, we obtain by symmetrization (see \Equation{symmetrization}) for $(r_j)$ a vector of independent Rademachers
\begin{align*}
& \Big(\E_{\del}\max_{1\leq \ell\leq b}\Big(\sum_{j=1}^n \del_{ij}\|(A_{jk})_{k\in B_{\ell}}\|_2^2\Big)^p\Big)^{1/p} \\
& \qquad \leq \frac{s}{m}\max_{1\leq \ell\leq b} \sum_{j=1}^n\|(A_{jk})_{k\in B_{\ell}}\|_2^2 + 2\Big(\E_{\del}\E_r\max_{1\leq \ell\leq b}\Big(\sum_{j=1}^n r_j\del_{ij}\|(A_{jk})_{k\in B_{\ell}}\|_2^2\Big)^p\Big)^{1/p}.
\end{align*}
By Khintchine's inequality,
\begin{align*}
& \Big(\E_{\del}\E_r\max_{1\leq \ell\leq b}\Big(\sum_{j=1}^n r_j\del_{ij}\|(A_{jk})_{k\in B_{\ell}}\|_2^2\Big)^p\Big)^{1/p} \\
& \qquad \lesssim ((\log b)^{1/2} + p^{1/2}) \Big(\E_{\del}\max_{1\leq \ell\leq b}\Big(\sum_{j=1}^n \del_{ij}\|(A_{jk})_{k\in B_{\ell}}\|_2^4\Big)^{p/2}\Big)^{1/p} \\
& \qquad \leq ((\log b)^{1/2} + p^{1/2}) \Big(\E_{\del}\max_{1\leq \ell\leq b}\Big(\sum_{j=1}^n \del_{ij}\|(A_{jk})_{k\in B_{\ell}}\|_2^2\Big)^{p}\Big)^{1/(2p)} \\
& \qquad \qquad \qquad \qquad \qquad \qquad \qquad \times \max_{1\leq \ell\leq b}\max_{1\leq j\leq n} \|(A_{jk})_{k\in B_{\ell}}\|_2.
\end{align*}
Putting the last two displays together we find a quadratic inequality. By solving it, we arrive at
\begin{align}
\EquationName{gamma2pLast}
& \Big(\E_{\del}\max_{1\leq \ell\leq b}\Big(\sum_{j=1}^n \del_{ij}\|(A_{jk})_{k\in B_{\ell}}\|_2^2\Big)^p\Big)^{1/p} \nonumber\\
& \qquad \lesssim \frac{s}{m} \max_{1\leq \ell\leq b} \sum_j \|(A_{jk})_{k\in B_{\ell}}\|_2^2 + (\log b + p) \max_{1\leq \ell\leq b}\max_{1\leq j\leq n}\|(A_{jk})_{k\in B_{\ell}}\|_2^2.
\end{align}
Combining this estimate with (\ref{eqn:ga2pIntermediate}) yields the assertion.
\end{proof}
\begin{theorem}
\label{thm:SJLmain}
Let $\cC\subset\R^d$ be a closed convex set and let $\al=(\log n)^6(\log m)^2(\log b)^2$. Let $x_*$ and $\hat{x}$ be minimizers of (\ref{eqn:CLS}) and (\ref{eqn:SCLS}), respectively. Set
$$d_{2,1} = \sup_{x\in T_{\cC}(x_*) \ : \ \|Ax\|_{2,1}=1}\|x\|_{2,1}.$$
Suppose that 
\begin{align*}
m & \gtrsim \eps^{-2}\log(\eta^{-1})(\al + \log(\eta^{-1})\log b) \tnorm{A}^2 d_{2,1}^2,\\
s & \gtrsim \eps^{-2}\log(\eta^{-1})(\log b + \log(\eta^{-1}))(\al + \log(\eta^{-1})\log b) \|A\|_{\ell_{2,1}\rightarrow \ell_{\infty}}^2 d_{2,1}^2,
\end{align*}
and $\eta\leq \frac{1}{m}$. Then, with probability at least $1-\eta$,
$$f(\hat{x})\leq (1-\eps)^{-2}f(x_*).$$
\end{theorem}
\begin{proof}
Set $T=A T_{\cC}(x_*)\cap S^{n-1}$, $Z_1=Z_1(A,\Phi,T_{\cC}(x_*))$ and $Z_2=Z_2(A,\Phi,T_{\cC}(x_*),u)$. We apply \Equation{delTermBd}, \Equation{gamma2pFirst} and \Equation{gamma2pLast} to find for any $p\geq \log m$,
\begin{align}
\EquationName{del2pl2l1}
& (\E d_{\del}^{2p}(T))^{1/p} \nonumber\\
& \qquad \leq \frac{1}{s}\Big[\sup_{x\in T_{\cC}(x_*) \ : \ \|Ax\|_{2,1}=1}\|x\|_{2,1}^2\Big] \Big(\E_{\del}\Big(\max_{1\leq i\leq m} \E_g\Big\|\sum_{j=1}^n g_j\del_{ij} A^*e_j\Big\|_{2,\infty}\Big)^{2p}\Big)^{1/p} \nonumber \\
& \qquad \leq \frac{1}{s} (\log b)\Big[\sup_{x\in T_{\cC}(x_*) \ : \ \|Ax\|_{2,1}=1}\|x\|_{2,1}^2\Big]\Big(\E_{\del}\max_{1\leq \ell\leq b}\Big(\sum_{j=1}^n \del_{ij}\|(A_{jk})_{k\in B_{\ell}}\|_2^2\Big)^p\Big)^{1/p} \nonumber \\
& \qquad \leq (\log b)\Big[\sup_{x\in T_{\cC}(x_*) \ : \ \|Ax\|_{2,1}=1}\|x\|_{2,1}^2\Big] \Big(\frac{1}{m} \tnorm{A}^2 + \frac{1}{s}(\log b + p) \|A\|_{\ell_{2,1}\rightarrow \ell_{\infty}}^2\Big).
\end{align}
By \Equation{RIPLpEst}, Lemma~\ref{lem:gamma2pEst} and \Equation{del2pl2l1}
\begin{align*}
& \Big(\E_{\del,\si}\sup_{x \in T} \Big|\|\Phi x\|_2^2 - \|x\|_2^2\Big|^p\Big)^{1/p}\nonumber\\
& \ \ \ \lesssim (\E_{\del}\ga_2^{2p}(T,\|\cdot\|_{\del}))^{1/p} + (\E_{\del}\ga_2^{p}(T,\|\cdot\|_{\del}))^{1/p} + \sqrt{p}(\E d_{\del}^p(T))^{1/p} + p(\E d_{\del}^{2p}(T))^{1/p} \\
& \ \ \ \lesssim \frac{\al E}{s}(\log b + p) + \frac{\al F}{m} + \Big(\frac{\al E}{s}(\log b + p) + \frac{\al F}{m}\Big)^{1/2} \\
& \ \ \ \qquad + \frac{E p\log b}{s}(\log b + p) + \frac{F p\log b}{m} + \Big(\frac{E p\log b}{s}(\log b + p) + \frac{F p\log b}{m}\Big)^{1/2}, 
\end{align*}
where we write
\begin{equation*}
\label{eqn:CDDef}
E = \|A\|_{\ell_{2,1}\rightarrow \ell_{\infty}}^2 d_{2,1}^2, \ F = \tnorm{A}^2 d_{2,1}^2.
\end{equation*}
Similarly, using Lemma~\ref{lem:Z2boundGen} (with $Z(u):=Z_2(A,\Phi,T_{\cC}(x_*),u)$) we find
\begin{align*}
(\E_{\del,\si}Z_2^p)^{1/p} & \lesssim \Big(\sqrt{\frac{p}{s}} + 1\Big)\Big((\E_{\del}\ga_2^p(T,\|\cdot\|_{\del}))^{1/p} + \sqrt{p}(\E_{\del}d_{\del}^p(T))^{1/p}\Big) \\
& \lesssim \Big(\sqrt{\frac{p}{s}} + 1\Big)\Big[\Big(\frac{\al E}{s}(\log b + p) + \frac{\al F}{m}\Big)^{1/2} \\
& \qquad \qquad \qquad \qquad \qquad \qquad + \Big(\frac{E p\log b}{s}(\log b + p) + \frac{F p\log b}{m}\Big)^{1/2}\Big].
\end{align*}
Now apply Lemma~\ref{lem:MomentsToTails} (with $w=\log(\eta^{-1})$) to conclude that with probability at least $1-\eta$ we have $Z_2\leq \eps$ and $Z_1\geq 1-\eps$ under the assumptions on $m$ and $s$. Lemma~\ref{lem:compare} now implies that 
$$f(\hat{x})\leq \Big(1+\frac{\eps}{1-\eps}\Big)^2f(x_*) = (1-\eps)^{-2}f(x_*).$$ 
\end{proof}
\begin{corollary}
\label{cor:SJLl21}
Set 
$$\cC=\{x\in \R^d \ : \ \|x\|_{2,1}\leq R\}.$$ 
Suppose that $x_*$ is $k$-block sparse and $\|x_*\|_{2,1}=R$. Define
$$\si_{\min,k} = \inf_{\|y\|_2=1, \ \|y\|_{2,1}\leq 2\sqrt{k}}\|Ay\|_2.$$
Set $\al =  (\log n)^6(\log m)^2(\log b)^2$. Assume that 
\begin{align*}
m & \gtrsim \eps^{-2}\log(\eta^{-1})(\al + \log(\eta^{-1}) \log b) \tnorm{A}^2 k\si_{\min,k}^{-2},\\
s & \gtrsim \eps^{-2}\log(\eta^{-1})(\log(\eta^{-1}) + \log b)(\al + \log(\eta^{-1}) \log b) \|A\|_{\ell_{2,1}\rightarrow \ell_{\infty}}^2 k\si_{\min,k}^{-2},
\end{align*}
and $\eta\leq \frac{1}{m}$. Then with probability at least $1-\eta$
$$f(\hat{x})\leq (1-\eps)^{-2}f(x^*).$$
\end{corollary}
\begin{proof}
As calculated in Example~\ref{exa:TCl21} of Appendix~\ref{sec:convex-tools}, every $x\in T_{\cC}(x_*)$ satisfies $\|x\|_{2,1}^2\leq 4k\|x\|_2^2$. If moreover $\|Ax\|_2=1$, then
$$\si_{\min,k}^2\leq \|x\|_2^{-2}\leq 4k\|x\|_{2,1}^{-2}.$$ 
In other words,
$$\sup_{x\in T_{\cC}(x_*) \ : \ \|Ax\|_2=1}\|x\|_{2,1}^2 \leq 4k\si_{\min,k}^{-2}.$$
\end{proof}
For a dense random sign matrix $\Phi$, it follows from \cite[Theorem 1]{PiW14} that the result in Corollary~\ref{cor:SJLl21} holds under the condition
$$m\gtrsim \eps^{-2}(\tnorm{A}^2 k\si_{\min,k}^{-2}+\log(\eta^{-1})).$$
Thus, our result shows that one can substantially increase the sparsity in the matrix $\Phi$, while maintaining the embedding dimension $m$ (up to logarithmic factors).\par
Let us now specialize our results to the case $D=1$, which corresponds to the Lasso. In this case, if we let $A_k$ denote the columns of $A$, 
$$\tnorm{A} = \max_{1\leq k\leq d}\|A_k\|_2,$$
the maximum column of $A$. We can alternatively interpret this as the norm $\|A\|_{\ell_1\rightarrow\ell_2}$. Moreover, the norm $\|A\|_{\ell_{2,1}\rightarrow \ell_{\infty}}$ reduces to 
$$\|A\|_{\ell_{1}\rightarrow \ell_{\infty}} = \max_{1\leq j\leq n}\max_{1\leq k\leq d} |A_{jk}|,$$
the largest entry of the matrix. The first norm can be significantly larger than the second. For example, if $A$ is filled with $\pm 1$ entries, then 
$$\|A\|_{\ell_1\rightarrow\ell_2} = \sqrt{n}, \qquad \|A\|_{\ell_{1}\rightarrow \ell_{\infty}} = 1.$$
A similar result for the fast J-L transform, but with a worse dependence of $m$ on the matrix $A$, was obtained in \cite{PiW14}. For completeness, we will show in Appendix~\ref{sec:fjlt-cls} that the fast J-L transform in fact satisfies similar results as Theorem~\ref{thm:SJLmain} and Corollary~\ref{cor:SJLl21}, with the essentially the same dependence of $m$ on $A$.

\section{Proof of the main theorem}\SectionName{general}
After having seen instantiations of various subsets of our ideas for specific applications (linear subspaces, $\tnorm{\cdot}_T$ of small type-2 constant, and closed convex cones), we now prove the main theorem of this work, \Theorem{main}. Our starting point is the observation that inequality \Equation{3.3}, i.e., 
\begin{align}
& \log\mathcal{N}(\tilde{T}, \vertiii{\cdot}, t) \nonumber\\
& \qquad \lesssim \Big(\log\frac {d_{\del}(T)}{t}\Big)(\log m)\log\mathcal{N}\Big(\mathrm{conv}\Big(\bigcup_{i=1}^m B_{J_i}\Big), \tnorm{\cdot}_T, \frac 18\sqrt{s} t\Big) \EquationName{3.3Rep}
\end{align}
holds in full generality. We will need the following replacement of \Claim{claim}.  
\begin{lemma}\LemmaName{maurey}
Let $\epsilon>0$. Then 
\begin{align}
\nonumber & \log\mathcal{N}\Big(\mathrm{conv}\Big(\bigcup_{i=1}^m B_{J_i}\Big), \tnorm{\cdot}_T, \epsilon\Big) \\ 
& \qquad \lesssim \frac 1{\epsilon^2}\log m + (\log 1/\epsilon) \max_{k\lesssim \frac 1{\epsilon^2}} \max_{|A| = k} \log\mathcal{N}\Big(\frac 1k\sum_{i\in A}B_{J_i}, \tnorm{\cdot}_T, \epsilon\Big) \EquationName{5.2}
\end{align}
\end{lemma}
\begin{proof}
Set $\rho_T(x,y)=\tnorm{x-y}_T$ and let $\rho_{\ell_2}(x,y)=\|x-y\|_2$. By Maurey's lemma (Lemma~\ref{lem:Maurey} for $\|\cdot\|_2$), given $x\in\mathrm{conv}(B_{J_i} : 1\le i\le m)$, there is a $k\lesssim 1/\epsilon^2$ and a set $A\subset\{1,\ldots,m\}$, $|A| = k$, such that
\begin{equation}
\rho_T(x, \mathrm{conv}(B_{J_i}; i\in A)) \leq \rho_{\ell_2}(x, \mathrm{conv}(B_{J_i}; i\in A)) \le \epsilon  \EquationName{5.1}
\end{equation}
Now, let us take an element $y\in\mathrm{conv}(B_{J_i} ; i\in A)$, $|A| \lesssim 1/\epsilon^2$. Thus
$$
y = \sum_{i=1}^k \lambda_i y_i
$$
with $y_i\in B_{J_i}$, $\lambda_i\ge 0$, $\sum_i \lambda_i = 1$, $k \lesssim 1/\epsilon^2$. Firstly, we may dismiss the coefficients $\lambda_i < \epsilon^3$. Indeed, let $S=\{i \ : \ \la_i<\epsilon^3\}$ and set $\hat{y}=\sum_{i\in S}\la_i y_i$. Then, 
$$\|\hat{y}\|_2 \leq \sum_{i\in S}\la_i\leq k\epsilon^3\lesssim \epsilon.$$
Consider now the $\la_i$ with $\epsilon^3 \le \lambda_i \le 1$. For $\ell=0,1,\ldots,\ell_*$, $\ell_* \simeq \log(1/\epsilon)$, denote
$$
A_\ell = \Big\{1\leq i\leq k \ : \ \frac 1{2^\ell} \ge \lambda_i > \frac 1{2^{\ell+1}} \Big\}
$$
and write
\begin{equation}
y = \sum_{\ell=0}^{\ell_*} 2^{-\ell} |A_\ell| \cdot\Big(\frac 1{|A_\ell|} \sum_{i\in A_{\ell}} y_i'\Big) + \hat{y} \EquationName{one}
\end{equation}
where $y_i' = \lambda_i 2^\ell y_i \in B_{J_i}$. Note that
\begin{equation}
\sum_{\ell=0}^{\ell_*} 2^{-\ell} |A_\ell| < 2\sum_i \lambda_i \le 2 \EquationName{two}
\end{equation}
and
$$
\frac 1{|A_\ell|} \sum_{i\in A_{\ell}} y_i' \in \frac 1{k_\ell} \sum_{i\in A_{\ell}} B_{J_i}
$$
where $k_\ell = |A_\ell| \lesssim 1/\epsilon^2$. Take finite sets $\xi_\ell \subset \frac 1{k_\ell} \sum_{i\in A_\ell} B_{J_i}$ such that
\begin{align}
&\rho_T(z, \xi_\ell) < \epsilon \text{ for all } z\in \frac 1{k_\ell}\sum_{i\in A_\ell} B_{J_i}\\
&|\xi_\ell| = \mathcal{N}\Big(\frac 1{k_\ell}\sum_{i\in A_\ell} B_{J_i},\tnorm{\cdot}_T,\epsilon\Big)
\end{align}

Let $z_\ell\in\xi_\ell$ satisfy
$$
\btnorm{\frac 1{|A_\ell|} \sum_{i\in A_\ell} y_i' - z_\ell}_T < \epsilon
$$
and set
\begin{equation}
z = \sum_{\ell=0}^{\ell_*} 2^{-\ell} |A_\ell| z_\ell. \EquationName{five}
\end{equation}
By \Eqsub{one}, \Eqsub{two}
$$
\tnorm{y - z}_T \lesssim \sum_{\ell=0}^{\ell_*} 2^{-\ell} |A_\ell| \epsilon + \epsilon \lesssim \epsilon
$$

To summarize, we can find an $\epsilon$-net for $\mathrm{conv}(\cup_{1\leq i\leq m}B_{J_i})$ with respect to $\tnorm{\cdot}_T$ as follows. For every $A\subset [m]$ with $|A|=k$, $k\lesssim 1/\epsilon^2$, we select $\xi_1^A,\ldots,\xi_{\ell_*}^A$ as above. Then,
$$\bigcup_{k\in [m] \ : \ k\lesssim 1/\epsilon^2}\bigcup_{A\subset[m], |A|=k} \sum_{\ell=0}^{\ell_*}2^{-\ell}|A_{\ell}|\xi_{\ell}^A$$
is an $\epsilon$-net of cardinality at most
\begin{align*}
& \max_{k\lesssim 1/\epsilon^2} \frac 1{\epsilon^2} m {m\choose k} \prod_{\ell=0}^{\ell_*} \cN\Big(\frac{1}{k_{\ell}}\sum_{i\in A_{\ell}} B_{J_i},\tnorm{\cdot}_T,\epsilon\Big) \\
& \qquad \lesssim \max_{k\lesssim 1/\epsilon^2}\frac 1{\epsilon^2} m \Big(\frac{em}{k}\Big)^k \Big[\max_{k\leq 1/\epsilon^2}\max_{|A|=k} \cN\Big(\frac{1}{k}\sum_{i\in A} B_{J_i},\tnorm{\cdot}_T,\epsilon\Big)\Big]^{\log(1/\epsilon)}.
\end{align*}
This yields the result.
\end{proof}

Next, we analyze further the set $(1/k)\sum_{i\in A} B_{J_i}$ for some $k\lesssim 1/\epsilon^2$ ($\epsilon>0$ will be fixed later). The elements of $(1/k) \sum_{i\in A}B_{J_i}$ are of the form
$$
y = \frac 1k\sum_{j=1}^n \Big(\sum_{i\in A} \delta_{ij}x_j^{(i)}\Big) e_j
$$
with
$$
\sum_{j\in J_i} |x_j^{(i)}|^2 \le 1, \qquad \mathrm{for \ all} \ i.
$$
Therefore, by Cauchy-Schwarz,
\begin{align*}
\|y\|_2 &= \frac 1k \Big(\sum_{j=1}^n \Big|\sum_{i\in A} \delta_{ij} x_j^{(i)}\Big|^2\Big)^{1/2}\\
{}&\le \frac 1k \Big[\sum_{j=1}^n \Big(\sum_{i\in A}\delta_{ij}\Big)\sum_{i\in A} |x_j^{(i)}|^2 \Big]^{1/2}
\end{align*}

Define for $\alpha = 1,\ldots,\log(\min\{k, s\})$ the set
\begin{equation}
U_{\alpha} = U_{\alpha}(\delta) = \Big\{1\leq j\leq n : 2^{\al}\leq \sum_{i\in A} \delta_{ij}<2^{\alpha+1}\Big\} \EquationName{5.3}
\end{equation}
and define
$$U_0 = U_0(\delta) = \Big\{1\leq j\leq n : \sum_{i\in A} \delta_{ij}<2\Big\}.$$
Estimate for fixed $j$ 
\begin{equation}
\EquationName{IntensityEst}
\tau_{k,\alpha} \eqdef \Pr_\delta\Big(2^\al\leq \sum_{i\in A} \delta_{ij} < 2^{\al+1}\Big) \le
\begin{cases}
1,\ \hbox{if } 2^\alpha \le \frac {2esk} m\\
\min\Big\{2^{-\alpha} \frac {sk} m, 2^{-2^\alpha}\Big\},\ \hbox{if } 2^\alpha > \frac {2esk} m,
\end{cases}
\end{equation}
where the first term in the $\min$ is a consequence of Markov's inequality, and the second term follows by using that
$$\tau_{k,\alpha} \leq \bP\Big(\sum_{j=1}^k\zeta_j\geq 2^{\al}\Big),$$
whenever $\zeta_1,\ldots,\zeta_k$ are i.i.d.\ Bernoulli random variables with mean $s/m$ (see \Eqsub{binom-calc} for bounding Bernoulli moments). Hence, $U_\alpha$ is a random set of intensity $\tau_{k,\alpha}$, i.e., $\Pr_\delta (j \in U_\alpha(\delta)) = \tau_{k,\alpha}$. Write according to the preceding
$$
y = \sum_\alpha y_\alpha, \text{ with } y_\alpha = \sum_{j\in U_\alpha} y_j e_j
$$
and
$$
\|y_\alpha\|_2 \lesssim \frac 1{\sqrt{k}} 2^{\alpha/2}
$$
Hence, denoting $B_{U_{\al}}:=\{\sum_{j\in U_{\al}}x_je_j \ : \ \sum_{j\in U_{\al}}|x_j|^2\leq 1\}$, 
$$
\frac 1k \sum_{i\in A} B_{J_i} \subset \sum_\alpha \frac 1{\sqrt{k}} 2^{\alpha/2} B_{U_\alpha}.
$$
Therefore, we deduce that
\begin{align}
\nonumber \log\mathcal{N}\Big(\frac 1k\sum_{i\in A} B_{J_i}, \tnorm{\cdot}_T, \epsilon\Big) &\lesssim \sum_\alpha \log\mathcal{N}\Big(\frac {1}{\sqrt{k}} 2^{\alpha/2}B_{U_\alpha}, \tnorm{\cdot}_T, \frac{\epsilon}{\log m}\Big)\\
{}&= \sum_\alpha \log\mathcal{N}\Big(B_{U_\alpha}, \tnorm{\cdot}_T, \sqrt{\frac k{2^\alpha}} \frac{\epsilon}{\log m}\Big)\EquationName{5.4}
\end{align}
By the dual Sudakov inequality (Lemma~\ref{lem:dualSudakov} of Appendix~\ref{sec:probTools}), 
$$
t\Big[\log\mathcal{N}(B_{U_\alpha}, \tnorm{\cdot}_T, t)\Big]^{1/2} \lesssim \E_g\btnorm{\sum_{j\in U_\alpha} g_j e_j}_T
$$
It follows that
$$
\Big[\log\mathcal{N}\Big(\frac 1k\sum_{i\in A} B_{J_i}, \tnorm{\cdot}_T, \epsilon\Big)\Big]^{1/2} \lesssim \sum_\alpha \Big(\frac{2^\alpha}k\Big)^{1/2} \Big(\frac{\log m}{\epsilon}\Big) \E_g\btnorm{\sum_{j\in U_\alpha} g_j e_j}_T 
$$
and, by \Lemma{maurey}
\begin{align}
\nonumber \Big[\log& \mathcal{N}\Big(\mathrm{conv}\Big(\bigcup_{i=1}^m B_{J_i}\Big), \tnorm{\cdot}_T, \epsilon\Big)\Big]^{1/2}\\
{}&\le \frac 1{\epsilon}(\log m)^{1/2} + \frac{\log m}{\epsilon}\Big(\log\frac 1{\epsilon}\Big)^{1/2} \max_{\substack{k\lesssim \frac 1{\epsilon^2}\\|A| = k}}\Big[\sum_\alpha \sqrt{\frac{2^\alpha}k} \E_g\btnorm{\sum_{j\in U_\alpha} g_j e_j}_T\Big] \EquationName{5.5}
\end{align}
Applying \Equation{5.5} to \Equation{3.3Rep}, taking $\epsilon \simeq \sqrt{s} t$, gives
\begin{align*}
t\log^{1/2} \mathcal{N}(\tilde{T}, \vertiii{\cdot}, t) \lesssim& \frac 1{\sqrt{s}} \Big(\log \frac{1}{\sqrt{s}t}\Big)^{1/2}\log m\\
\nonumber {}& + \frac 1{\sqrt{s}}(\log m)^{3/2}\Big(\log \frac{1}{\sqrt{s}t}\Big)\cdot\\
{}& \hspace{.2in}\max_{\substack{|A| = k\\k\lesssim \frac 1{\epsilon^2}}}\Big\{\sum_\alpha \sqrt{\frac{2^\alpha}k} \E_g\btnorm{\sum_{j\in U_\alpha} g_j e_j}_T \Big\}.
\end{align*}
Using this bound for large $t$ and the elementary bound in \Equation{volComp} for small $t$ in \Equation{Dudley}, we obtain
\begin{align}
\nonumber \gamma_2(T,\|\cdot\|_{\del}) \lesssim &\frac 1{\sqrt{s}}(\log n)^{3/2}\log m \\
\nonumber &{} + \frac 1{\sqrt{s}}(\log m)^{3/2}(\log n)^2\cdot\\
{}& \hspace{.2in}\max_{k\le m}\max_{|A| = k} \Big\{\sum_{\alpha,2^\alpha\le k} \sqrt{\frac{2^\alpha}k} \E_g\btnorm{\sum_{j\in U_\alpha} g_j e_j}_T \Big\}\EquationName{5.6}
\end{align}
We split the sum over $\al$ into three different parts. Firstly,
\begin{align}
\EquationName{MainBd1}
\max_{k\le m}\max_{|A| = k} \sum_{\al,2^\alpha\leq 2esk/m} \sqrt{\frac{2^\alpha}k} \E_g\btnorm{\sum_{j\in U_\alpha} g_j e_j}_T & \lesssim \sqrt{\frac{s}{m}} \E_g\btnorm{\sum_{j=1}^n g_j e_j}_T \nonumber \\
& \lesssim \sqrt{\frac{s}{m}}\ga_2(T,\|\cdot\|_2).
\end{align}
Next, by setting 
$$U_A=\Big\{j\in [n] \ : \ \frac{2esk}{m}\leq \sum_{i\in A} \del_{ij}\Big\},$$
we obtain
\begin{align*}
& \max_{k\le m}\max_{|A| = k} \sum_{\al,2esk/m<2^{\al}\leq 10\log n} \sqrt{\frac{2^\alpha}{k}} \E_g\btnorm{\sum_{j\in U_\alpha} g_j e_j}_T \\
& \qquad \lesssim (\log n)^{1/2} \max_{k\le m}\max_{|A| = k} \frac{1}{\sqrt{k}} \E_g\btnorm{\sum_{j\in U_A} g_j e_j}_T \\
& \qquad \leq (\log n)^{1/2} \Big[\max_{m/s\leq k\le m}\max_{|A| = k} \frac{1}{\sqrt{k}} \E_g\btnorm{\sum_{j\in U_A} g_j e_j}_T \\
& \qquad \qquad \qquad \qquad \qquad \qquad + \max_{k\leq m/s}\max_{|A| = k} \frac{1}{\sqrt{k}} \E_g\btnorm{\sum_{j\in U_A} g_j e_j}_T\Big].
\end{align*}
Note that the first term on the right hand side is bounded by
\begin{equation}
\EquationName{MainBd2a}
\max_{m/s\leq k\le m}\max_{|A| = k} \frac{1}{\sqrt{k}} \E_g\btnorm{\sum_{j\in U_A} g_j e_j}_T \lesssim \sqrt{\frac{s}{m}}\ga_2(T,\|\cdot\|_2).
\end{equation}
To bound the second term, we take expectations with respect to $\del$ and find for $p_k=k\log\Big(\frac{m}{sk}\Big)$
\begin{align}
\EquationName{UAint}
& \E_{\del}\Big[\max_{k\leq m/s}\max_{|A| = k} \frac{1}{\sqrt{k}} \E_g\btnorm{\sum_{j\in U_A} g_j e_j}_T\Big] \nonumber\\
& \qquad \lesssim \max_{k\leq m/s}\max_{|A| = k} \frac{1}{\sqrt{k}} \Big(\E_{\del}\Big(\E_g\btnorm{\sum_{j\in U_A} g_j e_j}_T\Big)^{p_k}\Big)^{1/p_k}.
\end{align}
By \Equation{IntensityEst}, for any $A\subset [m]$ with $|A|=k$, $U_A$ is a random set of intensity at most $esk/m=sq/(m\log s)$, where we set $q=ek(\log s)$. If we now let $\eta_1,\ldots,\eta_n$ be i.i.d.\ $\{0,1\}$-valued random variables with expectation $qs/(m\log s)$, then the random set $U$ defined by $\bP(j\in U)=\bP(\eta_j=1)$ has a higher intensity than $U_A$ and therefore,
\begin{align}
\EquationName{MainBd2b}
& \E_{\del}\Big[\max_{k\leq m/s}\max_{|A| = k} \frac{1}{\sqrt{k}} \E_g\btnorm{\sum_{j\in U_A} g_j e_j}_T\Big] \nonumber\\
& \qquad \lesssim (\log s)^{1/2}\max_{q\leq \frac{m}{s}\log s} \frac{1}{\sqrt{q}} \Big(\E_{\eta}\Big(\E_g\btnorm{\sum_{j=1}^n \eta_j g_j e_j}_T\Big)^q\Big)^{1/q}.
\end{align}
Finally, consider the $\al$ with $2^{\al}>10\log n$. Since $\tnorm{\cdot}_T\leq \|\cdot\|_2$, 
\begin{align}
\EquationName{appLambda}
& \max_{k\le m}\max_{|A| = k} \sum_{\alpha,10\log n<2^{\al}\leq k} \sqrt{\frac{2^\alpha}k} \E_g\btnorm{\sum_{j\in U_\alpha} g_j e_j}_T \nonumber\\
& \qquad \leq \max_{k\le m}\max_{|A| = k} \sum_{\alpha,10\log n<2^{\al}\leq k} \sqrt{\frac{2^\alpha}k} |U_{\al}|^{1/2}.
\end{align}
By \Equation{IntensityEst}, the intensity of $U_{\al}$ is bounded by $2^{-2^{\al}}$ and therefore,
\begin{align*}
& \E_{\del}\Big[\max_{k\le m}\max_{|A| = k} \sum_{\alpha,10\log n<2^{\al}\leq k} \sqrt{\frac{2^\alpha}k} \E_g\btnorm{\sum_{j\in U_\alpha} g_j e_j}_T \Big] \\
& \qquad \leq \max_{k\le m}\max_{|A| = k} \sum_{\alpha,10\log n<2^{\al}\leq k} \sqrt{\frac{2^\alpha}k} \Big(\E\Big|\sum_{j=1}^n \zeta_j\Big|^{q_k}\Big)^{1/(2q_k)},
\end{align*}
where $q_k\simeq k\log(m/k)$ and the $\zeta_j$ are i.i.d.\ $\{0,1\}$-valued with mean $2^{-2^{\al}}$. Since $\sum_j \zeta_j$ is binomially distributed, we find using that $2^{\al}>10\log n$,
\begin{equation}
\Big\|\sum_j \zeta_j\Big\|_{L_{\zeta}^{q_k}} \leq \Big(\sum_{t=1}^n t^{q_k} e^{-\frac{9}{10}2^{\al}t}\Big)^{1/q_k} \lesssim n^{1/q_k} 2^{-\alpha} q_k \lesssim 2^{-\alpha} k(\log m)\EquationName{binom-calc}
\end{equation}
and therefore
\begin{equation}
\EquationName{MainBd3}
\E_{\del}\Big[\max_{k\le m}\max_{|A| = k} \sum_{\alpha,10\log n<2^{\al}\leq k} \sqrt{\frac{2^\alpha}k} \E_g\btnorm{\sum_{j\in U_\alpha} g_j e_j}_T\Big] \lesssim (\log m)^{3/2}.
\end{equation}

By combining \Equation{5.6}, \Equation{MainBd1}, \Equation{MainBd2a}, \Equation{MainBd2b}, and \Equation{MainBd3} we obtain
\begin{align}
\nonumber & \E \gamma_2(T,\|\cdot\|_{\del}) \\
\nonumber & \qquad \lesssim \frac{(\log m)^{3/2}(\log n)^{5/2}\gamma_2(T,\|\cdot\|_2)}{\sqrt{m}} +\frac 1{\sqrt{s}}\Bigg\{(\log m)^3(\log n)^2 +\\
{}&\qquad \hspace{.6in} (\log m)^2(\log n)^{5/2}\max_{q\le \frac ms \log s} \frac 1{\sqrt{q}}\Big(\E_{\eta}\Big(\E_g\btnorm{\sum_{j=1}^n \eta_jg_j e_j}_T\Big)^q\Big)^{1/q}\Bigg\}, \EquationName{5.13}
\end{align}
with $\eta_j$ i.i.d.\ Bernoulli with mean $qs/(m\log s)$. 

From our proof it is clear that $\E \gamma_2^2(T,\|\cdot\|_{\del})$ can be bounded by the square of the right hand side of \Equation{5.13}. Using \Equation{RIPLpEst} (for $p=1$), we conclude that
$$\E\sup_{x\in T}\Big|\|\Phi x\|_2^2 - 1\Big|<\eps$$
provided that $m,s$ satisfy
\begin{align}
m &\gtrsim (\log m)^3(\log n)^5 \frac{\gamma_2^2(T,\|\cdot\|_2)}{\eps^2} \EquationName{5.14}\\
s &\gtrsim (\log m)^6 (\log n)^4 \frac 1{\eps^2} \EquationName{5.15}
\end{align}
and 
\begin{equation}
\Eqsub{5.16} = (\log m)^2(\log n)^{5/2}\max_{q\le \frac ms \log s} \Big\{\frac 1{\sqrt{qs}} \Big(\E_{\eta}\Big(\E_g\btnorm{\sum_{j=1}^n \eta_j g_j e_j}_T\Big)^q\Big)^{1/q}\Big\} < \eps. \EquationName{5.16}
\end{equation}
This completes the proof.

\section{Example applications of main theorem}\SectionName{applications}
In this section we use our main theorem to give explicit conditions under which
$$\E\sup_{x\in T}\Big|\|\Phi x\|_2^2 - 1\Big|<\eps$$
for several interesting sets $T\subset S^{n-1}$. This amounts to computing an upper bound for the parameters $\ga_2(T,\|\cdot\|_2)$ and
\begin{equation}
\EquationName{kappaTDefEx}
\kappa(T) = \max_{q\le \frac ms \log s} \Big\{\frac 1{\sqrt{qs}} \Big(\E_{\eta}\Big(\E_g\btnorm{\sum_{j=1}^n \eta_j g_j e_j}_T\Big)^q\Big)^{1/q}\Big\}.
\end{equation}
We focus on the latter two and refer to \cite{Dirksen14} for details on how to estimate $\ga_2(T,\|\cdot\|_2)$. Note, however, that $\ga_2(T,\|\cdot\|_2)\lesssim (m \log s)^{1/2}\ka(T)$. Indeed, take $q=\frac{m}{s}\log s$ in \Equation{kappaTDefEx} and note that the $\eta_j$ are identically equal to $1$ in this case. This gives  
$$\ka(T)\geq (m \log s)^{-1/2} g(T) \simeq (m \log s)^{-1/2} \ga_2(T,\|\cdot\|_2).$$
Thus, if we ignore logarithmic factors, it suffices to bound $\ka(T)$.
\subsection{Linear subspace}
In the application from \Section{linear} with $T$ the unit ball of a $d$-dimensional linear subspace $E$ of $\ell_2^n$, we have $\gamma_2(T,\|\cdot\|_2) \simeq \sqrt{d}$ and
\begin{equation}
\Eqsub{5.16} \lesssim (\log m)^2(\log n)^{5/2}\frac 1{\sqrt{s}} \max_{q\le \frac ms \log s} \frac 1{\sqrt{q}} \Big\|\sum_j \|P_E e_j\|_2^2 \eta_j\Big\|_{L_\eta^q}^{1/2} \EquationName{5.17}
\end{equation}
with $\eta_j \in \{0,1\}$ i.i.d.\ of mean $qs/(m\log s)$. Using Khintchine's inequality,
\begin{equation}
\Big\|\sum_j \eta_j \|P_E e_j\|_2^2\Big\|_{L_\eta^q} < \frac{dqs}{m\log s} + q\max_j \|P_E e_j\|_2^2 \EquationName{5.18}
\end{equation}
and therefore
\begin{equation}
\Eqsub{5.16} \lesssim (\log m)^2(\log n)^{5/2}\Big(\sqrt{\frac dm} + \frac 1{\sqrt{s}}\max_j \|P_E e_j\|_2\Big) \EquationName{5.19}
\end{equation}

\Equation{5.14} and \Equation{5.15} then give conditions
\begin{align}
m &\gtrsim d(\log n)^5(\log m)^4 /\eps^2 \EquationName{5.20}\\
s &\gtrsim (\log n)^4(\log m)^6/\eps^2 + (\log n)^5(\log m)^4\max_j \|P_E e_j\|_2^2/\eps^2  \EquationName{5.21}
\end{align}

\subsection{$k$-sparse vectors}\SectionName{ksparse2}
Consider next the application from \Section{ksparse}, replacing $T$ by $K$ given by \Equation{4.1}. Thus $\gamma_2(K) \lesssim \sqrt{k\log n}$ and \Equation{5.16} is bounded by
\begin{equation}
\sqrt{k}(\log m)^2(\log n)^3 \frac 1{\sqrt{s}} \max_{q\le \frac ms \log s}\Big\{\frac 1{\sqrt{q}} \Big\|\max_i \Big(\sum_j \eta_j |A_{ij}|^2\Big) \Big\|_{L^q_{\eta}}^{1/2}\Big\} \EquationName{5.22}
\end{equation}
and
\begin{equation}
\Big\|\max_i (\sum_j \eta_j |A_{ij}|^2)\Big\|_{L^q_{\eta}} \lesssim \frac{qs}{m\log s} + q\max_{i,j}|A_{ij}|^2 \EquationName{5.23}
\end{equation}
Hence
\begin{equation}
\Eqsub{5.16} \lesssim (\log m)^2(\log n)^3\sqrt{\frac km} + (\log m)^2(\log n)^3\sqrt{\frac ks}\max_{i,j}|A_{ij}| \EquationName{5.24}
\end{equation}
We then arrive at the conditions
\begin{align}
m &\gtrsim k(\log n)^6(\log m)^4/\eps^2 \EquationName{5.25}\\
s &\gtrsim (\log n)^4(\log m)^6/\eps^2  \EquationName{5.26}
\end{align}
provided that
\begin{equation}
\max_{i,j} |A_{ij}| < k^{-1/2} (\log n)^{-1} \EquationName{5.27}
\end{equation}
(compare with \Equation{4.7}-\Eqsub{4.8}).

\subsection{Flat vectors}
Let $T\subseteq S^{n-1}$ be finite with $\|x\|_\infty\le \alpha$ for all $x\in T$. Then by a similar calculation as in \Eqsub{sqrt-trick},
\begin{align}
\Big(\E_\eta \Big(\E_g \sup_{x\in T} \Big|\sum_{i=1}^n \eta_i g_i x_i\Big|\Big)^q\Big)^{1/q} &\lesssim \sqrt{\log|T|}\cdot \Big\|\sup_{x\in T}\sum_{i=1}^n \eta_i x_i^2\Big\|_{L^q_{\eta}}^{1/2}\\
{}&\lesssim \sqrt{\frac{qs\log|T|}{m}} + (\alpha\log |T|)\sqrt{q}
\end{align}
Since $\gamma_2(T,\|\cdot\|_2) \lesssim \sqrt{\log|T|}$ we find the conditions
\begin{align}
m&\gtrsim (\log |T|)(\log m)^4(\log n)^5/\eps^2\\
s&\gtrsim (\alpha\log |T|)^2(\log m)^4(\log n)^5/\eps^2 + (\log m)^6(\log n)^4/\eps^2,
\end{align}
which is qualitatively similar to \cite[Theorem 4.1]{Matousek08}.

\subsection{Finite collection of subspaces}
Let $\Theta$ be a finite collection of $d$-dimensional subspaces $E\subset \R^n$ with $|\Theta| = N$. Define
$$
T = \bigcup_{E\in \Theta} \{x\in E : \|x\|_2 = 1 \}. 
$$
In this case, $\gamma_2(T, \|\cdot\|_2) \lesssim \sqrt{d} + \sqrt{\log N}$. For the duration of the next two sections, we define 
$$\alpha= \sup_{E\in\Theta} \max_j \|P_E e_j\|_2.$$
Recall that $\max_j \|P_E e_j\|_2$ is referred to as the {\em incoherence} of $E$, and thus $\sqrt{d/n} \le \alpha\le 1$ (cf.\ Remark~\ref{rem:coherence}) is the maximum incoherence in $\Theta$.

\subsubsection{Collection of incoherent subspaces}
To estimate $\kappa(T)$, consider the collection $\mathcal{A}$ of operators $A = \sum_j (P_E e_j)\otimes e_j$, $E\in \Theta$. Fix $\eta$ and define $R_\eta x = \sum_j \eta_j x_j e_j$. Then applying \cite[Theorem 3.5]{KMR14} to $\mathcal{A} R_\eta$,
\begin{align}
\nonumber \Big\|\max_{E\in \Theta} \|\sum_j \eta_j g_j P_E e_j\| \Big\|_{L_g^1} &= \Big\|\max_{A\in\mathcal{A}} \| AR_\eta g\|\Big\|_{L_g^1}\\
{}&\lesssim d_F(\mathcal{A}R_\eta) + \gamma_2(\mathcal{A} R_\eta, \|\cdot\|) . \EquationName{twoc}
\end{align}
Clearly 
\begin{equation*}
\|A\| \le 1,\ \|AR_\eta\|_F = \Big(\sum_j \eta_j \|P_E e_j\|_2^2\Big)^{1/2} .
\end{equation*}
By \cite{RV07} (see the formulation in \cite[Proposition 7]{Tropp08}),
\begin{equation}
\Big(\E_\eta \|A R_\eta\|^p\Big)^{1/p} \le 3\sqrt{p} \Big(\E \|AR_\eta\|_{1,2}^p\Big)^{1/p} + \sqrt{\varrho}\cdot \|A\| ,
\end{equation}
where $\E \eta_j = \varrho$ and we assume $p \ge 2 \log n$. Here $\|\cdot\|_{1,2}$ is the $\ell_1^n \rightarrow \ell_n^2$ norm, hence 
$$\|AR_{\eta}\|_{1,2} = \max_{1\leq k\leq n}\|AR_{\eta}e_k\|_2 = \max_{1\leq k\leq n}\eta_k \|P_E e_k\|_2 \le \alpha.$$
Taking $\varrho = \frac{qs}{m\log s}$, it follows that
\begin{equation}
\Big(\E_\eta \|AR_\eta\|^p\Big)^{1/p} \lesssim \alpha\sqrt{p} + \Big(\frac{qs}{m\log s}\Big)^{1/2} \EquationName{fourc}
\end{equation}

To estimate $\kappa(T)$, we need to bound
\begin{align}
\|\Eqsub{twoc}\|_{L_\eta^q} &\le  \underbrace{\Big\|\max_E \Big(\sum_j \eta_j \|P_E e_j\|_2^2\Big)\Big\|_{L_\eta^q}^{1/2}}_{\sqrt{\Eqsub{fivec}}} \EquationName{fivec}\\
{}&+ \underbrace{\|\gamma_2(\mathcal{A} R_\eta)\|_{L_{\eta}^q}}_{\Eqsub{sixc}} \EquationName{sixc}
\end{align}
First, by \Equation{5.18} and denoting $q_1 = q + \log N$,
\begin{equation}
\Eqsub{fivec} \lesssim \varrho d + q_1 \max_j \|P_E e_j\|_2^2 = \frac{dqs}{m\log s} + (q + \log N)\alpha^2 \EquationName{sevenc}
\end{equation}
Estimate trivially
$$
\gamma_2(\mathcal{A}R_\eta) \le \sqrt{\log N} \cdot \max_{A\in\mathcal{A}} \|A R_\eta\|
$$
and hence, applying \Equation{fourc} with $p = q + 2\log n + \log N$
\begin{align}
\nonumber \Eqsub{sixc} &\lesssim \sqrt{\log N} \cdot \Big(\max_A \E \|AR_\eta\|^p\Big)^{1/p} \\
{}&\lesssim \sqrt{\log N} \cdot \Big\{(q + \log n + \log N)\alpha^2 + \frac{qs}{m\log s}\Big\}^{1/2} \EquationName{eightc}
\end{align}

Collecting our estimates we find
\begin{align}
\nonumber \kappa(T) & \lesssim \max_{q\leq\frac{m}{s}\log s}\frac{1}{\sqrt{qs}} \Bigg\{\Big(\frac{dqs}{m\log s}\Big)^{1/2} + \alpha\sqrt{q}\\
\nonumber&\hspace{.1in}{} \qquad + \sqrt{\log N}\cdot\Big[(\sqrt{q} + \sqrt{\log n} + \sqrt{\log N})\alpha + \Big(\frac{qs}{m\log s}\Big)^{1/2}\Big]\Bigg\}\\
{}&\lesssim \Big(\frac{d + \log N}{m\log s}\Big)^{1/2} + \frac{(\sqrt{\log n}{\sqrt{\log N}} + \log N)\alpha}{\sqrt{s}} \EquationName{ninec}
\end{align}

Hence in this application (assuming $\log N \ge \log n$),
\begin{equation}
\Eqsub{5.16} < (\log m)^2(\log n)^{5/2} \Big\{\Big(\frac{d + \log N}{m}\Big)^{1/2} + \frac{\alpha\log N}{\sqrt{s}}\Big\} \EquationName{10c}
\end{equation}
and the conditions on $m,s$ become 
\begin{align}
m &\gtrsim (\log m)^4(\log n)^5(d + \log N)\eps^{-2} \EquationName{11c}\\
s&\gtrsim ((\log m)^6(\log n)^4 + (\alpha \log N)^2(\log m)^4(\log n)^5)\eps^{-2} \EquationName{12c}
\end{align}

Notice that $s$ depends \emph{only} on $|\Theta|$ and not on the dimension $d$. Thus this bound is of interest when $\log|\Theta|$ is small compared with $d$.

\subsubsection{Collection of coherent subspaces}
We can also obtain a bound on $s$ that does not improve for small $\alpha$, but has linear dependence on $\log N$. Here we will not rely on \cite{RV07}. As described in \Section{intro}, this setting has applications to model-based compressed sensing. For example, for approximate recovery of block-sparse signals using the notation of \Section{intro}, our bounds will show that a measurement matrix $\Phi$ may have $m\lts kb + k\log(n/k)$, $s \lts k\log(n/k)$ and allow for efficient recovery. This is non-trivial if the number of blocks satisfies $b \gts \log(n/k)$.

One may indeed trivially bound
$$
\gamma_2(\mathcal{A} R_\eta) \lesssim \sqrt{\log N}
$$
since certainly $\|AR_\eta\| \le \|A\| \le 1$. Hence
$$
\Eqsub{sixc} \lesssim \sqrt{\log N}
$$
leading to the following bound on $\kappa(T)$
\begin{equation*}
\max_{q\leq\frac{m}{s}\log s} \frac 1{\sqrt{qs}} \Big\{\sqrt{\frac{dqs}{m\log s}} + (\sqrt{q} + \sqrt{\log N})\alpha + \sqrt{\log N}\Big\} \lesssim \sqrt{\frac dm} + \sqrt{\frac {\log N}s}
\end{equation*}

Instead of \Eqsub{11c}, \Eqsub{12c}, one may impose the conditions
\begin{align}
m &\gtrsim (\log m)^4(\log n)^5 \frac{d}{\eps^2} + (\log m)^3(\log n)^5\frac{\log N}{\eps^2}\\
s &\gtrsim (\log m)^6(\log n)^4\frac 1{\eps^2} + (\log m)^4(\log n)^5\frac{\log N}{\eps^2}
\end{align}
We remark that previous work \cite{NN13b} which achieved $m \approx d/\eps^2$ for small $s$ had worse dependence on $\log N$: in particular $s \gtrsim (\log N)^3, m \gtrsim (\log N)^6$. In fact, Conjecture 14 of \cite{NN13b} if true would imply that the correct dependence on $N$ in both $m$ and $s$ should be $\log N$ (which is optimal due to known lower bounds for the Johnson-Lindenstrauss lemma, i.e., the special case $d=1$). We thus have shown that this implication is indeed true.

\subsection{Possibly infinite collection of subspaces}\SectionName{finsler}
Assume next $\Theta$ is an arbitrary family of linear $d$-dimensional subspaces $E\subset \R^n$, equipped with the Finsler metric
$$
\rho_{\mathrm{Fin}}(E, E') = \|P_E - P_E'\| .
$$
Let
$$
T = \bigcup_{E\in\Theta} \{ x\in E : \|x\|_2 = 1 \}
$$
for which (cf.\ \cite{Dirksen14})
\begin{equation}
\gamma_2(T, \|\cdot\|_2) \lesssim \sqrt{d} + \gamma_2(\Theta, \rho_{\mathrm{Fin}}) . \EquationName{13c}
\end{equation}
Fix some parameter $\eps_0>0$ and let $\Theta_1\subset \Theta$ be a finite subset such that
\begin{align}
&|\Theta_1| \le \mathcal{N}(\Theta, \rho_{\mathrm{Fin}}, \eps_0) \EquationName{15c}\\
&\rho_{\mathrm{Fin}}(E, \Theta_1) \leq \eps_0\text{ for any } E \in \Theta \EquationName{16c}
\end{align}
Let further
$$
T_1 = \bigcup_{E'\in\Theta_1} \{x \in E' : \|x\|_2 = 1\}
$$
Let $x\in T, x\in E, E\in\Theta$ and $\rho_{\mathrm{Fin}}(E,E') \leq \eps_0$ for some $E' \in \Theta_1$. Hence
\begin{equation}
x = P_{E'} x + (P_Ex - P_{E'} x) = x_1 + x_2 \EquationName{17c}
\end{equation}
where $x_1\in T_1$ and $x_2\in B_E + B_{E'}$, $\|x_2\|_2 \leq \eps_0$. Hence $x_2\in T_2$ with
$$
T_2 = \bigcup_{E, F\in\Theta} \{x \in B_E + B_F : \|x\|_2 \leq \eps_0 \}
$$

For $t < \eps_0$, we estimate $\mathcal{N}(T_2, \|\cdot\|_2, t)$. Let $\Theta_t \subset \Theta$ satisfy
\begin{align}
&|\Theta_t| \leq \mathcal{N}\Big(\Theta, \rho_{\mathrm{Fin}}, \frac t{4}\Big) \EquationName{18c}\\
&\rho_{\mathrm{Fin}}(E, \Theta_t) \leq \frac t{4} \text{ for all } E\in\Theta \EquationName{19c}
\end{align}
By \Equation{volComp}, for each $E'\in\Theta_t$ we can find $\xi_{E'} \subset B_{E'}$ such that
\begin{align}
&\log|\xi_{E'}| \lesssim d \log\frac 1t \EquationName{20c}\\
&\rho_{\ell_2}(x, \xi_{E'}) \leq t\text{ for all } x\in B_{E'} \EquationName{21c}
\end{align}
Denote
$$
\xi_t = \bigcup_{E',F'\in\Theta_t} (\xi_{E'} - \xi_{F'})
$$
for which by construction
$$
\log|\xi_t| \lesssim \log\mathcal{N}\Big(\Theta, \rho_{\mathrm{Fin}}, \frac t{4}\Big) + d\log \frac 1t
$$
Also, for $x\in T_2$, $x = y + z \in B_E + B_F$ and $E',F'\in\Theta_t$ satisfying
$$
\rho_{\mathrm{Fin}}(E, E') \leq \frac t{4}, \rho_{\mathrm{Fin}}(F, F') \leq \frac t{4} ,
$$
we have
$$
\|x - (P_{E'} y + P_{F'} z)\|_2 \leq \frac{t}{2},
$$
while
$$
\rho_{\ell_2}(P_{E'} y, \xi_{E'}) \leq \frac t{4}, \qquad \rho_{\ell_2}(P_{F'} z, \xi_{F'}) \leq \frac t{4}
$$
Therefore
$$
\rho_{\ell_2}(x, \xi_t) \leq t
$$
and we get for $t < \eps_0$ (otherwise $\Eqsub{22c} = 0$)
\begin{equation}
\log\mathcal{N}(T_2, \|\cdot\|_2, t) \lesssim \log\mathcal{N}\Big(\Theta, \rho_{\mathrm{Fin}}, \frac t{4}\Big) + d\log \frac 1t \EquationName{22c}
\end{equation}

Using the decomposition \Equation{17c} and the bound \Equation{10c} for the contribution of $T_1$, we find
\begin{align}
\kappa(T) & \lesssim \Big\{\Big(\frac{d + \log\mathcal{N}(\Theta, \rho_{\mathrm{Fin}}, \eps_0)}{m}\Big)^{1/2} + \frac{\alpha\log\mathcal{N}(\Theta, \rho_{\mathrm{Fin}}, \eps_0)}{\sqrt{s}}\Big\} \EquationName{23c}\\
& \qquad \ + \max_{q\le \frac ms\log s}\Big\{\frac 1{\sqrt{qs}} \Big(\E_{\eta}\Big(\E_g\sup_{x\in T_2} \Big|\sum_{j=1}^n \eta_j g_j x_j\Big|\Big)^q\Big)^{1/q}\Big\} \EquationName{24c}
\end{align}
with $(\eta_j)\in\{0,1\}$ i.i.d. of mean $\frac{qs}{m\log s}$.

Estimate by the contraction principle \cite{Kahane68}, the Gaussian concentration for Lipschitz functions (\Equation{GaussLip}), and Dudley's inequality \Equation{Dudley}\cite{Dudley67}
\begin{align}
\nonumber \frac 1{\sqrt{q}} \Big\|\sup_{x\in T_2} \Big|\sum_{j=1}^n \eta_j g_j x_j\Big|\Big\|_{L_g^1,L_\eta^q} &\le \frac 1{\sqrt{q}} \Big\|\sup_{x\in T_2} \Big|\sum_{j=1}^n g_j x_j\Big|\Big\|_{L_g^q}\\
\nonumber {}&\lesssim \Big\|\sup_{x\in T_2} \Big|\sum_{j=1}^n g_j x_j\Big|\Big\|_{L_g^1}\\
\nonumber {}&\lesssim \int_0^{\eps_0} [\log\mathcal{N}(T_2, \|\cdot\|_2, t)]^{1/2} dt \\
{}&\lesssim \int_0^{\eps_0}[\log\mathcal{N}(\Theta, \rho_{\mathrm{Fin}}, t)]^{1/2} dt + \sqrt{d} \eps_0 \sqrt{\log\frac 1{\eps_0}}, \EquationName{25c}
\end{align}
where in the final step we used \Equation{22c}.

Summarizing, the conditions on $m$ and $s$ are as follows (for any $\eps_0>0$)
\begin{align}
\nonumber m &\gtrsim \eps^{-2}(\log m)^3(\log n)^5 \gamma_2^2(\Theta, \rho_{\mathrm{Fin}})\\
&\hspace{.2in}{}+\eps^{-2}(\log m)^4(\log n)^5[d + \log\mathcal{N}(\Theta, \rho_{\mathrm{Fin}}, \eps_0)] \EquationName{26c}\\
\nonumber s&\gtrsim \eps^{-2}(\log m)^6(\log n)^4 + \eps^{-2}(\log m)^4(\log n)^5[\alpha \log \mathcal{N}(\Theta, \rho_{\mathrm{Fin}}, \eps_0)]^2\\
\nonumber &\hspace{.2in}{}+\eps^{-2}(\log m)^4(\log n)^5\eps_0^2\Big(\log\frac 1{\eps_0}\Big)d\\
&\hspace{.2in}{} + \eps^{-2}(\log m)^4(\log n)^5\Big[\int_0^{\eps_0}[\log\mathcal{N}(\Theta, \rho_{\mathrm{Fin}}, t)]^{1/2} dt\Big]^2 \EquationName{27c}
\end{align}

If $|\Theta| = N < \infty$, then $\log\mathcal{N}(\Theta, \rho_{\mathrm{Fin}}, t)\leq\log N$ and \Eqsub{26c}, \Eqsub{27c} turn into \Eqsub{11c}, \Eqsub{12c} for $\eps_0\rightarrow 0$.

\subsection{Manifolds}
Let $\mathcal{M}\subset\R^n$ be a $d$-dimensional manifold obtained as the image $\mathcal{M} = F(B_{\ell_2^d})$, for a smooth map $F:B_{\ell_2^d}\rightarrow \R^n$. More precisely, we assume that $\|F(x) - F(y)\|_2 \simeq \|x - y\|_2$ and the {\em Gauss map}, which sends $x\in\mathcal{M}$ to $E_x$, the tangent plane at $x$, is Lipschitz from the geodesic distance $\rho_{\mathcal{M}}$ to $\rho_{\mathrm{Fin}}$. Following \cite{Dirksen14}, we want to  ensure that the sparse matrix $\Phi$ satisfies
\begin{equation}
(1-\eps)|\gamma| \le |\Phi(\gamma)| \le (1+\eps)|\gamma|\EquationName{28c}
\end{equation}
for any $C^1$-curve $\gamma\subset \mathcal{M}$, where $|\cdot|$ denotes curve length. Note that \Equation{28c} is equivalent to requiring
\begin{equation}
1-\eps \le \|\Phi(v)\|_2 \le 1+\eps \EquationName{29c}
\end{equation}
for any tangent vector $v$ of $\mathcal{M}$ at a point $x\in \mathcal{M}$. Denote by
$$
\Theta = \{E_x : x\in \mathcal{M}\}
$$
the tangent bundle of $\mathcal{M}$, to which we apply the estimates on $m,s$ obtained above. By assumption,
$$
\rho_{\mathrm{Fin}}(E_x, E_y) \lesssim \rho_{\mathcal{M}}(x,y) \simeq \|F^{-1}(x) - F^{-1}(y)\|_2,
$$
so that for $0\le t \le 1/2$ by \Equation{volComp},
\begin{equation*}
\log \mathcal{N}(\Theta, \rho_{\mathrm{Fin}}, t) \lesssim \log \mathcal{N}(\mathcal{M}, \|\cdot\|_2, ct) \simeq \log \mathcal{N}(B_{\ell_2^d}, \|\cdot\|_2, t)\lesssim d\log\frac 1t.
\end{equation*}

In this application
\begin{equation}
\alpha = \max_{x\in\mathcal{M}} \max_{\substack{v\in E_x\\\|v\|_2 = 1}} \max_{1\le j\le n}|\inprod{v, e_j}| \EquationName{30c}
\end{equation}
We assume
\begin{equation}
\alpha \ll \frac 1{\sqrt{d}} \EquationName{31c}
\end{equation}
to make the below of interest (otherwise apply the result of \cite{Dirksen14}). Taking $\eps_0 = \alpha\sqrt{d}$ in \Eqsub{26c}, \Eqsub{27c}, it follows that \Equation{28c} may be ensured under parameter conditions
\begin{align}
m &\gtrsim \eps^{-2}(\log m)^4(\log n)^5 d\cdot \log\Big(\frac{1}{\alpha\sqrt{d}}\Big) \EquationName{32c}\\
s &\gtrsim \eps^{-2}(\log m)^6(\log n)^4 + \eps^{-2}(\log m)^4 (\log n)^7(\alpha d)^2 \EquationName{33c}
\end{align}
Thus for $\alpha = o(1/\sqrt{d})$, the condition on $s$ becomes non-trivial. Recall from Remark~\ref{rem:coherence} that $\alpha \ge \sqrt{d/n}$ and therefore the $\log(1/(\alpha\sqrt{d}))$ term in \Equation{32c} is at most $\log n$.

\begin{remark}
\textup{
Returning to the assumptions on $F$, consider $DF : B_{\ell_2^d}\rightarrow \mathcal{L}(\R^d, \R^n)$, the space of linear operators from $\R^d$ to $\R^n$. The first statement means that uniformly for $x\in B_{\ell_2^d}$,
\begin{equation}
c^{-1} \|\xi\|_2 \le \|DF(x)\xi\|_2 \le c\|\xi\|_2\text{ for } \xi\in\R^d \EquationName{34c}
\end{equation}
The second statement follows from requiring
\begin{equation}
\|DF(x) - DF(y)\|_{2\rightarrow 2} \le c\|x - y\|_2 \text{ for } x,y\in B_{\ell_2^d} \EquationName{35c}
\end{equation}
}
\textup{
Indeed, since
$E_x = DF(x)(\R^d), E_y = DF(y)(\R^d)$, it follows from \Equation{35c} 
$$
\inf_{\substack{u\in E_x\\ \|u\|_2 = 1}} \rho_{\ell_2^n}(u, E_y) \lesssim \|x - y\|_2
$$
implying
$$
\rho_{\mathrm{Fin}}(E_x, E_y) = \|P_{E_x} - P_{E_y} \|_{2\to 2} \lesssim \|x - y\|_2 .
$$
}
\end{remark}

\begin{remark}\RemarkName{bad-manifold}
\textup{
If $\alpha \ge 1/\sqrt{d}$, then necessarily $s\gg d$ in \Equation{33c} (up to polylogarithmic factors). Thus the manifold case is very different from the case of linear subspaces. We sketch a construction to demonstrate this.
}

\textup{
Let $n > d^{10}$. Denote by $0\le \varphi \le 1$ a smooth bump function on $\R$ such that 
\begin{equation}
\varphi(t) = t \text{ for } \frac 14\le t \le \frac 12,\ \mathrm{supp}(\varphi) \subset [0,1] \EquationName{36c}
\end{equation}
}
\textup{
By the lower bound in \Equation{volComp}, there is a collection $\{a_\beta\}_{1\le \beta\le 2^d}$ of $2^d$ points in $B_{\ell_2^d}(0, 1/2) \eqdef B_{1/2}\subset\R^d$ such that
$$
\|a_{\beta'} - a_{\beta} \|_2 > \frac 1{10}\text{ for } \beta'\neq \beta
$$
and let $(\eta_\beta)_{1\le \beta\le 2^d}$ be any collection of unit vectors in $\R^n$.
}
\textup{
Consider the map $f:\R^d\rightarrow \R^n$ defined by
\begin{equation}
f(x) = \sum_{\beta=1}^{2^d} \varphi(10^4\|x - a_\beta\|_2^2)\eta_\beta \EquationName{37c}
\end{equation}
Thus by construction, the summands in \Equation{37c} are disjointly supported functions of $x$. Clearly
\begin{equation}
Df(x)\xi = 2\cdot 10^4\sum_\beta \varphi'(10^4\|x - a_\beta\|_2^2) \inprod{x - a_\beta, \xi}\eta_\beta \EquationName{38c}
\end{equation}
implying that
$$
\|Df(x)\|_{2\rightarrow 2} \le C
$$
and
\begin{equation}
\|Df(x) - Df(y)\|_{2\rightarrow 2} \le C\|x - y\|_2 \EquationName{39c}
\end{equation}
}

\textup{
Next, let $\theta_1,\ldots,\theta_d\in\R^n$ be orthogonal vectors such that
\begin{equation}
\|\theta_j\|_2 = 1, \|\theta_j\|_\infty = \frac 1{\sqrt{n}} \EquationName{40c}
\end{equation}
and define the map
\begin{align*}
&F : B_{\ell_2^d}\subset\R^d \rightarrow \R^{2n} = \R^n\times \R^n\\
&F(x) = \Big(\sum_{j=1}^d x_j\theta_j, f(x)\Big)
\end{align*}
}
\textup{
In view of \Equation{39c}, $F$ satisfies conditions \Eqsub{34c},\Eqsub{35c}. Also, by \Eqsub{38c},\Eqsub{40c}
\begin{equation}
\alpha \lesssim \sup_{x\in B_{\ell_2^d}} \max_{\|\xi\|_2 = 1}\|DF(x)\xi\|_\infty \lesssim \sqrt{\frac dn} + \max_\beta \|\eta_\beta\|_\infty \EquationName{41c}
\end{equation}
}

\textup{
Let $\mathcal{M} = F(B_{\ell_2^d})$. Fix some $1\le \beta\le 2^d$ and let for $1/200 \le t \le 1/(100\sqrt{2})$
$$
\gamma(t) = F(a_\beta + t e_1) = \Big(\sum_{j=1}^d a_{\beta,j}\theta_j + t\theta_1, 10^4t^2\eta_\beta\Big)
$$
where we used \Equation{36c}. Thus $\gamma$ is a $C^1$-curve in $\mathcal{M}$ and
\begin{equation}
\gamma'\Big(\frac 1{200}\Big) = (\theta_1, 100\eta_{\beta}) \EquationName{42c}
\end{equation}
}

\textup{
Let $\Phi$ be a sparse $m\times 2n$ matrix for which \Equation{28c} holds. Then $\Phi$ has to satisfy in particular
$$
(1-\eps)(1+10^4)^{1/2} \le \|\Phi(\theta_1, 100\eta_\beta)\|_2 \le (1+\eps)(1+10^4)^{1/2}
$$
and hence
\begin{equation}
\|\Phi(\theta_1, 100\eta_\beta)\|_2 < 10^3 \text{ for all } 1\le \beta\le 2^d \EquationName{43c}
\end{equation}
}
\textup{
Choose $k$ such that
$$
2^d > \binom{n}{k}2^k,\text{ i.e. } k\simeq \frac d{\log n}
$$
and let $(\eta_\beta)$ be the collection of $2^k\binom{n}{k}$ vectors in $\R^n$ of the form
\begin{equation}
\eta = \frac 1{\sqrt{k}} \sum_{j\in I} \pm e_j\text{ with } I\subset \{1,\ldots,n\}, |I| = k \EquationName{44c}
\end{equation}
}
\textup{
By \Equation{43c}, $\Phi$ needs to satisfy
\begin{equation}
\|\Phi(0,\eta)\|_2 \le 20 \EquationName{45c}
\end{equation}
for all vectors $\eta$ of the form \Equation{44c}. But if $m < n/d$, \Equation{45c} implies that $s\ge k/400 \gtrsim d/\log n$. On the other hand, \Equation{41c} gives
$$
\alpha \lesssim \sqrt{\frac dn} + \frac 1{\sqrt{k}} \simeq\  \sqrt{\frac{\log n}d}
$$
}
\end{remark}

\subsubsection{Geodesic distances}\SectionName{geodesic}

In this section we show that not only do sparse maps preserve curve lengths on manifolds, but in fact they preserve geodesic distances.

\begin{lemma}\LemmaName{lem1}
Fix $x\in\R^n$ such that $|x_j| \ge 1$ for $j\in J\subset\{1,\ldots,n\}, |J| = r$. Then 
\begin{equation}
\Big|\Big\{ 1\le i\le m: \sum \Phi_{ij}^2 x_j^2 \ge \frac 1s \Big\} \Big| > \min\Big\{\frac{sr}3, cm\Big\} \eqdef r_1 
\end{equation}
with probability (with respect to $\Phi$) at least
\begin{equation}
1 - 2^{-sr}
\end{equation}
\end{lemma}
\begin{proof}
Fix $I\subset \{1,\ldots,m\}$, $|I|\le r_1$ and assume $\sum_j \Phi_{ij}^2 x_j^2 < 1/s$ for $i\notin I$. This means that for any $j\in J$, $S = \eqdef\{i ; \Phi_{ij} \neq 0\} \subset I$. The probability (with respect to $\Phi$) that $\{i ; \Phi_{ij} \neq 0\} \subset I$ (with $j$ fixed) is (by \Equation{neg-cor})
\begin{align*}
\sum_{\substack{K\subset I\\|K|=s}}\Pr_{\Phi}(S = K) &= \sum_{\substack{K\subset I\\|K|=s}} \E \prod_{i\in K} \delta_{ij}\\
{}&\le \sum_{\substack{K\subset I\\|K|=s}} \Big(\frac sm\Big)^s\\
{}&\le \binom{|I|}{s} \cdot \Big(\frac sm\Big)^s\\
{}&\le \Big(\frac{e\cdot |I|}m\Big)^s
\end{align*}

Since for different $j$ the events are independent, it follows the probability that
$$
\Big\{ 1\le i\le m ; \sum \Phi_{ij}^2 x_j^2 \ge \frac 1s \Big\} \subset I
$$
is bounded by $((e|I|/m)^s)^r \le (er_1/m)^{sr}$. Taking a union bound over all $I\subset\{1,\ldots,m\},|I|\le r_1$ gives by the choice of $r_1$
$$
\binom{m}{|I|}\cdot \Big(\frac{er_1}m\Big)^{sr} \le  \Big(\frac {em}{r_1}\Big)^{r_1} \Big(\frac{er_1}m\Big)^{sr} \le \Big(\frac{e^2r_1}m\Big)^{2sr/3} < 2^{-sr}
$$
\end{proof}
For the following lemma recall that for $a\in\R^n$ and $(\sigma_j)$ a Rademacher vector
\begin{equation}
\EquationName{Berger}
\Pr_\sigma\Big(\Big|\sum_j a_j \sigma_j\Big| < \frac 12 \|a\|_2\Big) < \frac 45. 
\end{equation}
This is a consequence of the Paley-Zygmund inequality, see e.g.\ \cite[Theorem 3.6]{Berger97}.
\begin{lemma}\LemmaName{lem2}
Let $x\in\R^n$ satisfy the assumption of \Lemma{lem1}. Then
\begin{equation}
\Pr_{\Phi}\Big(\|\Phi x\|_2 < \frac 1{2\sqrt{s}}\Big) < 2^{-cr_1} \EquationName{three}
\end{equation}
\end{lemma}
\begin{proof}
We may assume that $\sum \Phi_{ij}^2x_j^2 \ge 1/s$ for $i$ in a subset $I\subset \{1,\ldots,m\}$, $|I| \ge r_1$. Exploiting the random signs of $\Phi_{ij}$, we find by \Equation{Berger} for each $i\in I$
\begin{equation}
\Pr\Big(\Big|\sum_j \Phi_{ij}x_j\Big| < \frac 1{2\sqrt{s}}\Big) < \frac 45\EquationName{four}
\end{equation}
If $\|\Phi x\|_2 < 1/(2\sqrt{s})$, then $|\sum_j \Phi_{ij}x_j| < 1/(2\sqrt{s})$ for all $i$, in particular for all $i\in I$. By \Equation{four} the probability for this event is at most
$$
\Big(\frac 45\Big)^{|I|} < 2^{-cr_1} ,
$$
proving \Equation{three}.
\end{proof}

\begin{corollary}\CorollaryName{cor1}
Let $\xi\subset\R^n$ be a set of vectors $x$ with following properties:
\begin{enumerate}
\item Each $x\in\xi$ has a decomposition $x = x' + x''$, $x'\in \xi'$ and there is a set $J = J_{x'}\subset \{1,\ldots,n\}$, $|J|\ge r$ so that $|x_j'| \ge 1$ for $j\in J$. Moreover 
\begin{equation}
\|x''\|_2 < \frac 1{10n} \EquationName{smallnorm}
\end{equation}
\item \begin{equation}|\xi'| < 2^{c\min\{sr, m\}}\EquationName{xiprime}\end{equation}
\end{enumerate}
Then
\begin{equation}
\|\Phi x\|_2 > \frac 1{4\sqrt{s}}\text{ for all } x\in\xi \EquationName{normlb}
\end{equation}
with probability at least $1 - 2^{-c\min\{sr,m\}}$.
\end{corollary}
\begin{proof}
Write
$$
\|\Phi x\|_2 \ge \|\Phi x'\|_2 - \|\Phi x''\|_2 \ge \|\Phi x'\|_2 - \sqrt{n}\cdot \|x''\|_2 \ge \|\Phi x'\|_2 - \frac 1{10\sqrt{n}} .
$$
since clearly $\|\Phi\|_2 \le \sqrt{n}$. Next, \Lemma{lem2} and a union bound will ensure that $\|\Phi x'\|_2 \ge 1/(2\sqrt{s})$ for all $x'\in\xi'$ with the desired probability.
\end{proof}

\begin{lemma}\LemmaName{smalllb}
Let $s\ge c(\log n)^2$. Let $\xi\subset \R^n$ be a finite set of unit vectors satisfying
\begin{equation}
\log|\xi|< cm
\end{equation}
Then with probability at least $1 - e^{-cs}$ for some constant $c>0$, $\Phi$ satisfies
\begin{equation}
\|\Phi x\|> e^{-c(\log n)^2}\text{ for } x\in\xi .
\end{equation}
\end{lemma}
\begin{proof}
Given $x\in\R^n$, let $x^*$ be the decreasing rearrangement of $|x_i|$. Let $K = c\log n$. We partition $\xi$ as 
$$
\xi = \bigcup_{\ell} (\xi_\ell\backslash (\xi_0\cup \xi_1 \cup \ldots \cup \xi_{\ell-1}))
$$
where $\xi_{-1}=\emptyset$, and for $\ell\geq 0$ satisfying $2^\ell K^2 < m$,
$$\xi_\ell = \{x \in\xi : x^*_{m/(2^\ell K^2)} > n^{-K + 2\ell}\}$$
For the vectors $x\in\xi_0$, apply \Corollary{cor1} with $x=x', r = m/K^2$ (after rescaling by $n^{-K}$). Since $csr > cm> \log|\xi| \ge \log|\xi_0|$, \Eqsub{xiprime} holds and by \Equation{normlb}, we ensure that $\|\Phi x\|_2 > n^{-K}/(4\sqrt{s})$ for all $x\in\xi_0$.

Consider the set $\xi_\ell' = \xi_\ell\backslash(\xi_0\cup\xi_1 \cup\ldots\cup\xi_{\ell-1})$. Thus each $x\in\xi_\ell'$ has a decomposition
$$
x = y + z
$$
where $y$ is obtained by considering the $m/(2^{\ell-1}K^2)$ largest coordinates of $x$ and $\|z\|_\infty\le x^*_{m/(2^{\ell-1}K^2)} \le n^{-K+2\ell-2}$ since $x\notin \xi_{\ell-1}$.

Note also that
$$
y\in \bigcup_{\substack{S\subset\{1,\ldots,n\}\\|S| = \frac m{2^{\ell-1}K^2}}} B_{X_S}, \text{ with } X_S = \{e_j : j\in S\} .
$$
Let $\mathfrak{f}_S \subset B_{X_S}$ be a finite subset such that
\begin{align}
\nonumber \mathrm{dist}(y, \mathfrak{f}_S) &< n^{-K}\text{ for all } y\in B_{X_S}\\
\log|\mathfrak{f}_S| &\lesssim \frac{m}{2^{\ell-1}K^2}\log n^K = \frac{m\log n}{2^{\ell-1}K} \EquationName{finite}
\end{align}

Hence $y = x' + w$ with
\begin{equation}
x' \in \bigcup_S \mathfrak{f}_S,\ \|w\|_2 < n^{-K} \EquationName{contained}
\end{equation}

Apply \Corollary{cor1} to the set $n^{K - 2\ell}\xi'_\ell$, considering the decomposition
$$
n^{K-2\ell}x = n^{K-2\ell}x' + n^{K-2\ell}(w + z)
$$
satisfying by \Equation{contained}
\begin{align*}
\|n^{K - 2\ell}(w+z)\|_2 &< n^{-2\ell} + n^{K-2\ell}\|z\|_\infty\\
{}&< n^{-2\ell} + n^{-2} < \frac 1{10 n}
\end{align*}
which is condition \Eqsub{smallnorm}.

Also $n^{K-2\ell}x'\in\xi'$, where by \Equation{finite} (and choice of $s,K$)
$$
\log|\xi'| \le \log\binom{n}{m/(2^{\ell-1}K^2)} + \frac{m\log n}{2^{\ell-1}K} \lesssim \frac{m\log n}{2^{\ell-1}K} < c\min\Big\{\frac{sm}{2^\ell K^2}, m\Big\} ,
$$
which is condition \Eqsub{xiprime} with $r = m/(2^{\ell}K^2)$.

Let $J_x$ be the set of $r = m/(2^\ell K^2)$ largest coordinates of $x$ (which are also the $r$ largest coordinates of $y$). Hence, for $j\in J_x$
\begin{align*}
n^{K-2\ell}|x_j'| &\ge n^{K-2\ell}|y_j| - n^{K-2\ell}|w_j|\\
{}&= n^{K-2\ell}|x_j| - n^{K-2\ell}|w_j|\\
{}& > n^{K-2\ell}n^{-K+2\ell} - n^{K-2\ell}n^{-K} > \frac 12
\end{align*}
since $x\in\xi_\ell$. By \Equation{normlb}
$$
\|\Phi x\| > n^{-K+2\ell}\cdot \frac 1{4\sqrt{s}}\text{ for all } x\in\xi_\ell
$$
with probability at least $1 - e^{-c\min\{m, sr\}} \geq 1 - e^{-cs}$, since $m/(2^\ell K^2) \ge 1$.
\end{proof}

Let $\mathcal{M}\subset\R^n$ be a $d$-dimensional manifold obtained as a bi-Lipschitz image of the unit ball $B_{\ell_2^d}$, i.e.\ $F:B_{\ell_2^d}\rightarrow \mathcal{M}$ satisfies
\begin{equation}
\|F(x) - F(y)\|_2 \simeq \|x - y\|_2 .\EquationName{bilip}
\end{equation}

We assume moreover $F$ is smooth, more specifically $DF:B_{\ell_2^d}\rightarrow \mathcal{L}(\R^d, \R^n)$ (i.e.\ linear maps from $\R^d$ to $\R^n$ under operator norm) is Lipschitz, i.e.,
\begin{equation}
\|DF(x) - DF(y)\|_{\ell_2^d\rightarrow\ell_2^n} \lesssim \|x - y\|_2\EquationName{lipschitz-derivative}
\end{equation}
\begin{lemma}
Let $\mathcal{M}$ be as above. Assume
\begin{align}
\nonumber m&\gtrsim d(\log n)^2\\
s&\gtrsim (\log n)^2 \EquationName{parameters}
\end{align}
Then with probability at least $1 - e^{-cs}$ for some constant $c>0$, $\Phi|_{\mathcal{M}}$ is bi-Lipschitz, and more specifically
\begin{equation}
\|\Phi(x) - \Phi(y)\|_2 \ge e^{-c(\log n)^2}\cdot \|x-y\|_2\text{ for }x,y\in\mathcal{M} \EquationName{terrible-distortion}
\end{equation}
\end{lemma}
\begin{proof}
We treat separately the pairs $x,y\in\mathcal{M}$ which are at a ``large'' and ``small'' distance from each other. Fix $\eps_1 > 0$ and $\eps_2>\eps_1$ to be specified later. Let $A_{\eps_1} \subset\mathcal{M}$ be an $\eps_1$-net for $\cM$. By \Equation{bilip} and \Equation{volComp}, we can assume that
$$
\log|A_{\eps_1}| \lesssim d\log(1/\eps_1).
$$
Assume that $x,y\in\mathcal{M}$, $\|x-y\|_2>\eps_2$. Take $x_1,y_1\in A_{\eps_1}$ s.t.\ $\|x - x_1\|_2 < \eps_1$, $\|y - y_1\|_2 < \eps_1$. Since $\Phi$ is linear 
$$
\|\Phi(x) - \Phi(y)\|_2 \ge \|\Phi(x_1 - y_1)\|_2 - 2\sqrt{n}\eps_1
$$
Apply \Lemma{smalllb} to the set
$$
\xi = \Big\{\frac{x_1 - y_1}{\|x_1 - y_1\|_2} : x_1,y_1 \in A_{\eps_1}\Big\} .
$$
Assuming
\begin{equation}
m \gtrsim d\log(1/\eps_1)> c\log|A_{\eps_1}|\EquationName{basic-entropy}
\end{equation}
we ensure $\Phi$ to satisfy
$$
\Big\|\Phi\Big(\frac{x_1 - y_1}{\|x_1 - y_1\|_2}\Big)\Big\|_2 > e^{-c(\log n)^2}\text{ for }x_1,y_1\in A_{\eps_1} .
$$

Therefore, if $x,y\in\mathcal{M}$, $\|x - y\|_2 > \eps_2$
\begin{align}
\nonumber \|\Phi(x) - \Phi(y)\|_2 &\ge e^{-c(\log n)^2}\cdot \|x_1 - y_1\|_2 - 2\sqrt{n}\frac{\eps_1}{\eps_2}\cdot\|x-y\|_2\\
{}&\ge \frac 12 e^{-c(\log n)^2}\cdot\|x - y\|_2
\end{align}
choosing $\eps_2 = 5\sqrt{n} e^{c(\log n)^2}\eps_1$. This takes care of large distances.

In order to deal with small distances, we first ensure that
\begin{equation}
\|\Phi(v)\|> e^{-c(\log n)^2}\text{ for any unit tangent vector }v\text{ of } \mathcal{M} \EquationName{tangents}
\end{equation}

Consider
\begin{equation}
\mathcal{T}_{\eps_1} = \bigcup_{x_1\in A_{\eps_1}} \Big\{v\in T_{x_1}: \|v\| = 1\Big\}\EquationName{all-tangents}
\end{equation}

Under the assumption \Eqsub{basic-entropy} on $m$ (and making an $\eps_1$-discretization of each $\{v\in T_{x_1} : \|v\|_2 = 1\}$), another application of \Lemma{smalllb} will ensure that \Eqsub{tangents} holds for all $v\in\mathcal{T}_{\eps_1}$. Next, if $x\in\mathcal{M}$, $x_1\in A_{\eps_1}$, $\|x - x_1\|_2 < \eps_1$, it follows from \Equation{lipschitz-derivative} that $\|DF(x) - DF(x_1)\|_{\ell_2\to\ell_2} \lesssim \eps_1$ and hence $\rho_{\mathrm{Fin}}(T_x, T_{x_1}) \lesssim \eps_1$. Therefore, if $v\in T_x$, $\|v\|_2 = 1$, there is $v_1\in\mathcal{T}_{\eps_1}$ s.t.\ $\|v - v_1\|_2 \lesssim \eps_1$. Therefore
$$
\|\Phi(v)\|_2 \ge \|\Phi(v_1)\|_2 - \sqrt{n}\eps_1 \ge e^{-c(\log n)^2} - \sqrt{n}\eps_1 > \frac 12 e^{-c(\log n)^2}
$$
since $\eps_1 < n^{-1/2} e^{-c(\log n)^2}/2$. Thus $\Phi$ satisfies \Eqsub{tangents}.

Assume $x,y\in\mathcal{M}$, $\|x-y\|_2<\eps_2$. Let $u = F^{-1}(x), w = F^{-1}(y) \in B_{\ell_2^d}$. By \Equation{bilip}, $\|u-w\|_2 \simeq \|x-y\|_2$. Write
\begin{align*}
y-x = F(w) - F(u) &= \int_0^1 \partial_t F(tw + (1-t)u) dt\\
{}&= \int_0^1 DF(tw + (1-t)u)(w - u)dt.
\end{align*}
Hence, again invoking \Equation{lipschitz-derivative},
\begin{align}
& \|F(w) - F(u) - DF(u)(w-u)\|_2 \nonumber \\
& \qquad \le \|w-u\|_2\cdot \sup_t \|DF(tw + (1-t)u) - DF(u)\|_{\ell_2\to\ell_2} \nonumber\\
& \qquad \lesssim \|u - w\|_2^2. \EquationName{squared-dist}
\end{align}

Denote
$$
v = DF(u)\Big(\frac{w-u}{\|w-u\|_2}\Big) \in T_u .
$$
By \Equation{squared-dist} and using $\|\Phi\|\leq \sqrt{n}$,
\begin{align*}
\|y - x - \|u-w\|_2 v\|_2 &< c\eps_2\|x-y\|_2\\
\|\Phi(y) - \Phi(x) - \|u-w\|_2\Phi(v)\|_2 &< c\sqrt{n}\eps_2\|x-y\|_2
\end{align*}
and by \Eqsub{tangents}
\begin{align*}
\|\Phi(x) - \Phi(y)\|_2 &\ge \|u-w\|_2 e^{-c(\log n)^2} - c\sqrt{n}\eps_2\|x-y\|_2\\
{}&\ge \Big(ce^{-c(\log n)^2} - c\sqrt{n}\eps_2\Big)\|x-y\|_2\\
{}&\gtrsim e^{-c(\log n)^2}\|x-y\|_2
\end{align*}
provided
\begin{equation}
\eps_2 < c\frac 1{\sqrt{n}} e^{-c(\log n)^2}
\end{equation}

Thus we may take $\log(1/\eps_1),\log(1/\eps_2)\simeq (\log n)^2$ and condition \Eqsub{basic-entropy} will hold for
$$
m\gtrsim d(\log n)^2.
$$
\end{proof}

\begin{theorem}
Let $\mathcal{M}$ be as above and $m,s$ satisfy \Eqsub{parameters}. Assume moreover that $m,s$ satisfy the appropriate conditions to ensure that
\begin{equation}
1-\eps \le \|\Phi(v)\|_2 \le 1+\eps \EquationName{curves}
\end{equation}
for all unit tangent vectors $v$ of $\mathcal{M}$. Then with probability at least $1 - e^{-cs}$ for some constant $c>0$, $\Phi$ preserves geodesic distances up to factor $1+\eps$.
\end{theorem}
\begin{proof}
By \Eqsub{curves}, $(1-\eps)|\gamma| < |\Phi(\gamma)| < (1+\eps)|\gamma|$ for any $C^1$-curve $\gamma$ in $\mathcal{M}$. Let $\rho_{\mathcal{M}}$ refer to the geodesic distance in $\mathcal{M}$. Clearly
\begin{equation}
\rho_{\Phi(\mathcal{M})}(\Phi(x), \Phi(y)) \le (1+\eps)\rho_{\mathcal{M}}(x,y)
\end{equation}
from the above. We need to show the reverse inequality. Let $\gamma_1$ be a $C^1$-curve in $\Phi(\mathcal{M})$ joining $\Phi(x),\Phi(y)$ such that $\rho_{\Phi(\mathcal{M})}(\Phi(x), \Phi(y)) = |\gamma_1|$. Since $\Phi$ satisfies \Eqsub{terrible-distortion}, $\Phi|_{\mathcal{M}} : \mathcal{M}\rightarrow\Phi(\mathcal{M})$ is bi-Lipschitz and hence a diffeomorphism (since $\Phi$ is also smooth). It follows that $\gamma = \Phi^{-1} \gamma_1$ is a $C^1$-curve joining $x$ and $y$ and
$$
\rho_{\mathcal{M}}(x,y) \le |\gamma| \le (1+\eps) |\Phi(\gamma)| = (1+\eps)|\gamma_1| = (1+\eps)\rho_{\Phi(\mathcal{M})}(\Phi(x), \Phi(y)) .
$$
\end{proof}

\section{Discussion}\SectionName{discussion}
We have provided a general theorem which captures sparse dimensionality reduction in Euclidean space and qualitatively unifies much of what we know in specific applications. There is still much room though for quantitative improvement. We here list some known quantitative shortcomings of our bounds and discuss some avenues for improvement in future work.

First, our dependence on $1/\eps$ in $s$ in all our theorems is, up to logarithmic factors, quadratic. Meanwhile the works \cite{KN14,NN13b} show that the correct dependence in the case of small $m$ should be linear. Part of the reason for this discrepancy may be our use of the chaining result of \cite{KMR14}. Specifically, chaining is a technique that in general converts tail bounds into bounds on the expected supremum of stochastic processes (see \cite{Talagrand05} for details). Perhaps one could generalize the tail bound of \cite{KN14} to give a broad understanding of the decay behavior for the error random variable in the SJLT as a function of $s,m$ then feed such a quantity into a chaining argument to improve our use of \cite{KMR14}.

Another place where we lost logarithmic factors is in our use of the duality of entropy numbers \cite[Proposition 4]{BPST89} in \Equation{3.3}. It is believed that in general if $K,D$ are symmetric convex bodies and $\mathcal{N}(K,D)$ is the number of translations of $D$ needed to cover $K$ then
$$
\mathcal{N}(D^\circ, a K^\circ)\lesssim \mathcal{N}(K, D) \lesssim \mathcal{N}(D^\circ, a^{-1} K^\circ)
$$
for some universal constant $a>0$. This is known as Pietsch's duality conjecture \cite[p.\ 38]{P72}, and unfortunately resolving it has been a challenging open problem in convex geometry for over 40 years.

We have also lost logarithmic factors in our use of dual Sudakov minoration, which bounds $\sup_{t>0} t [\log \mathcal{N}(B_E, \|\cdot\|_X, t)]^{1/2}$ by a constant times $\E_g \|P_E g\|_X$ for any norm $\|\cdot\|_X$ even though what we actually wish to bound is the $\gamma_2$ functional. Passing from this functional via \Equation{Dudley} to $\sup_{t>0} t [\log \mathcal{N}(B_E, \|\cdot\|_X, t)]^{1/2}$ costs us logarithmic factors. The majorizing measures theory (see \cite{Talagrand05}) shows that the loss can be avoided when $\|\cdot\|_X$ is the $\ell_2$ norm; it would be interesting to explore how the loss can be avoided in our case.

Finally, note that the SJLT $\Phi$ considered in this work is a randomly signed adjacency matrix of a random bipartite graph with $m$ vertices in the left bipartition, $n$ in the right, and with all right vertices having equal degree $s$. Sasha Sodin has asked in personal communication whether taking a random signing of the adjacency matrix of a random biregular graph (degree $s$ on the right and degree $ns/m$ on the left) can yield improved bounds on $s,m$ for some $T$ of interest. We leave this as an interesting open question.

\bigskip

\paragraph{Acknowledgement}
We would like to thank Huy L.\ Nguy$\tilde{\hat{\mbox{e}}}$n for pointing out the second statement in Lemma~\ref{lem:compare}. Part of this work was done during a visit of the second-named author to Harvard University. He wishes to thank the Theory of Computation group for its hospitality.

\newcommand{\etalchar}[1]{$^{#1}$}

\appendix

\section{Tools from probability theory}
\label{sec:probTools}

We collect some useful tools from probability theory that are used throughout the text. For further reference we record an easy consequence of Markov's inequality.
\begin{lemma}
\label{lem:MomentsToTails}
If $\xi$ is a real-valued random variable satisfying
$$(\E|\xi|^p)^{1/p} \leq a_1 p^2 + a_2 p^{3/2} + a_3 p + a_4 p^{1/2} + a_5, \qquad \mathrm{for \ all} \ p\geq p_0,$$
for some $0\leq a_1,a_2,a_3,a_4,a_5<\infty$, then
\begin{equation}
\label{eqn:MomentsToTails}
\bP(|\xi|\geq e(a_1 w^2 + a_2 w^{3/2} + a_3 w+ a_4 \sqrt{w} + a_5)) \leq \exp(-w) \qquad (w\geq p_0).
\end{equation}
\end{lemma}
Let us call $\si=(\si_i)_{i=1}^n$ a Rademacher vector if its entries are independent Rademacher random variables. Khintchine's inequality states that for any $1\leq p<\infty$
\begin{equation}
\EquationName{KI}
(\E_{\si}|\langle \si,x\rangle|^p)^{1/p} \leq C_p \|x\|_2 \qquad (x\in \R^n),
\end{equation}
where $C_p\leq \sqrt{p}$. The noncommutative Khintchine inequality \cite{LP86,LPP91} says that for any $A_1,\ldots,A_n\in\R^{m\ti n}$ and $1\leq p<\infty$
\begin{equation*}
\Big(\E\Big\|\sum_{i=1}^n \si_i A_i\Big\|_{S_p}^p\Big)^{1/p} \leq C\sqrt{p} \max\Big\{\Big\|\Big(\sum_i A_i^* A_i\Big)^{1/2}\Big\|_{S_p},\Big\|\Big(\sum_i A_i A_i^*\Big)^{1/2}\Big\|_{S_p}\Big\},  
\end{equation*}
where $\|\cdot\|_{S_p}$ is the $p$-th Schatten norm. In particular, if $p\geq \max\{\log m,\log n\}$, then
\begin{equation}
\EquationName{NCKI}
\Big(\E\Big\|\sum_{i=1}^n \si_i A_i\Big\|^p\Big)^{1/p} \leq C\sqrt{p} \max\Big\{\Big\|\Big(\sum_i A_i^* A_i\Big)^{1/2}\Big\|,\Big\|\Big(\sum_i A_i A_i^*\Big)^{1/2}\Big\|\Big\}  
\end{equation}
as $\|\cdot\|\leq \|\cdot\|_{S_p}\leq e\|\cdot\|$ in this case. Khintchine's inequality and its noncommutative version are frequently used in combination with symmetrization: if $\zeta_1,\ldots,\zeta_n$ are $X$-valued random variables and $\si$ is a Rademacher vector, then for any $1\leq p<\infty$,
\begin{equation}
\label{eq:symmetrization}
\Big(\E_{\zeta}\Big\|\sum_{i=1}^n \zeta_i - \E\zeta_i\Big\|_X^p\Big)^{1/p} \leq \Big(\E_{\zeta}\E_{\si}\Big\|\sum_{i=1}^n \si_i\zeta_i\Big\|_X^p\Big)^{1/p}.
\end{equation}
The following decoupling inequality is elementary to prove (see e.g.\ \cite[Theorem 8.11]{FoR13}). Let $\cM\subset\R^{n\ti n}$, let $\si$ be an $n$-dimensional Rademacher vector and $\si'$ an independent copy. Then, for any $1\leq p<\infty$
\begin{equation}
\label{eqn:decRad}
(\E_{\si}\sup_{M\in \cM}|\si^* M\si - \E(\si^* M\si)|^p)^{1/p} \leq 4(\E_{\si,\si'}\sup_{M\in \cM}|(\si')^* M\si|^p)^{1/p}. 
\end{equation}
A special case of this bound, combined with Khintchine's inequality, implies the following.
\begin{lemma}
\label{lem:Lpl2}
Let $A\in \R^{m\ti n}$ and let $\si$ be an $n$-dimensional Rademacher vector. For any $p\geq 1$,
$$(\E\|A\si\|_2^p)^{1/p} \leq \|A\|_F + 2\sqrt{2p}\|A\|.$$
\end{lemma}
\begin{proof}
By decoupling,
\begin{align*}
(\E\|A\si\|_2^p)^{2/p} & \leq (\E_{\si}|\si^*A^*A\si - \E\si^*A^*A\si|^{p/2})^{2/p} + \E(\si^*A^*A\si)\\
& \leq 4(\E_{\si,\si'}|(\si')^*A^*A\si|^{p/2})^{2/p} + \|A\|_F^2
\end{align*}
and therefore Khintchine's inequality implies that
\begin{align*}
(\E\|A\si\|_2^p)^{2/p} & \leq 4\sqrt{p/2} (\E_{\si}\|A^*A\si\|_2^{p/2})^{2/p} + \|A\|_F^2 \\
& \leq 2\sqrt{2p} \|A\| (\E_{\si}\|A\si\|_2^p)^{1/p} + \|A\|_F^2.
\end{align*}
Solving this quadratic inequality yields the claim.
\end{proof}
To conclude we collect some tools to estimate covering numbers. Given two closed sets $U,V\subset \R^n$, we let the \emph{covering number} $\cN(U,V)$ be the minimal number of translates of $V$ needed to cover $U$. If $V$ is closed, convex and symmetric, then we can associate with it a semi-norm
$$\|x\|_V = \inf\{t>0 \ : \ x\in tV\}.$$
In this case it is customary to write
$$\cN(U,\|\cdot\|_V,\eps) := \cN(U,\eps V) \qquad (\eps>0).$$
Conversely, if $\|\cdot\|$ is a semi-norm, then $V=\{x\in \R^n \ : \ \|x\|\leq 1\}$ is closed, convex and symmetric, and $\|\cdot\|_V=\|\cdot\|$.\par 
As a first tool we state two elementary bounds that follow from volumetric comparision. If $\|\cdot\|$ is any semi-norm on $\R^n$ and $B_{\|\cdot\|}$ is the associated unit ball, 
\begin{equation}
\Big(\frac{1}{\eps}\Big)^n \leq \cN(B_{\|\cdot\|},\|\cdot\|,\eps) \leq \Big(1+\frac{2}{\eps}\Big)^n \qquad (0<\eps\leq 1). \EquationName{volComp}
\end{equation}
The following is known as Sudakov minoration or the dual Sudakov inequality \cite[Proposition 4.2]{BLM89}, \cite{PTJ86}. \Equation{polar-def} in Appendix~\ref{sec:convex-tools} gives a definition of the polar $V^{\circ}$.
\begin{lemma}
\label{lem:dualSudakov}
Let $V\subset \R^n$ be closed, convex and symmetric and let $g$ be a standard Gaussian vector. Then,
$$\sup_{\eps>0} \eps[\log \cN(B_{\ell_2^n},\|\cdot\|_V,\eps)]^{1/2} \lesssim \E\sup_{x\in V^{\circ}}\langle g,x\rangle.$$
\end{lemma}
We will also use the following duality for covering numbers from \cite[Proposition 4]{BPST89}.
\begin{lemma}
\label{lem:BPST}
Let $U,V\subset \R^n$ be closed, bounded, convex and symmetric. For every $\theta\geq \eps$,
$$\cN(U,\|\cdot\|_V,\eps)\leq \cN(U,\|\cdot\|_V,\theta)[\cN(V^{\circ},\|\cdot\|_{U^{\circ}},\eps/8)]^r,$$
with $r\leq (2^7 T_2(\|\cdot\|_V))^2(1+\log(\theta/\eps))$.
\end{lemma}
Finally, we state Maurey's lemma \cite{Pi80,Carl85} and its usual proof.
\begin{lemma}
\label{lem:Maurey}
Let $\|\cdot\|$ be a semi-norm and $T_2(\|\cdot\|)$ be its type-$2$ constant. Let $\Omega$ be a set of points $x$ each with $\|x\| \le 1$. Then for any $z\in \mathrm{conv}(\Omega)$ and any integer $\ell>0$ there exist $y_1,\ldots,y_\ell\in\Omega$ with
$$
\Big\|z - \frac 1{\ell}\sum_{i=1}^\ell y_i\Big\| \le \frac{T_2(\|\cdot\|)}{\sqrt{\ell}}
$$
\end{lemma}
\begin{proof}
Write $z = \sum_i \lambda_i x_i$ for $0<\lambda_i\le 1$, $\sum_i \lambda_i = 1$, $x_i\in \Omega$. Let $y = (y_1,\ldots,y_\ell)$ have i.i.d.\ entries from the distribution in which $y_j = x_j$ with probability $\lambda_j$. Let $\sigma$ be a Rademacher vector and let $g$ be a Gaussian vector. By symmetrization,
\begin{align*}
\E_y\Big\|z - \frac 1{\ell}\sum_{i=1}^\ell y_i\Big\| & \le \frac{2}{\ell}\E_{\si,y}\Big\|\sum_{i=1}^\ell \sigma_iy_i\Big\|\\
{}& = \frac{2}{\ell}\E_{\si,y}\Big\|\E_g\sum_{i=1}^\ell \sigma_i|g_i|y_i\Big\| \lesssim \frac 1{\ell}\E_{y,g}\Big\|\sum_{i=1}^\ell g_iy_i\Big\| \le \frac{T_2(\|\cdot\|)}{\sqrt{\ell}}.
\end{align*}
\end{proof}

\section{Tools from convex analysis}\label{sec:convex-tools}

In this appendix we recall some basic facts from convex analysis. More details can be found in e.g.\ \cite{HiL01,Roc70}. For any set $S\subset\R^n$ we let $\mathrm{conv}(S)$ denote the closed convex hull of $S$, i.e., the closure of the set of all convex combinations of elements in $S$. The polar of a set $S$ is 
\begin{equation}
S^{\circ} = \{y \in \R^n \ : \ \langle x,y\rangle\leq 1 \ \mathrm{for \ all} \ x\in S\}, \EquationName{polar-def}
\end{equation}
which is always closed and convex. A set $\cK$ is called a cone if $\al\cK\subset \cK$ for all $\al>0$. It is called a convex cone if it is, in addition, convex. We use $\mathrm{cone}(S)$ to denote the closed convex cone generated by a set $S$. The polar $\cK^{\circ}$ of a convex cone $\cK$ can be written as
$$\cK^{\circ} = \{y \in \R^n \ : \ \langle x,y\rangle\leq 0 \ \mathrm{for \ all} \ x\in \cK\}$$
and is always a closed convex cone. The bipolar theorem states that $\cK^{\circ\circ}$ is equal to the closure $\overline{\cK}$ of $\cK$ (cf.\ \cite[Theorem 6.11]{MeiseV97}). Let $\cC$ be a closed convex set in $\R^n$. The tangent cone $T_{\cC}(x)$ to $\cC$ at $x \in \cC$ is the closed cone generated by $\cC - \{x\}$, i.e.,
$$T_{\cC}(x) = \mathrm{cone}\{d \in \R^n \ : \ d=y-x, \ y\in \cC\}.$$
Clearly, if $x$ is in the interior of $\cC$, then $T_{\cC}(x)=\R^n$. The normal cone $N_{\cC}(x)$ to $\cC$ at $x\in \cC$ is given by
$$N_{\cC}(x) = \{s\in\R^n \ : \ \langle s,y-x\rangle\leq 0 \ \forall y\in\cC\}.$$
It is easy to see that $(N_{\cC}(x))^{\circ}=T_{\cC}(x)$. Since $T_{\cC}(x)$ is closed, the bipolar theorem implies that $(T_{\cC}(x))^{\circ} = N_{\cC}(x)$.\par
If $f:\R^n\rightarrow \R$ is any function, then a vector $\xi\in \R^n$ is called a subgradient of $f$ at $x\in \mathrm{dom}(f)$ if for any $y\in \mathrm{dom}(f)$
\begin{equation}
\label{eqn:subgradDef}
f(y)\geq f(x) + \langle \xi,y-x\rangle.
\end{equation}
The set of all subgradients of $f$ at $x$ is called the subdifferential of $f$ at $x$ and denoted by $\partial f(x)$. If $f:\R^n\rightarrow \R\cup\{\infty\}$ is a proper convex function (i.e., $f(x)<\infty$ for some $x\in\R^n$), then the descent cone $D(f,x)$ of $f$ at $x\in \R^n$ is defined by
$$D(f,x) = \bigcup_{t>0} \{d \in \R^n \ : \ f(x+td)\leq f(x)\}.$$
The descent cone is always a convex cone, but it may not be closed.
\begin{theorem}
\label{thm:DCPolar}
\cite[Theorem 23.7]{Roc70} Let $f:\R^n\rightarrow \overline{\R}$ be a proper convex function. Suppose that the subdifferential $\partial f(x)$ is nonempty and does not contain $0$. Then,
$$D(f,x)^{\circ} = \mathrm{cone}(\partial f(x)) = \overline{\bigcup_{t\geq 0} \ t\partial f(x)}.$$
\end{theorem}
Let $\|\cdot\|$ be any norm on $\R^n$ and set
$$\cC = \{x \in \R^n \ : \ \|x\|\leq R\}.$$
Then, for any $x\in \cC$ with $\|x\|=R$, $T_{\cC}(x)$ is equal to the descent cone of $\|\cdot\|$ at $x$. By Theorem~\ref{thm:DCPolar} and the bipolar theorem this implies that
\begin{equation}
\label{eqn:tanNormBall}
T_{\cC}(x) = [\mathrm{cone}(\partial\|\cdot\|(x))]^{\circ}.
\end{equation}
Let $\|\cdot\|_*$ denote the dual norm of $\|\cdot\|$, i.e.,
$$\|y\|_* = \sup_{\|x\|\leq 1} \langle x,y\rangle.$$
It is easy to verify from the definition (\ref{eqn:subgradDef}) that
$$
\partial\|\cdot\|(x) =
\begin{cases}
\{y \in \R^n \ : \ \|y\|_*\leq 1\}, & \mathrm{if} \ x=0, \\
\{y\in\R^n \ : \ \|y\|_*\leq 1 \ \mathrm{and} \ \langle y,x\rangle = \|x\|\}, & \mathrm{otherwise}.
\end{cases}
$$
Using these tools, we can readily calculate the tangent cone to an $\ell_{2,1}$-ball.
\begin{example}
\label{exa:TCl21} 
Consider the $\ell_{2,1}$-norm defined in (\ref{eqn:l21NormDef}) and its dual norm, the $\ell_{2,\infty}$-norm. Set 
$$\cC = \{x\in \R^{bD} \ : \ \|x\|_{2,1} \leq R\}.$$ 
Suppose that $x\neq 0$ and let $S\subset[b]$ be the indices corresponding to nonzero blocks of $x$. For $z\in \R^{bD}$, let $z_S$ be the vector 
\begin{equation*}
(z_S)_{B_i} = \begin{cases} z_{B_i} & \mathrm{if} \ i\in S \\
0 & \mathrm{if} \ i\in S^c. \end{cases}  
\end{equation*}
Then,
\begin{align*}
\partial\|\cdot\|_{2,1}(x) & = \{z\in \R^{bD} \ : \ \|z\|_{2,\infty}\leq 1 \ \mathrm{and} \ \langle z,x\rangle = \|x\|_{2,1}\} \\
& = \{z\in \R^{bD} \ : \ \|z_{S^c}\|_{2,\infty}\leq 1 \ \mathrm{and} \ z_{B_i} = x_{B_i}/\|x_{B_i}\|_2 \ \mathrm{for \ all} \ i\in S\}.
\end{align*}
If $x\in \R^{bD}$ satisfies $\|x\|_{2,1}=R$, then by (\ref{eqn:tanNormBall}),
\begin{align*}
T_{\cC}(x) & = \{y\in \R^{bD} \ : \ \langle z,y\rangle\leq 0 \ \mathrm{for \ all} \ z\in \partial\|\cdot\|_{2,1}(x)\} \\
& = \Big\{y\in \R^{bD} \ : \ \sum_{i\in S}\Big\langle \frac{x_{B_i}}{\|x_{B_i}\|_2},y_{B_i}\Big\rangle + \|y_{S^c}\|_{2,1} \leq 0\Big\}.
\end{align*}
In particular, any $y\in T_{\cC}(x)$ satisfies
$$\|y\|_{2,1} = \|y_S\|_{2,1} + \|y_{S^c}\|_{2,1} \leq 2\|y_S\|_{2,1} \leq 2\sqrt{|S|} \ \|y\|_2.$$
\end{example}

\section{Sketching least squares programs with an FJLT}\label{sec:fjlt-cls}

In this appendix we study sketching of least squares programs using a fast Johnson-Lindenstrauss transform (FJLT). We show in Theorem~\ref{thm:FJLTMain} that for $\ell_{2,1}$-constrained least squares minimization, one can achieve the same sketching dimension as in the sparse case (cf.\ Section~\ref{sec:l21Constraint}).\par
We first recall the definition of the FJLT. Let $F$ be the discrete Fourier transform. Let $\theta_1,\ldots,\theta_n:\Om_{\theta}\rightarrow \{0,1\}$ be independent random selectors satisfying $\bP_{\theta}(\theta_i = 1) = m/n$. Let $\si_1,\ldots,\si_n:\Om_{\si}\rightarrow \{-1,1\}$ be independent Rademacher random variables. The FJLT is defined by $\Psi=\Theta FD_{\si}$, where $\Theta = \sqrt{n/m} \ \mathrm{diag}((\theta_i)_{i=1}^n)$ and $D_{\si}=\mathrm{diag}((\si_i)_{i=1}^n)$. Here $\mathrm{diag}((x_i))$ denotes the diagonal matrix with the elements of the sequence $(x_i)$ on its diagonal. To prove Theorem~\ref{thm:FJLTMain} we apply Lemma~\ref{lem:compare} using a suitable upper bound for $Z_2$ and lower bound for $Z_1$. To obtain the latter, we use the following chaining estimate.\par
Let $T$ be some index set. For a given set $(x_{t,i})_{t\in T,1\leq i\leq n}$ in $\R$, we define a semi-metric $\rho_x$ on $T$ by
$$\rho_x(t,s) = \max_{1\leq i\leq n} |x_{t,i} - x_{s,i}|$$
and the denote the associated radius by
$$d_x(T) = \sup_{t\in T} \max_{1\leq i\leq n} |x_{t,i}|.$$
For every $1\leq i\leq n$ we fix a probability space $(\Om_i,\cF_i,\bP_i)$ and let $(\Om,\cF,\bP)$ denote the associated product space. The following observation was proven in the special case $p=1$ in \cite[Theorem 1.2]{GMP07}.
\begin{lemma}
\label{lem:expAtBddLp}
Fix $1\leq p<\infty$. For every $t\in T$ and $1\leq i\leq n$ let $X_{t,i}\in L_{2p}(\Om_i)$. Then,
\begin{align*}
& \Big(\E\sup_{t\in T}\Big|\sum_{i=1}^n X_{t,i}^2-\E X_{t,i}^2\Big|^p\Big)^{1/p} \\
& \qquad \lesssim (\E\ga_2^{2p}(T,\rho_X))^{1/p} + \sup_{t\in T} \Big(\sum_{i=1}^n \E X_{t,i}^2\Big)^{1/2}(\E\ga_2^p(T,\rho_X))^{1/p} \\
& \qquad \qquad \qquad  + \sqrt{p} \sup_{t\in T} \Big(\sum_{i=1}^n \E X_{t,i}^2\Big)^{1/2}(\E d_X^p(T))^{1/p} + p(\E d_X^{2p}(T))^{1/p}.
\end{align*}
\end{lemma}
\begin{proof}
Let $(r_i)_{i\geq 1}$ be a Rademacher sequence. By symmetrization,
\begin{equation}
\label{eqn:AtSymm}
\Big(\E\sup_{t\in T}\Big|\sum_{i=1}^n X_{t,i}^2-\E X_{t,i}^2\Big|^p\Big)^{1/p}
\leq 2 \Big(\E\E_r\sup_{t\in T}\Big|\sum_{i=1}^n r_i X_{t,i}^2\Big|^p\Big)^{1/p}.
\end{equation}
By Hoeffding's inequality, we have for any $s,t \in T$,
$$\bP_{r}\Big(\sum_{i=1}^n r_i (X_{t,i}^2 - X_{s,i}^2) \geq u\Big(\sum_{i=1}^n (X_{t,i}^2 - X_{s,i}^2)^2\Big)^{1/2}\Big)\leq \exp(-u^2/2).$$
Since
$$\Big(\sum_{i=1}^n (X_{t,i}^2 - X_{s,i}^2)^2\Big)^{1/2} \leq \sqrt{2}\sup_{t\in T}\Big(\sum_{i=1}^n X_{t,i}^2\Big)^{1/2}\rho_X(t,s)$$
we conclude that $(\sum_{i=1}^n r_i X_{t,i}^2)_{t\in T}$ is subgaussian with respect to the semi-metric
$$\rho_*(s,t) = \sqrt{2}\sup_{t\in T}\Big(\sum_{i=1}^n X_{t,i}^2\Big)^{1/2}\rho_X(s,t).$$
By Lemma~\ref{lem:supSubg},
\begin{align*}
& \Big(\E_{r}\sup_{t\in T}\Big|\sum_{i=1}^n r_i X_{t,i}^2\Big|^p\Big)^{1/p} \\
& \qquad \lesssim \sup_{t\in T}\Big(\sum_{i=1}^n X_{t,i}^2\Big)^{1/2}\ga_2(T,\rho_X) + \sqrt{p}\sup_{t\in T}\Big(\sum_{i=1}^n X_{t,i}^2\Big)^{1/2}d_X(T) \\
& \qquad \leq \sup_{t\in T}\Big|\sum_{i=1}^n X_{t,i}^2-\E X_{t,i}^2\Big|^{1/2}\Big(\ga_2(T,\rho_X)+\sqrt{p}d_X(T)\Big) \\
& \qquad \qquad \qquad + \sup_{t\in T}\Big(\sum_{i=1}^n \E X_{t,i}^2\Big)^{1/2}\Big(\ga_2(T,\rho_X) + \sqrt{p}d_X(T)\Big).
\end{align*}
Taking $L_p$-norms on both sides, using (\ref{eqn:AtSymm}) and applying H\"{o}lder's inequality yields
\begin{align*}
& \Big(\E\sup_{t\in T}\Big|\sum_{i=1}^n X_{t,i}^2-\E X_{t,i}^2\Big|^p\Big)^{1/p} \\
& \ \ \lesssim \Big(\E\sup_{t\in T}\Big|\sum_{i=1}^n X_{t,i}^2-\E X_{t,i}^2\Big|^{p}\Big)^{1/2p}\Big((\E\ga_2^{2p}(T,\rho_X))^{1/2p} + \sqrt{p}(\E d_X^{2p}(T))^{1/2p}\Big) \\
& \ \ \qquad + \sup_{t\in T}\Big(\sum_{i=1}^n \E X_{t,i}^2\Big)^{1/2}\Big((\E\ga_2^p(T,\rho_X))^{1/p} + \sqrt{p}(\E d_X^p(T))^{1/p}\Big).
\end{align*}
Solving this quadratic inequality yields the result.
\end{proof}
To estimate the $\ga_2$-functional occuring in Lemma~\ref{lem:expAtBddLp} we use a covering number estimate from \cite{GMP08}. Recall the following definitions. Let $E$ be a Banach space and let $E^*$ denote its dual space. The modulus of convexity of $E$ is defined by
$$\del_E(\eps) = \inf\Big\{1-\frac{1}{2}\|x+y\|, \ \|x\|=1, \ \|y\|=1, \ \|x-y\|>\eps\Big\} \qquad (0\leq \eps\leq 2).$$
We say that $E$ is uniformly convex if $\del_{E}(\eps)>0$ for all $\eps>0$ and that $E$ is uniformly convex of power type $2$ with constant $\la$ if $\del_E(\eps)\geq \eps^2/(8\la^2)$ for all $\eps>0$. The following observation is due to Figiel \cite[Proposition 24]{Fig76}.
\begin{lemma}
\label{lem:Figiel}
Suppose that $E$ is a $p$-convex and $q$-concave Banach lattice for some $1<p\leq q<\infty$. Set $r=\max\{2,q\}$ and $K=\max\{2,\frac{2}{\sqrt{p-1}}\}$. Then, $\del_E(\eps)\geq r^{-1}K^{-r}\eps^r$ for all $0\leq \eps\leq 2$.
\end{lemma}
Using the H\"{o}lder-Minkowski inequalities one readily checks that $\ell_{2,p}=\ell_p(\ell_2)$ is $p$-convex and $2$-concave if $1<p\leq 2$. Therefore, $\del_{\ell_{2,p}}(\eps)\geq \frac{1}{8}(p-1)\eps^2$.  
\begin{lemma}
\label{lem:covEstdvGen}
\cite[Lemma 1]{GMP08} Let $E$ be uniformly convex of power type $2$ with constant $\la$. Let $T_2(E^*)$ be the type $2$ constant of $E^*$. Consider $v_1,\ldots,v_N\in E^*$ and define an associated semi-metric on $E$ by
\begin{equation}
\label{eqn:dvDefGen}
\rho_v(x,y) = \max_{1\leq i\leq N} |\langle v_i,x-y\rangle|.
\end{equation}
Set $\nu=\max_{1\leq i\leq N} \|v_i\|_{E^*}$ and let $U\subset B_E$. Then, for all $t>0$,
\begin{equation}
\label{eqn:covEstdvGen}
\log^{1/2}(2\cN(U,\rho_v,t)) \lesssim \nu\la^2 T_2(E^*)\log^{1/2}(N)t^{-1}.
\end{equation}
\end{lemma}
We can now estimate the parameter $Z_1$. 
\begin{lemma}
\label{lem:Z1EstFJLT}
Set $d=bD$. Let $\Psi$ be the FJLT, fix $A\in \R^{n\ti d}$ and let $\cK\subset \R^d$. Consider the norm $\tnorm{A}$ defined in (\ref{eqn:blockNormA}) and set
\begin{equation}
\label{eqn:betaDef}
\beta = (\log^2(\eta^{-1}) + \log(b) + \log(n))\log(n)\log^3(b)\log^2(d).
\end{equation}
If
\begin{equation}
\label{eqn:Z1EstFJLTCond}
m\gtrsim \beta\eps^{-2}\tnorm{A}^2\Big[\sup_{x\in \cK \ : \|Ax\|_2=1}\|x\|_{2,1}^2\Big],
\end{equation}
then with probability at least $1-\eta$ we have
$$(1-\eps)\|x\|_2^2\leq \|\Psi x\|_2^2 \leq (1+\eps)\|x\|_2^2, \qquad \mathrm{for \ all} \ x\in A\cK\cap S^{n-1}.$$
In particular, $Z_1(A,\Psi,\cK)\geq 1-\eps$.
\end{lemma}
\begin{proof}
Let $F_i$ denote the $i$-th row of $F$. Since 
$$\E_{\theta}\|\Psi x\|_2^2 = \|D_{\si}x\|_2^2 = \|x\|_2^2, \qquad \mathrm{for \ all} \ x\in \R^n,$$
we have
\begin{align*}
& \sup_{x\in A\cK\cap S^{n-1}} \Big|\|\Psi x\|_2^2 - \|x\|_2^2\Big| \\
& \qquad = \sup_{x\in A\cK\cap S^{n-1}} \Big|\sum_{i=1}^n \theta_i\Big\langle \sqrt{\frac{n}{m}}D_{\si}F_i,x\Big\rangle^2 - \E_{\theta}\theta_i\Big\langle \sqrt{\frac{n}{m}}D_{\si}F_i,x\Big\rangle^2\Big| \\
& \qquad = \sup_{x\in \cK\cap A^{-1}(S^{n-1})} \Big|\sum_{i=1}^n \theta_i\Big\langle \sqrt{\frac{n}{m}}A^*D_{\si}F_i,x\Big\rangle^2 - \E_{\theta}\theta_i\Big\langle \sqrt{\frac{n}{m}}A^*D_{\si}F_i,x\Big\rangle^2\Big|.
\end{align*}
We apply Lemma~\ref{lem:expAtBddLp} (with $X_{x,i}:=\langle\theta_i\sqrt{\frac{n}{m}}A^*D_{\si}F_i,x\rangle$) to find
\begin{align*}
& \Big(\E_{\theta}\sup_{x\in A\cK\cap S^{n-1}} \Big|\|\Psi x\|_2^2 - \|x\|_2^2\Big|^p\Big)^{1/p} \\
& \qquad \leq (\E_{\theta}\ga_2^{2p}(\cK\cap A^{-1}(S^{n-1}),\rho_X))^{1/p} + (\E_{\theta}\ga_2^p(\cK\cap A^{-1}(S^{n-1}),\rho_X))^{1/p} \\
& \qquad \qquad + \sqrt{p} (\E_{\theta} d_X^p(\cK\cap A^{-1}(S^{n-1}))^{1/p} + p(\E_{\theta} d_X^{2p}(\cK\cap A^{-1}(S^{n-1}))^{1/p} \\
& \qquad \leq \ga_2^2(\cK\cap A^{-1}(S^{n-1}),\rho_v) + \ga_2(\cK\cap A^{-1}(S^{n-1}),\rho_v) \\
& \qquad \qquad + \sqrt{p}d_v(\cK\cap A^{-1}(S^{n-1})) + p d_v^2(\cK\cap A^{-1}(S^{n-1})),
\end{align*}
where we have set $v_i = \sqrt{\frac{n}{m}}A^*D_{\si}F_i$, defined $\rho_v$ as in (\ref{eqn:dvDefGen}) and used that $\rho_X\leq \rho_v$ uniformly. Set $d_{2,1} = d_{\|\cdot\|_{2,1}}(\cK\cap A^{-1}(S^{n-1}))$ and $\nu := \max_{1\leq i\leq n} \|v_i\|_{2,\infty}$.
Clearly
\begin{equation}
\label{eqn:DelOestdv}
d_v(\cK\cap A^{-1}(S^{n-1})) \leq \DelO_{2,1}\nu. 
\end{equation}
We estimate the $\ga_2$-functional by an entropy integral
\begin{align*}
& \ga_2(\cK\cap A^{-1}(S^{n-1}),\rho_v) \\
& \qquad \lesssim \int_0^{t^*}\cN(\cK\cap A^{-1}(S^{n-1}),\rho_v,t) + \int_{t_*}^{\infty} \cN(\cK\cap A^{-1}(S^{n-1}),\rho_v,t) \ dt.
\end{align*}
The first integral we estimate using the volumetric bound
$$\cN(\cK\cap A^{-1}(S^{n-1}),\rho_v,t) \leq \cN(\cK\cap A^{-1}(S^{n-1}),\nu\|\cdot\|_{2,1},t) \leq \Big(1+\frac{2\nu\DelO_{2,1}}{t}\Big)^d.$$
For the second integral we set $p=\log(b)/(\log(b)-1)$ and apply Lemma~\ref{lem:covEstdvGen} with $E=\ell_{2,p}^d$ and $E^*=\ell_{2,p'}^d$, where $p'=\log(b)$. Note that $\|\cdot\|_{2,\infty} \leq \|\cdot\|_{2,p'}\leq e\|\cdot\|_{2,\infty}$ and therefore $T_2(E^*)\leq e\log^{1/2}(b)$. Also, by the remark after Lemma~\ref{lem:Figiel} we have $\la^2=(p-1)^{-1}=\log(b)-1$. By (\ref{eqn:covEstdvGen}),
\begin{equation}
\label{eqn:covEstdv21}
\log^{1/2}(\cN(B_{\ell_{2,1}^d},\rho_v,t)) \leq \log^{1/2}(\cN(B_{\ell_{2,p}^d},\rho_v,t)) \lesssim \nu\log^{3/2}(b)\log^{1/2}(n)t^{-1}.
\end{equation}
Since $\cK\cap A^{-1}(S^{n-1})\subset\DelO_{2,1} B_{\ell_{2,1}^d}$, we arrive at
\begin{align*}
& \ga_2(\cK\cap A^{-1}(S^{n-1}),\rho_v) \\
& \qquad \lesssim \sqrt{d}\int_0^{t_*} \log^{1/2}\Big(1+\frac{2\nu\DelO_{2,1}}{t}\Big) \ dt + \int_{t_*}^{\nu\DelO_{2,1}}\nu\DelO_{2,1} \log^{1/2}(n)\log^{3/2}(b)t^{-1} \ dt \\
& \qquad \leq \sqrt{d} t_*\log^{1/2}(e + 2et_*^{-1}\nu\DelO_{2,1}) + \nu\DelO_{2,1}\log^{1/2}(n)\log^{3/2}(b)\log(t_*^{-1}\nu\DelO_{2,1}).
\end{align*}
Take $t_*=d^{-1/2}\nu\DelO_{2,1}$ to obtain
\begin{align}
\label{eqn:ga2dvEst}
& \ga_2(\cK\cap A^{-1}(S^{n-1}),\rho_v) \nonumber\\
& \qquad \lesssim \nu\DelO_{2,1}\log^{1/2}(e + 2e\sqrt{d}) + \nu\DelO_{2,1}\log^{1/2}(n)\log^{3/2}(b)\log(\sqrt{d}) \nonumber\\
& \qquad \lesssim \nu\DelO_{2,1}\log^{1/2}(n)\log^{3/2}(b)\log(d).
\end{align}
In conclusion,
\begin{align}
\label{eqn:ga2dvEstLp}
& \Big(\E_{\theta}\sup_{x\in A\cK\cap S^{n-1}} \Big|\|\Psi x\|_2^2 - \|x\|_2^2\Big|^p\Big)^{1/p} \nonumber\\
& \qquad \lesssim \nu^2\DelO_{2,1}^2\log(n)\log^3(b)\log^2(d)+ \nu\DelO_{2,1}\log^{1/2}(n)\log^{3/2}(b)\log(d) \nonumber \\
& \qquad \qquad \qquad + \sqrt{p} \nu\DelO_{2,1} + p\nu^2\DelO_{2,1}^2.
\end{align}
Since $|\sqrt{n}F_{ij}|\leq 1$, Khintchine's inequality implies that
\begin{align}
\label{eqn:estNu}
& (\E_{\si}\nu^p)^{1/p} \nonumber\\ 
& \ = \frac{1}{\sqrt{m}}(\E_{\si}\max_{1\leq i\leq n}\sqrt{n}\|A^*D_{\si}F_i\|_{2,\infty}^p)^{1/p} \nonumber\\
& \ = \frac{1}{\sqrt{m}}\Big(\E_{\si}\max_{1\leq i\leq n}\max_{1\leq \ell\leq b}\Big(\sum_{k\in B_{\ell}}\Big|\sum_{j=1}^n \si_j A_{jk}\sqrt{n}F_{ij}\Big|^2\Big)^{p/2}\Big)^{1/p} \nonumber\\
& \ \leq \frac{1}{\sqrt{m}}(\sqrt{p}+\log^{1/2}(b)+\log^{1/2}(n))\max_{1\leq i\leq n}\max_{1\leq \ell\leq b}\Big(\sum_{k\in B_{\ell}}\sum_{j=1}^n|A_{jk}\sqrt{n}F_{ij}|^2\Big)^{1/2} \nonumber\\
& \ \leq \frac{1}{\sqrt{m}}(\sqrt{p}+\log^{1/2}(b)+\log^{1/2}(n))\max_{1\leq \ell\leq b}\Big(\sum_{k\in B_{\ell}}\sum_{j=1}^n|A_{jk}|^2\Big)^{1/2}.
\end{align}
Taking $L_p(\Om_{\si})$-norms in (\ref{eqn:ga2dvEstLp}) and using (\ref{eqn:estNu}), we conclude that
\begin{align*}
& \Big(\E_{\theta,\si}\sup_{x\in A\cK\cap S^{n-1}} \Big|\|\Psi x\|_2^2 - \|x\|_2^2\Big|^p\Big)^{1/p} \\
& \qquad \lesssim \frac{1}{m}(\sqrt{p}+\log^{1/2}(b)+\log^{1/2}(n))^2\tnorm{A}^2\DelO_{2,1}^2\log(n)\log^3(b)\log^2(d)\\
& \qquad \qquad + \frac{1}{\sqrt{m}} (\sqrt{p}+\log^{1/2}(b)+\log^{1/2}(n))\tnorm{A}\DelO_{2,1}\log^{1/2}(n)\log^{3/2}(b)\log(d)  \\
& \qquad \qquad + \sqrt{p} \frac{1}{\sqrt{m}}(\sqrt{p}+\log^{1/2}(b)+\log^{1/2}(n))\tnorm{A}\DelO_{2,1} \\
& \qquad \qquad + p\frac{1}{m}(\sqrt{p}+\log^{1/2}(b)+\log^{1/2}(n))^2\tnorm{A}^2\DelO_{2,1}^2.
\end{align*}
The result now follows from Lemma~\ref{lem:MomentsToTails} and taking $w=\log(\eta^{-1})$ in (\ref{eqn:MomentsToTails}).
\end{proof}
To prove an upper bound for $Z_2$ we use the following variation of Lemma~\ref{lem:expAtBddLp}. Note that the element $u$ below does not need to be in the index set $T$.
\begin{lemma}
\label{lem:expProdBddLp}
For every $t\in T$ and $1\leq i\leq n$ let $X_{t,i}\in L_{2p}(\Om_i)$. Fix also $X_{u,i} \in L_p(\Om_i)$. For any $1\leq p<\infty$,
\begin{align}
\label{eqn:expProdBddLp}
& \Big(\E\sup_{t\in T}\Big|\sum_{i=1}^n X_{t,i}X_{u,i} - \E(X_{t,i}X_{u,i})\Big|^p\Big)^{1/p} \nonumber \\
& \qquad \lesssim \Big(\sqrt{p} \Big(\E\max_{1\leq i\leq n}|X_{u,i}|^{2p}\Big)^{1/(2p)} + \Big(\sum_{i=1}^n \E X_{u,i}^2\Big)^{1/2}\Big) \nonumber \\
& \qquad \qquad \qquad \qquad * \Big((\E\ga_2^{2p}(T,\rho_X))^{1/(2p)} + \sqrt{p} (\E d_X^{2p}(T))^{1/(2p)}\Big).
\end{align}
\end{lemma}
\begin{proof}
Let $(r_i)$ be a Rademacher sequence. By Hoeffding's inequality, we have for any $s,t \in T$ and $w\geq 0$,
$$\bP_{r}\Big(\sum_{i=1}^n r_i (X_{t,i}X_{u,i} - X_{s,i}X_{u,i}) \geq w\Big(\sum_{i=1}^n (X_{t,i}X_{u,i} - X_{s,i}X_{u,i})^2\Big)^{1/2}\Big)\leq e^{-w^2/2}.$$
Since
$$\Big(\sum_{i=1}^n (X_{t,i}X_{u,i} - X_{s,i}X_{u,i})^2\Big)^{1/2} \leq \Big(\sum_{i=1}^n X_{u,i}^2\Big)^{1/2}\rho_X(t,s),$$
we conclude that $(\sum_{i=1}^n r_i X_{t,i}X_{u,i})_{t\in T}$ is subgaussian with respect to the semi-metric
$$\rho_*(s,t) = \Big(\sum_{i=1}^n X_{u,i}^2\Big)^{1/2}\rho_X(s,t).$$
By Lemma~\ref{lem:supSubg},
$$\Big(\E_r\sup_{t\in T}\Big|\sum_{i=1}^n r_iX_{t,i}X_{u,i}\Big|^p\Big)^{1/p} \lesssim \Big(\sum_{i=1}^n X_{u,i}^2\Big)^{1/2} \ga_2(T,\rho_X) + \sqrt{p} d_X(T).$$
Using symmetrization \Eqsub{symmetrization}, this implies that
\begin{align*}
& \Big(\E\sup_{t\in T}\Big|\sum_{i=1}^n X_{t,i}X_{u,i} - \E(X_{t,i}X_{u,i})\Big|^p\Big)^{1/p} \\
& \qquad \leq 2 \Big(\E\E_r\sup_{t\in T}\Big|\sum_{i=1}^n r_iX_{t,i}X_{u,i}\Big|^p\Big)^{1/p} \\
& \qquad \lesssim \Big(\E\Big(\sum_{i=1}^n X_{u,i}^2\Big)^{p}\Big)^{1/(2p)}\Big((\E\ga_2^{2p}(T,\rho_X))^{1/(2p)} + \sqrt{p} (\E d_X^{2p}(T))^{1/(2p)}\Big).
\end{align*}
By symmetrization and Khintchine's inequality,
\begin{align*}
& \Big(\E\Big(\sum_{i=1}^n X_{u,i}^2\Big)^{p}\Big)^{1/p} \\
& \qquad \leq \Big(\E\Big|\sum_{i=1}^n X_{u,i}^2 - \E X_{u,i}^2\Big|^{p}\Big)^{1/p} + \sum_{i=1}^n \E X_{u,i}^2 \\
& \qquad \leq 2 \Big(\E\E_r\Big|\sum_{i=1}^n r_i X_{u,i}^2\Big|^{p}\Big)^{1/p} + \sum_{i=1}^n \E X_{u,i}^2 \\
& \qquad \leq 2\sqrt{p} \Big(\E\Big|\sum_{i=1}^n X_{u,i}^4\Big|^{p/2}\Big)^{1/p} + \sum_{i=1}^n \E X_{u,i}^2 \\
& \qquad \leq 2\sqrt{p} \Big(\E\max_{1\leq i\leq n}|X_{u,i}|^{2p}\Big)^{1/(2p)}\Big(\E\Big(\sum_{i=1}^n X_{u,i}^2\Big)^{p}\Big)^{1/(2p)} + \sum_{i=1}^n \E X_{u,i}^2. 
\end{align*}
Solving this quadratic inequality yields
$$\Big(\E\Big(\sum_{i=1}^n X_{u,i}^2\Big)^{p}\Big)^{1/(2p)} \leq 2\sqrt{p} \Big(\E\max_{1\leq i\leq n}|X_{u,i}|^{2p}\Big)^{1/(2p)} + \Big(\sum_{i=1}^n \E X_{u,i}^2\Big)^{1/2}.$$
\end{proof}
\begin{lemma}
\label{lem:Z2EstFJLT}
Let $\Psi$ be the FJLT, let $A\in \R^{n\ti d}$, $\cK\subset\R^d$ and $u\in S^{n-1}$. Let $\tnorm{A}$ be as in (\ref{eqn:blockNormA}) and $\beta$ as in (\ref{eqn:betaDef}). If
\begin{align*}
m & \gtrsim \beta\eps^{-2}\tnorm{A}^2\Big[\sup_{x\in \cK \ : \|Ax\|_2=1}\|x\|_{2,1}^2\Big],
\end{align*}
then $Z_2(A,\Psi,\cK,u)\leq \eps$ with probability at least $1-\eta$.
\end{lemma}
\begin{proof}
If $\Psi_i$ denotes the $i$-th row of $\Psi$, then we can write
$$Z_2(A,\Psi,\cK,u) = \sup_{x\in \cK\cap A^{-1}(S^{n-1})}\Big|\sum_{i=1}^n \langle A^*\Psi_i,x\rangle \langle \Psi_i,u\rangle - \E(\langle A^*\Psi_i,x\rangle \langle \Psi_i,u\rangle)\Big|.$$
Set $v_i = \sqrt{\frac{n}{m}}A^*D_{\si}F_i$ and let $\rho_v$ be the semi-metric in (\ref{eqn:dvDefGen}). We apply Lemma~\ref{lem:expProdBddLp} with $X_{x,i}:=\langle\theta_i\sqrt{\frac{n}{m}}A^*D_{\si}F_i,x\rangle$ and $X_{u,i}:=\langle\theta_i\sqrt{\frac{n}{m}}D_{\si}F_i,u\rangle$. By (\ref{eqn:DelOestdv}) and (\ref{eqn:ga2dvEst}) we know that 
\begin{align*}
(\E_{\theta} d_X^{2p}(T))^{1/(2p)} & \leq \DelO_{2,1}\nu \\
(\E_{\theta}\ga_2^{2p}(T,\rho_X))^{1/(2p)} & \lesssim \DelO_{2,1}\nu\log^{1/2}(n)\log^{3/2}(b)\log(d).
\end{align*}
Moreover, 
$$\sum_{i=1}^n \E_{\theta} X_{u,i}^2 = \|D_{\si} u\|_2^2=1, \qquad (\E_{\theta}\max_{1\leq i\leq n}|X_{u,i}|^{2p})^{1/(2p)} \leq \max_{1\leq i\leq n} \sqrt{\frac{n}{m}} |\langle F_i,D_{\si}u\rangle|.$$
Applying these estimates in (\ref{eqn:expProdBddLp}) and taking $L_p(\Om_{\si})$-norms yields
\begin{align*}
(\E_{\theta,\si}Z_2^p)^{1/p} & \lesssim \Big(\sqrt{\frac{p}{m}}\Big(\E_{\si}\max_{1\leq i\leq n}\sqrt{n}|\langle F_i,D_{\si}u\rangle|^{2p}\Big)^{1/(2p)} + 1\Big) \\
& \qquad \qquad * (\E_{\si}\nu^{2p})^{1/(2p)}\DelO_{2,1}\Big(\log^{1/2}(n)\log^{3/2}(b)\log(d)+\sqrt{p}\Big).
\end{align*}
By Khintchine's inequality,
\begin{align*}
\Big(\E_{\si}\max_{1\leq i\leq n}\sqrt{n}|\langle F_i,D_{\si}u\rangle|^{2p}\Big)^{1/(2p)} & \lesssim (\sqrt{p} + \log^{1/2}(n)) \max_{1\leq i\leq n} \Big(\sum_{j=1}^n n F_{ij}^2 u_j^2\Big)^{1/2} \\
& \lesssim \sqrt{p} + \log^{1/2}(n).
\end{align*}
and by (\ref{eqn:estNu})
$$(\E\nu^{2p})^{1/(2p)} \lesssim \frac{1}{\sqrt{m}}(\sqrt{p}+\log^{1/2}(b)+\log^{1/2}(n))\tnorm{A}.$$
Combining these estimates and using Lemma~\ref{lem:MomentsToTails} yields the result.
\end{proof}
Combining Lemmas~\ref{lem:compare}, \ref{lem:Z1EstFJLT}, and \ref{lem:Z2EstFJLT} yields the following result.
\begin{theorem}
\label{thm:FJLTMain}
Set $d=bD$. Let $\Psi$ be the FJLT, $A\in \R^{n\ti d}$ and let $\cC$ be a closed convex set in $\R^d$. Set $\beta$ as in (\ref{eqn:betaDef}). Let $x_*$ and $\hat{x}$ be the minimizers of (\ref{eqn:CLS}) and (\ref{eqn:SCLS}), respectively. If
\begin{align*}
m & \gtrsim \beta\eps^{-2}\tnorm{A}^2\Big[\sup_{x\in T_{\cC}(x_*) \ : \ \|Ax\|_2=1}\|x\|_{2,1}^2\Big],
\end{align*}
then, with probability at least $1-\eta$,
$$f(\hat{x})\leq (1-\eps)^{-2} f(x_*).$$
\end{theorem}
The proof of Corollary~\ref{cor:SJLl21} immediately yields the following consequence.
\begin{corollary}
\label{cor:FJLTMainCor}
Set $d=bD$. Let $\Psi$ be the FJLT, $A\in \R^{n\ti d}$ and let $\cC=\{x\in \R^d \ : \ \|x\|_{2,1}\leq R\}$. Define
$$\si_{\min,k} = \inf_{\|y\|_2=1, \ \|y\|_{2,1}\leq 2\sqrt{k}}\|Ay\|_2.$$
Suppose that $x_*$ is $k$-block sparse and $\|x_*\|_{2,1}=R$. If
\begin{align*}
m & \gtrsim \beta \eps^{-2}\tnorm{A}^2 k \si_{\min,k}^{-2},
\end{align*}
then, with probability at least $1-\eta$,
$$f(\hat{x})\leq (1-\eps)^{-2}f(x_*).$$
\end{corollary}
Observe that the condition on $m$ in Theorem~\ref{thm:FJLTMain} and Corollary~\ref{cor:FJLTMainCor} is, up to different log-factors, the same as the condition for the SJLT in Theorem~\ref{thm:SJLmain} and Corollary~\ref{cor:SJLl21}.\par
In the special case $D=1$, which corresponds to the Lasso, the result in Corollary~\ref{cor:FJLTMainCor} gives a qualitative improvement over \cite[Corollary 3]{PiW14}. Recall from the discussion following Corollary~\ref{cor:SJLl21} that in this case
$$\tnorm{A} = \max_{1\leq k\leq d}\|A_k\|_2, \qquad \si_{\min,k} = \inf_{\|y\|_2=1, \ \|y\|_1\leq 2\sqrt{k}}\|Ay\|_2.$$ 
In \cite[Corollary 3]{PiW14} the condition
$$m\gtrsim \eps^{-2}\log(\eta^{-1}) + \eps^{-2}\min\Big\{\log^2(d)\Big(\tnorm{A}^2 k \si_{\min,k}^{-2}\Big)^2, k\log(d)\log^4(n)\frac{\si_{\max,k}^4}{\si_{\min,k}^4}\Big\},$$
was obtained, where $\si_{\max,k} = \sup_{\|y\|_2=1, \ \|y\|_{1}\leq 2\sqrt{k}}\|Ay\|_2$. In terms of the dependence on $A$ this bound is worse than our result. We note, however, that the bound contains fewer log-factors and in particular the dependence on $\eta$ is better. 

\end{document}